\documentclass[11pt,letter]{article}
\usepackage[left,mathlines]{lineno} 
\usepackage[top=1in, bottom=1in, left=0.92in,right=0.92in]{geometry}
\usepackage{graphicx} 
\usepackage{mathtools}
\usepackage{amssymb}
\usepackage{sectsty}
\usepackage{titlesec}
\usepackage{subfigure}
\usepackage{color}
\usepackage[table]{xcolor} 
\usepackage{tikz}          

\usepackage{amsmath}
\usepackage{amsthm} 
\usepackage{verbatim}
\usepackage[english]{babel}
\usepackage{xspace}
\usepackage{amsfonts}
\usepackage{mathrsfs}
\usepackage{bbm}
\usepackage[all]{xy}
\usepackage{mathtools}
\usepackage{tabularx}
\usepackage{array}
\usepackage{enumitem}
\usepackage{listings}
\usepackage{url}
\usepackage{tikz}
\usepackage[backref]{hyperref} 
\usepackage{cleveref}
\usepackage[ruled,linesnumbered]{algorithm2e}
\usepackage[toc,page]{appendix}
\usepackage{cleveref} 
\usepackage{turnstile}
\usepackage[normalem]{ulem} 
\usepackage{appendix}
\usepackage[sanserif,full]{complexity}
\usepackage[inline,final]{showlabels} 
\usepackage{draftwatermark}

 \SetWatermarkText{}
\SetWatermarkScale{7}     
\SetWatermarkColor[gray]{0.94}  

\definecolor{airforceblue}{rgb}{0.36, 0.54, 0.76}
\definecolor{darkcerulean}{rgb}{0.03, 0.27, 0.49}
\definecolor{darkelectricblue}{rgb}{0.33, 0.41, 0.77}
\hypersetup{colorlinks = true,
        linkcolor =  airforceblue,
        urlcolor = blue,
        citecolor = darkelectricblue,
        pdfstartview=FitH}

\newtheorem{theorem}{Theorem}[section]
\newtheorem*{theorem*}{Theorem}
\newtheorem{lemma}[theorem]{Lemma}
\newtheorem{claim}[theorem]{Claim}
\newtheorem*{claim*}{Claim}
\newtheorem{definition}[theorem]{Definition}
\newtheorem{proposition}[theorem]{Proposition}
\newtheorem*{proposition*}{Proposition}
\newtheorem{corollary}[theorem]{Corollary}
\newtheorem*{corollary*}{Corollary}
\newtheorem{remark}[theorem]{Remark}
\newtheorem*{remark*}{Remark}

\newcommand{\para}{\paragraph}

\theoremstyle{remark} 

\makeatletter
\DeclareRobustCommand\onedot{\futurelet\@let@token\@onedot}
\def\@onedot{\ifx\@let@token.\else.\null\fi\xspace}

\def\ie{\emph{i.e}\onedot}

 \def\resp{\emph{resp}\onedot}

\makeatother

\newcommand{\probP}{\text{I\kern-0.15em P}}

\newcommand{\RankP}{\mathrm{RankP}} 
\newcommand{\TRankP}{\mathrm{TRankP}} 
\newcommand{\IRankP}{\mathrm{IRankP}} 
\newcommand{\IRankPE}{\mathrm{IRankPE}} 
\newfunc{\depth}{Depth}
\newfunc{\sdeg}{sdeg}
\newcommand{\sgn}{{\rm sgn}}
\newcommand{\perm}{{\rm perm}}
\newcommand{\IMM}{\ComplexityFont{IMM}}
\newfunc{\coeff}{coeff}
\newfunc{\Size}{Size}
\newcommand{\alg}{\ComplexityFont{alg}}

\newfunc{\ac}{ac}
\newcommand{\AUB}{\mathrm{AUB}} 
\newfunc{\refute}{ref}
\newfunc{\cnf}{cnf}
\newfunc{\ecnf}{ecnf}

\newcommand{\IPS}{%
    {
        \ComplexityFont{IPS} 
    }
}
\newcommand{\CNF}{ %
    {
        \ComplexityFont{CNF}
    }
}
\newcommand{\VAL}{ %
    {
        \ComplexityFont{VAL}
    }
}
\newcommand{\arit}{ %
    {
        \ComplexityFont{arit}
    }
}
\newcommand{\carry}{ %
    {
        \ComplexityFont{CARRY}
    }
}
\newcommand{\add}{ %
    {
        \ComplexityFont{ADD}
    }
}
\newcommand{\addition}{ %
    {
        \ComplexityFont{Addition}
    }
}
\newcommand{\modular}{ %
    {
        \ComplexityFont{Modular}
    }
}
\newcommand{\mult}{ %
    {
        \ComplexityFont{MULT}
    }
}

\newcommand{\SLP}{ %
    {
        \ComplexityFont{SLP}
    }
}

\newcommand{\PC}{ %
    {
        \ComplexityFont{PC}
    }
}
\newcommand{\PCR}{ %
    {
        \ComplexityFont{PCR}
    }
}

\newcommand{\SCNF}{%
    {
        \ComplexityFont{SCNF}
    }
}
\newcommand{\UBIT}{%
    {
        \ComplexityFont{UBIT}
    }
}

\newcommand{\x}{\mathbf{x}}

\newcommand{\F}{\mathbb{F}}

\newcommand{\scnf}{
    {
        \ComplexityFont{scnf}
    }
}
\newcommand{\VAC}{
    {
        \ComplexityFont{VAC}
    }
}
\newcommand{\calM}{\mathcal{M}}

\newcommand{\ACZ}{\ensuremath{\ComplexityFont{AC}^0}}
\newcommand{\ACZP}{\ensuremath{\ACZ[p]}}

\newcommand{\NCTwo}{\ensuremath{\ComplexityFont{NC^2}}}

\newcommand{\commentout}[1]{}

\let\etalchar\relax 

\title{\bf \ACZ$[p]$-Frege Cannot Efficiently Prove that Constant-Depth Algebraic Circuit Lower Bounds are Hard}

\author{
 Jiaqi Lu\thanks{Department of Computing. This project has received funding from the European Research Council (ERC) under the European Union’s Horizon 2020 research and innovation programme (grant agreement No 101002742, \textit{EPRICOT} project).}
 \smallskip
 \\
 \small Imperial College London
 \and 
 Rahul Santhanam%
 \thanks{This project was partly funded by the EPSRC grant EP/Z534158/1, \textit{Integrated Approach to Computational Complexity: Structure, Self-Reference and Lower Bounds}.}
 \smallskip
 \\
 \small University of Oxford
 \and
 Iddo Tzameret\thanks{Department of Computing. This project has received funding from the European Research Council (ERC) under the European Union’s Horizon 2020 research and innovation programme (grant agreement No 101002742, \mbox{\textit{EPRICOT}} project). It was also supported by the Engineering and Physical Sciences Research Council (EPSRC) under grant EP/Z534158/1, \textit{Integrated Approach to Computational Complexity: Structure, Self-Reference and Lower Bounds}.. Email: \text{iddo.tzameret@gmail.com}}
 \\ 
  \small Imperial College London
 }

\newcommand{\iddo}[1]
 {%
    [%
        {   
                {%
            \textcolor{green}%
                    {\footnotesize{Iddo: #1}}%
                }%
        }%
    ]%
}

\newcommand{\Rahul}[1]%
{%
    [%
        {   
                {%
            \textcolor{green}%
                    {\footnotesize{Rahul: #1}}%
                }%
        }%
    ]%
}

\newcommand{\Jiaqi}[1]
{%
    [%
        {   
                {%
            \textcolor{green}%
                    {\footnotesize{Jiaqi: #1}}%
                }%
        }%
    ]%
}
\newcommand{\new}[1]{}

\begin{document}

\maketitle

\thispagestyle{empty} 
  
\begin{abstract}
We study whether lower bounds against constant-depth algebraic circuits computing the Permanent over finite fields (Limaye--Srinivasan--Tavenas [{\footnotesize J.~ACM, 2025}] and Forbes [{\footnotesize CCC'24}]) are hard to prove in certain proof systems.
We focus on a DNF formula that expresses that such lower bounds are hard for constant-depth algebraic proofs. 
Using an adaptation of the diagonalization framework of Santhanam and Tzameret ({\footnotesize SIAM~J.~Comput., 2025}), we show \emph{unconditionally} that this family of DNF formulas does not admit polynomial-size propositional $\ACZ[p]$-Frege proofs, infinitely often.  
This rules out the possibility that the DNF family is easy, and establishes that its status is either that of a hard tautology for $\ACZ[p]$-Frege or else unprovable (i.e., not a tautology).  
While it remains open whether the DNFs in question are tautologies, we provide evidence in this direction.  
In particular, under the plausible assumption that certain (weak) properties of multilinear algebra---specifically, those involving tensor rank---do not admit short constant-depth algebraic proofs, the DNFs \emph{are} tautologies.  
We also observe that several weaker variants of the DNF formula are provably tautologies, and we show that the question of whether the DNFs are tautologies connects to conjectures of Razborov ({\footnotesize ICALP'96}) and Kraj\'{\i}\v{c}ek ({\footnotesize J.~Symb.~Log., 2004}).

\smallskip 
Additionally, our result has the following special features: 

(\textbf{i}) 
\textbf{Existential depth amplification}: the DNF formula considered is parameterised by a constant depth $d$ bounding the depth of the algebraic proofs. We show that there \emph{exists some fixed} depth $d$ such that if there are no small depth-$d$ algebraic proofs of certain circuit lower bounds for the Permanent, then there are no such small algebraic proofs in {\it any} constant depth. 

(\textbf{ii}) 
\textbf{Necessity}: 
We show that our result is a necessary step towards establishing lower bounds against constant-depth algebraic proofs, and more generally against any sufficiently strong proof system. 
In particular, showing there are no short proofs for our  DNF formulas, obtained by replacing `constant-depth algebraic circuits' with any ``reasonable" algebraic circuit class $\mathcal C$, is {\it necessary} in order to prove any super-polynomial lower bounds against algebraic proofs operating with circuits from $\mathcal C$.

    \end{abstract}

\newpage
\clearpage
\pagenumbering{arabic}
    
    \section{Introduction}

    Propositional proof complexity studies the sizes of proofs for propositional tautologies in {\it propositional proof systems} of interest~\cite{CR79, BP98, krajivcek2019proof}. In general, a propositional proof system $Q$ is defined by a polynomial-time computable binary relation
$R_Q$ such that a formula $\phi$ is a propositional tautology if and only if there is some $y$ such that $R_Q(\phi,y)$ holds; any such $y$ is called a {\it proof} of $\phi$. The $R$-proof size of $\phi$ is the size of the smallest $y$ such that $R_Q(\phi,y)$ holds. For natural sequences of tautologies (such as the Pigeonhole Principle, Tseitin graph formulas, and formulas encoding circuit lower bounds) $\phi_n$ and propositional proof systems of interest (such as Resolution, Frege and Extended Frege), we would like to understand how the proof size of $\phi_n$ grows with the size of the formula  $|\phi_n|$, and in particular whether the proof size is polynomially bounded or not as a function of the size of the formula.
    

     In their seminal paper on propositional proof complexity, Cook and Reckhow \cite{CR79} observed that $\mathsf{NP} = \mathsf{coNP}$ if and only if there is a propositional proof system in which every sequence of tautologies has polynomially bounded proof size. This is the basis of the ``\emph{Cook-Reckhow program}'' \cite{BP98}, of which the ideal limit is the separation of  $\mathsf{NP}$ from  $\mathsf{coNP}$ (and hence also $\mathsf{P}$ from  $\mathsf{NP}$) by showing super-polynomial proof size lower bounds for progressively stronger propositional proof systems. 

    The Cook-Reckhow program can be seen as a dual, nondeterministic, analogue of the {\it circuit complexity approach} to $\mathsf{P}$ vs $\mathsf{NP}$, which proceeds by showing super-polynomial circuit size lower bounds for progressively stronger circuit classes. This analogy is further strengthened by the fact that there is a close correspondence between Boolean circuit classes and inference-based propositional logic, i.e., Frege-style proof systems~\cite{BP98}, in which new proofs are derived from axioms and previously derived proof lines using simple sound derivation rules. The basic  Resolution proof system works with proof lines that are clauses; the Frege proof system with lines that are Boolean formulas; and the Extended Frege (EF) proof system with lines that are essentially Boolean circuits. The strength of these Frege-style proof systems is generally believed to grow as the proof lines get more expressive, just as the corresponding circuit classes are believed to increase in computational power as they grow more expressive.

\subsection{Hard Formulas}

We say that a sequence of formulas $\phi_n$ is \emph{hard} if there is no proof of $\phi_n$ of polynomial size in $|\phi_n|$. One of the earliest steps in the Cook-Reckhow program---establishing a lower bound on the size of proofs---was taken by Haken \cite{Hak85}. He showed a super-polynomial lower bound for the Pigeonhole Principle formulas in Resolution. Soon after, Ajtai \cite{Ajt88} showed a super-polynomial lower bound for the Pigeonhole Principle in $\AC^0$-Frege, which is the Frege-style system where proof lines are constant-depth Boolean circuits. Since then,  several improvements and extensions of Ajtai's result have been obtained  \cite{Ajt94, BeameIKPP96, BussIKPRS96}, and the lower bound techniques used in these works mirror the random restriction techniques used to prove lower bounds against the circuit class $\AC^0$, further reinforcing the analogy between proof complexity lower bounds and circuit complexity lower bounds (cf.~\cite{PBI93,KPW95}). In the circuit complexity setting, we also know lower bounds for the circuit class $\AC^0[p]$ of constant-depth Boolean circuits with prime modular gates, shown using the polynomial method of Razborov and Smolensky  \cite{Razborov87-eng,Smolensky87}. Can the polynomial method or other techniques used to show proof complexity lower bounds for the corresponding Frege-style proof system $\AC^0[p]$-Frege?

    Despite much effort, this question has remained open for more than three decades, and continues to be a frontier question in proof complexity. It is already highlighted as a key open problem in the 1998 survey of Beame and Pitassi \cite{BP98}, and the recent survey of Razborov re-iterates this \cite{Razb16}. Progress on the question has focused on restricted \emph{algebraic} subsystems of $\AC^0[p]$-Frege such as the Nullstellensatz \cite{BeameIKPP96} and Polynomial Calculus \cite{CEI96} proof systems, which we discussed next.\vspace{-4pt}

\subsection{Related Work on Algebraic Proof Systems}
    
    Much of the research on $\AC^0[p]$-Frege lower bounds has focused on subsystems such as Nullstellensatz and Polynomial Calculus, which encode a CNF and the Boolean constraints on its variables as polynomials and reduce the task of proving that the CNF has no satisfying Boolean assignments\footnote{We will sometimes switch back and forth in our discussion between the tasks of refuting that a CNF formula is satisfiable and of proving that a DNF formula is a tautology, which are equivalent by De Morgan's laws.} to proving that these polynomials do not have a common zero \cite{BeameIKPP96,BussIKPRS96}. Nullstellensatz and Polynomial Calculus are propositional proof systems in the Cook-Reckhow sense \cite{CR79}, since the verification of proofs can be done in deterministic polynomial time.

    Pitassi \cite{Pit97} proposed  more powerful algebraic proof systems where verifying the correctness of a proof requires identity testing of algebraic circuits, i.e., checking whether a given algebraic circuit is identically zero. Identity testing is known to be doable in randomized polynomial time, but it remains a long-standing open question whether it can be done in deterministic polynomial time for general algebraic circuits. Thus, the more general algebraic systems proposed by Pitassi are not propositional proof systems in the traditional Cook-Reckhow sense, but they do have efficient {\it randomized} verification.

    Grochow and Pitassi \cite{GP18} defined a strong algebraic proof system in this sense, called the {\it Ideal Proof System} (IPS), where a single algebraic circuit acts as a {\it certificate} that a set of polynomial equations does not have a common zero. The size of an IPS proof is simply the size of the corresponding algebraic circuit which acts as a certificate. \cite{GP18} showed that IPS is at least as strong as EF, and also that super-polynomial IPS lower bounds for {\it any} sequence of unsatisfiable formulas implies that the Permanent does not have polynomial-size algebraic circuits. This gives a long sought-after connection between proof complexity lower bounds for strong proof systems and (algebraic) circuit complexity lower bounds, but for an algebraic proof system with randomized verification rather than for a propositional proof system. In a more recent paper \cite{ST25}, the implication from proof complexity lower bounds to algebraic complexity lower bounds was strengthened to an {\it equivalence} for a certain explicit sequence of formulas. 

    Given that a proof in IPS is a single algebraic circuit, we can define and study variants of IPS where this circuit is of a restricted form. This has been done in several works in the past decade \cite{FSTW21,AF22,GHT22,HLT24}, which seek to show proof complexity lower bounds for subsystems of IPS or to find closer connections between IPS variants and propositional proof systems. Proof complexity lower bounds in this setting are typically shown by adapting algebraic circuit lower bound techniques. One apparent drawback to this approach is that the hard candidates in these works are not propositional formulas, but rather purely algebraic instances (e.g., $x_1+\dots+x_n + 1 =0$) that cannot be directly translated to propositional logic (cf.~\cite{EGLT25} for a discussion of this point).
    
    In 2021, a breakthrough super-polynomial algebraic circuit lower bound against constant-depth algebraic circuits was shown by Limaye, Srinivasan and Tavenas \cite{LST25} for large enough fields, and very recently has been extended by Forbes \cite{forbes2024low} to fields of characteristic $p$ for any prime $p$. This motivates the question of whether a similar super-polynomial {\it proof size} lower bound holds for constant-depth IPS, where the certificate is a constant-depth algebraic circuit. Some progress on this question has been made in recent work \cite{AF22,GHT22,HLT24}, but it remains open whether there are unsatisfiable CNF formulas (or propositional logic formulas more generally) requiring super-polynomial constant-depth IPS proofs.

\subsection{Framework and Results}\label{sec:framework-results}

While most lower bounds in propositional proof complexity for concrete tautologies rely on combinatorial or algebraic techniques, one can also try to use logic—specifically, diagonalization—for this purpose. 
This is a natural idea, but in the propositional setting it faces an inherent obstacle: self-reference.  
Suppose that a formula $\Phi$ expresses a proof-complexity lower bound. Ideally, one would like $\Phi$ to encode the statement that ``$\Phi$ itself has no short proofs.'' This, however, is impossible: any reasonable propositional encoding of $\Phi$ must use at least $|\Phi|$ symbols, so $\Phi$ would necessarily be longer than itself.%
\footnote{Friedman~\cite{Friedman79} and Pudlak~\cite{Pudlak86, Pudlak87} show how to adapt the proof of G\"{o}del's Second Incompleteness Theorem to give {\it sub-linear} proof size lower bounds for strong enough propositional proof systems.}

One way to circumvent this problem, due to Kraj\'{\i}\v{c}ek, is to encode proofs implicitly, and hence more economically, via a circuit that computes the bits of the proof given their index (see \cite{Krajicek04b}).  

Santhanam--Tzameret \cite{ST25} proposed a different approach: instead of referring to itself directly, a formula refers to a smaller version of itself.  
Concretely, they introduced the \emph{iterated lower bound formulas}, which encode inductively the statement that ``the previous-level formula has no short proof,'' thus avoiding direct self-reference. This yields diagonalization-based formulas that provably lack short proofs infinitely often. Specifically, if at level $\ell$ the formula  
\[
\varphi_\ell := \text{``there is no short proof of } \varphi_{\ell-1}\text{''}
\]  
has either a long proof or a short proof, then in both cases it establishes the truth of the no-short-proof claim for $\varphi_{\ell-1}$. If $\varphi_\ell$ has no proof at all, then in particular it has no short proof. Thus, we are done: either $\varphi_\ell$ or $\varphi_{\ell-1}$ has no short proof, and this holds for infinitely many $\ell$.

By referring to a smaller (and therefore different) version of itself, the resulting formulas demonstrate that $\varphi_\ell$ has no short proofs infinitely often. However, this still leaves open whether these formulas are hard tautologies or not tautologies at all (and thus vacuously unprovable). Let us elaborate on the difference between showing that a statement is not easy, and showing that it is hard, i.e., that it is a tautology with no short proofs.  

In general, given a proof system, every propositional
formula $\varphi$ falls into one of three categories:  
1. a tautology with a short proof (easy),  
2. a tautology without short proofs (hard), or  
3. not a tautology (hence unprovable).  
This is depicted in \Cref{tab:1}.

\begin{table}[h]
\centering
\begin{tabular}{|>{\columncolor{gray!20}}c|c|c|}
\hline
\cellcolor{gray!22}\textbf{Proof Complexity $\downarrow$ / Validity $\rightarrow$} 
 & \cellcolor{gray!22}\textbf{Tautology} 
 & \cellcolor{gray!22}\textbf{Non-tautology} \\
\hline
Easy       & 1. easy formula & \cellcolor{gray!6}\textbf{X} \\
\hline
Hard       & 2. hard formula & \cellcolor{gray!6}\textbf{X} \\
\hline
Unprovable & \cellcolor{gray!6}\textbf{X} & 3. trivially unprovable \\
\hline
\end{tabular}
\caption{A priori, every propositional formula falls into one of the cells labeled 1, 2, or 3.}
\label{tab:1}
\end{table}

The standard aim in proof complexity is to establish hardness, i.e., to identify a formula in cell 2. For strong proof systems this remains open. Thus, a natural intermediate step is \emph{ruling out that a formula is easy} (cell 1), while leaving open whether it belongs to cell 2 or 3.  
For this to be meaningful, one must consider sequences of formulas whose tautological status is unknown. A natural choice is formulas expressing open lower bounds in complexity theory. This viewpoint was emphasised by Razborov \cite{Razb15-annals}, who highlighted the importance of studying the proof complexity of natural statements whose validity is unknown. This approach was pursued in the 1990s by Razborov and others~\cite{Razb95, Razb95a, Razb98, Razb15-annals} (see also Kraj\'{i}\v{c}ek \cite{Krajicek04c} and Raz \cite{Raz04-JACM}).\footnote{Typically, these propositional formulas are known unconditionally to be tautologies for random objects. For example, if a formula encodes a circuit lower bound for some Boolean function $f$, then for random $f$ the lower bound holds. Thus, while the interesting candidate formula stating, say, SAT $\not\subseteq \P/\poly$ is not known to be a tautology, for a random $f$ the statement $f \not\in \P/\poly$ is.}

In this work, as in \cite{ST25}, we follow this approach. We show that a family of DNF formulas of unknown validity has no short proofs. This situation is illustrated in \Cref{tab:2}.

\begin{table}[h]
\centering
\begin{tabular}{|>{\columncolor{gray!20}}c|c|c|}
\hline
\cellcolor{gray!20}\textbf{Proof Complexity $\downarrow$ / Validity $\rightarrow$} 
 & \cellcolor{gray!20}\textbf{Tautology} 
 & \cellcolor{gray!20}\textbf{Non-tautology} \\
\hline
Easy       & 1. \cellcolor{gray!10}\textbf{\color{red} X}& \cellcolor{gray!10}\textbf{X} \\
\hline
Hard       & 2. hard formula & \cellcolor{gray!10}\textbf{X} \\
\hline
Unprovable & \cellcolor{gray!10}\textbf{X} & 3. trivially unprovable \\
\hline
\end{tabular}
\caption{The DNF formula under consideration in this paper is shown to lie in either cell 2 or cell 3; in particular, we \emph{rule out} the possibility that it is in cell 1 (with respect to the proof system \ACZP-Frege).}\label{tab:2}
\end{table}

There is another reason to consider statements of unknown validity when proving lower bounds for strong systems: very few plausible hard candidates are known for Frege and Extended Frege. The only compelling ones are random CNFs and circuit lower bound formulas, and for neither do we have effective methods for certifying validity. Indeed, confirming the validity of a fixed formula $\phi$ typically requires a proof, and existing proof methods can usually be captured in polynomial-size Frege or Extended Frege proofs—implying that $\phi$ is not a plausible hardness candidate for these systems.

\smallskip

As mentioned above, iterated lower bound formulas provide a simple way to show that a sequence of formulas is not easy. However, unlike formulas asserting SAT $\notin \P/\poly$, it is unclear whether there is any reason to believe they are tautologies. 
%
%
What \cite{ST25} demonstrated is that assuming some circuit lower bounds one can go beyond the iterated lower bound formulas: rule out that some formulas are easy for formulas that are \textit{more likely to be tautologies}---specifically, formulas expressing that proving algebraic \textit{circuit} lower bounds is hard. 
In particular, \cite{ST25} combined the diagonalization framework with the \cite{GP18} reduction from proof complexity to circuit lower bounds that showed that a proof-size lower bound implies an algebraic circuit-size lower bound (VP $\neq$ VNP). 
Hence, instead of stating recursively that lower bounds are hard to prove, one formulates
$\varphi$ to be the statement ``there are no short proofs of algebraic circuit lower bounds of the Permanent''. This avoids the recursively defined statement of the iterated lower bound formulas, and is a statement whose validity is easier to assess and use.
In this way we obtain:  
\[
    \text{circuit lower bound}
    \;\;\Rightarrow\;\;
    \text{no short proof of }%
    \overbrace{%
                \text{%
                    `no short proof of the %
                    circuit lower %
                    bound'%
                    }%
              }^{\varphi}.%
\]  
However, this still leaves an extra conditional layer: we need to rely on $\VP\neq\VNP$ (the circuit lower bound) to rule out that $\varphi$ is easy.  
\bigskip

\emph{In the present work, we eliminate this conditionality in the constant-depth regime}. Using the unconditional constant-depth algebraic circuit lower bound of Limaye, Srinivasan and Tavenas \cite{LST25}, we extend the diagonalization framework to constant-depth circuits.
\begin{quote}
We demonstrate formulas that \textit{unconditionally} are not easy for $\ACZ[p]$-Frege, such that:  
\begin{itemize}
    \item \textbf{Plausibly tautologies:} we give some evidence supporting the validity of these formulas;

    \item \textbf{Necessary (complete):} any hard instance for a sufficiently strong proof system must imply that our formulas are also not easy for that system;

    \item \textbf{Amplifying:} exhibiting a form of existential depth amplification: there exists a constant $d$, such that if our formulas are hard for depth-$d$ proofs then they are hard for \emph{every} constant-depth $d'$ proofs.  
\end{itemize}
\end{quote}

\begin{theorem} [Corollary of the more formal \Cref{thm:mainthm2} below] \label{thm:mainthm1}
    For every prime $p$ there is an explicit sequence $\{\phi_n\}$ of DNF formulas (of unknown validity) such that there are no polynomial-size \ACZP-Frege proofs of $\{\phi_n\}$.
\end{theorem}

    As mentioned above,  \Cref{thm:mainthm1} is shown via a diagonalization argument applied to the algebraic proof system IPS. It exploits the lower bound in \cite{LST25} and Forbes' finite fields version  \cite{forbes2024low}, and the fact that Grochow-Pitassi \cite{GP18} showed that constant-depth IPS over finite prime fields simulates $\ACZ[p]$-Frege. 

The formulas $\phi_n$ express proof size lower bounds for algebraic proof systems, and the evidence for their validity comes from several sources.


\begin{remark*}[Switching between DNF tautologies and unsatisfiable CNF formulas]
    Since we work with the \emph{refutation} system IPS (whose constant-depth version simulates \ACZ$[p]$-Frege), it is more natural to consider CNF formulas that are conjectured to be unsatisfiable, rather than DNF formulas that are conjectured to be valid. This is a matter of convenience, because of the trivial equivalence between showing that a CNF is unsatisfiable and showing that its complementary DNF is valid. We sometimes abuse notation and use ``refutations'' and ``proofs'' interchangeably; in all cases, the exact meaning should be clear from the context. 
\end{remark*}


The CNF formulas we consider are related to the circuit lower bound formulas considered by Razborov~\cite{Razb95, Razb95a, Razb98, Razb15-annals}, but with two differences: 
we consider {\it algebraic} rather than Boolean circuits, and our formulas express lower bounds for {\it proving} constant-depth algebraic lower bounds in constant-depth IPS, rather than express the circuit lower bounds directly. This allows us to adapt the diagonalization technique used to show an equivalence between circuit complexity and proof complexity in \cite{ST25} to derive an {\it unconditional} proof complexity lower bound in our setting.

To state our result more precisely we need the following notation. Let $\mathbb{F}$ be some underlying (finite) field. We let:
\begin{itemize}
    \item ${\rm ckt}_d(\perm_n,s)$: a CNF formula expressing that the Permanent on $n \times n$ matrices has depth-$d$ algebraic circuits of size $s$ over $\mathbb{F}$.
    
    
    \item ${\rm ref\text{-}IPS}_{d}(\phi, t)$: a CNF formula expressing that $\phi$ has size-$t$ IPS refutations of depth-$d$ over $\mathbb{F}$.

    \item The {\bf \textit{diagonalizing CNF formula}} is:
        $$\psi_{d,d',n} := 
        {\rm ref\text{-}IPS}_{d'}%
        ({\rm ckt}_d(%
                        \perm_n, n^{O(1)}
                    )%
        , N^{O(1)}%
        ),%
        $$  
        where $N = |{\rm ckt}_d(\perm_n, n^{O(1)})|$ is $O(2^{n^{O(1)}})$, expressing that depth-$d'$ IPS refutes in polynomial-size that the permanent is computed by a depth-$d$ algebraic circuit of polynomial-size.
\end{itemize}
The use of $O(1)$ in the notation above is informal. We use it to avoid using more quantifiers that would make the statement below hard to parse. 

\begin{theorem}[Informal Statement%
        ; \Cref{theorem: main theorem}%
            ]
    \label{thm:mainthm2}
    For every constant prime $p$ and for all positive integers $d$, there is a positive integer $d'$ such that for all positive integers $d''$, there are no polynomial-size depth-$d''$ IPS refutations of the formula 
    $
    \psi_{d,d',n} 
    $, 
    for infinitely many $n$ over $\mathbb{F}_{p}$.
\end{theorem}


Theorem \ref{thm:mainthm2} states the existence of no polynomial constant-depth IPS refutations for formulas that themselves express constant-depth IPS short refutation of constant-depth circuit upper bounds.
%
Here the size of a proof is always measured as a function of the length of the formula being proved. The result holds for any finite field, and further for sequence of prime fields of increasing size that are not too large, as we show below in \Cref{thm:mainthm4}.

%

    Theorem \ref{thm:mainthm1} follows from Theorem \ref{thm:mainthm2} as follows: i) take the same $p$ as in the statement of Theorem \ref{thm:mainthm2};
    ii) let the formulas $\phi_n$ in the former result to be the negation of $\psi_{d,d',n}$ in the latter result\footnote{We take negation because we consider \ACZP-Frege to be a {\it proof system} for DNF tautologies, while constant-depth IPS is a {\it refutation system} for unsatisfiable CNFs.}, where $d$ (the stated depth of circuit computing Permanent) is chosen to be large enough and $d'$ (the stated depth of IPS refutations) is chosen as a function of $d$ so that Theorem \ref{thm:mainthm2} holds; 
    iii) finally, use the fact that constant-depth IPS over fields of characteristic $p$ simulates $\AC^0[p]$-Frege \cite{GP18}.

\subsubsection{Implications and Several Important Aspects of \Cref{thm:mainthm2}}

\paragraph{Supporting evidence for the unsatisfiability of the diagonalizing CNF formula.}
We describe three forms of supporting evidence that $\psi_{d,d^\prime,n}$ is unsatisfiable.

\medskip 

\noindent{(I)\textit{\uline{Hardness of multilinear algebra for constant-depth proofs}}.
To establish the unsatisfiability of $\psi_{d,d',n}$, it suffices to show that constant-depth algebraic circuit upper bound formulas do not admit small-size constant-depth IPS refutations. 

The tensor rank principle denoted $\TRankP_{m,n}^r(A)$ and introduced in~\cite{GGLST25}, states that an order-$r$ tensor $A$ of rank $m$ can be written as the sum of $n$ rank-$1$ tensors of order $r$. This principle is easily shown to be unsatisfiable when $m > n$. 
%
%
In~\cite{GGLST25} the tensor rank principle was reduced to \emph{constant-depth} algebraic circuit upper bound formulas in constant-depth IPS. As a result, proving super-polynomial-size lower bounds for the tensor rank principle against constant-depth IPS implies the unsatisfiability of the diagonalizing CNF formula:

\begin{corollary}[\cite{GGLST25}; informal, see \Cref{Cor:6.6}]
    If the sequence of CNF formulas $\psi_{d,d',n}$ are satisfiable then the tensor rank principle $\TRankP_{m,n}^r(A)$ admits polynomial-size refutations in depth-$O(d')$ IPS. 
\end{corollary}

It is known from \cite{LST25,forbes2024low} that the determinant cannot be computed by polynomial-size constant-depth algebraic circuits over any field. Consequently, following the informal alignment between proof complexity and circuit complexity, one expects that proof systems operating with constant-depth algebraic circuits cannot efficiently prove statements expressing linear (or multilinear) algebraic properties (whose standard proofs use notions like rank and determinants). Therefore, it is reasonable to expect that $\TRankP_{m,n}^r(A)$ does not admit polynomial-size refutations in constant-depth IPS, and hence, that $\psi_{d,d',n}$ is unsatisfiable by the corollary.

\medskip 
\noindent(II)\textit{\uline{Weaker versions of the CNF that are provably unsatisfiable}}.
To establish the unsatisfiability of $\psi_{d,d',n}$, it suffices to show that constant-depth algebraic circuit upper bound formulas do not admit small-size constant-depth IPS refutations. Here we mention two recent results that establish this for weaker variants of 
        $\psi_{d,d',n} = 
        {\rm ref\text{-}IPS}_{d'}%
        ({\rm ckt}_d(%
                        \perm_n, n^{O(1)}
                    )%
        , N^{O(1)}%
        ),%
        $ 
namely when the proof system is Polynomial Calculus with Resolution ($\PCR$) instead of depth-$d'$ IPS (denoted ``IPS$_{d'}$'' in $\psi_{d,d',n}$) and when the lower bound statement is against either algebraic circuits of unrestricted depth instead of general algebraic circuits (denoted ``{\rm ckt$_d$}'' in $\psi_{d,d',n}$), or against noncommutative algebraic branching programs.






The refutation system $\PCR$ can be considered roughly as depth-$2$ $\IPS$ (see~\cite{GP18}).
Note that when we increase the strength of the algebraic circuit model replacing ${\rm ckt}_d$ in $\psi_{d,d',n}$, we are actually \emph{weakening} the statement, since proving lower bounds against stronger circuit model is \emph{harder}, meaning that it is \emph{easier} to show that such lower bounds are harder. 

\begin{corollary}[\cite{GGLST25}; \Cref{theorem: super-poly lower bound for AUP}]
    \label{cor: 1.3}
        Let $f$ be any polynomial in $n$ variables over $\F_2$.
        The CNF formula      
        $%
        {\rm ref\text{-}\PCR}%
        ({\rm ckt}(%
                        f, n^{O(1)}
                    )%
        , N^{O(1)}%
        ),%
        $ stating that $\PCR$\ over $\F_2$ has a polynomial-size refutation of the statement that $f$ is computable by polynomial-size algebraic circuits, is unsatisfiable. 
\end{corollary}
\smallskip

When the proof system in $\psi_{d,d',n}$ is weakened again to $\PCR$ instead of depth-$d'$ IPS, while the algebraic circuit model is very weak (which \emph{strengthens} the proof complexity lower bound statement), namely, noncommutative algebraic branching program (denoted ncABP), we have the following: 

\begin{corollary}[\cite{GGRT25}]
    \label{cor: ncABP}
        Let $f$ be any noncommutative polynomial in $n$ variables over $\F_2$.
        The CNF formula      
        $%
        {\rm ref\text{-}\PCR}%
        ({\rm ncABP}(%
                        f, n^{O(1)}
                    )%
        , N^{O(1)}%
        ),%
        $ stating that $\PCR$ over $\F_2$ has polynomial-size refutations of the statement that $f$ is computable by polynomial-size noncommutative algebraic branching programs, is unsatisfiable. 
\end{corollary}

\Cref{cor: ncABP} is proved via a reduction similar to the reduction from iterated proof complexity generators to Boolean circuit upper bound formulas in \cite{Razb15-annals}. Since an algebraic branching program is characterized by iterated matrix multiplication, one can reduce the iterated rank principle to the ncABP upper bound formulas \cite{GGRT25}. Thus, \Cref{cor: ncABP} follows by an exponential lower bound for the iterated rank principle.




\medskip 

\noindent{(III)}\textit{\uline{Algebraic analogues of proof complexity conjectures}}: Kraj\'{i}\v{c}ek~\cite{Krajicek04c} and Razborov~\cite{Razb96, Razb16} have conjectured the following:%
\footnote{Razborov conjectured in \cite{Razb15-annals} that Frege cannot efficiently prove super-polynomial circuit lower bounds for any Boolean function. More specifically, \cite[Conjecture 1]{Razb15-annals} with suitable parameters for the underlying combinatorial designs implies under some hardness assumptions that Frege cannot efficiently prove that SAT$\not\subseteq\P/{\poly}$. Further conjectures about the impossibility of \emph{Extended} Frege to efficiently prove circuit lower bounds have been circulated in the proof complexity literature and discussions  (cf.~\cite{Razb-2016-talk,Razb-2021-talk,Kra11}).}
Extended Frege cannot efficiently prove {\it any} super-polynomial Boolean circuit lower bound.
Since Extended Frege is essentially a proof system that operates with Boolean circuits, this conjecture says that proof systems operating with Boolean circuits cannot efficiently prove Boolean circuit lower bounds. 
We raise the following analogous conjecture: 


\begin{quote}
\textbf{Constant-depth algebraic analogue of 
 Kraj\'{i}\v{c}ek-Razborov conjecture:}
For every $d$, there is a $d'$ such that there are no polynomial-size depth-$d'$ IPS proofs of super-polynomial depth-$d$ lower bounds for {\it any} polynomial $f$. \end{quote}
(Of course, the conjectured statement above depends on a natural encoding or formulation of the lower bound statement.)      


%
%
The constant-depth algebraic analogue of 
the Kraj\'{i}\v{c}ek-Razborov conjecture implies the unsatisfiability of the CNF formulas $\psi_{d,d',n}$ for arbitrary $d$ and large enough $d'$. 
%
%
The formulas $\psi_{d,d',n}$ assert this only for $f = \perm$, and therefore follow trivially from the conjecture.

\smallskip


\paragraph{Diagonalizing formulas are necessary for lower bounds.}
 We show that the diagonalizing formulas are not easy %
 is \emph{logically necessary} in order to prove super-polynomial lower bounds for $\mathcal{C}$-IPS for \emph{any} ``reasonable'' algebraic circuit class $\mathcal{C}$. In other words, we show that the non-easiness of the diagonalizing formulas is implied by {\it any} super-polynomial lower bounds on tautologies for algebraic proofs operating with circuits from $\mathcal{C}$. To show this, we use the circuits-to-proofs connection of \cite{GP18}. 
 
 The notation $\mathcal{C}$-IPS stands for the IPS proof system in which an
 IPS refutation (i.e., certificate) is written as an algebraic circuit from the class $\mathcal{C}$ (for instance, depth-$d$ circuits, for a constant $d$,  algebraic formulas, etc.). A ``reasonable" algebraic circuit class $\mathcal{C}$ is one for which the Grochow-Pitassi implication from $\mathcal{C}$-IPS lower bounds to $\mathcal{C}$ circuit lower bounds holds, and moreover this implication is efficiently provable in $\mathcal{C}$-IPS. Our methods in this paper imply that all the commonly studied algebraic circuit classes which contain the class of constant-depth algebraic circuits are reasonable.

 \begin{theorem}[Informal; \Cref{theorem: necessity}; If algebraic proofs are not p-bounded then it is not easy to prove that circuit lower bounds are hard for algebraic proofs.]
%
    \label{thm:mainthm3}
    Let $\mathcal{C}$ be any ``reasonable'' algebraic circuit class. If there is a sequence $\{\phi_n\}_n$ of unsatisfiable CNF formulas that requires super-polynomial size $\mathcal{C}$-\textup{IPS} proofs for infinitely many $n$, then the sequence $\{\psi_n\}_n$ of CNF formulas does not have polynomial size $\mathcal{C}$-{\textup {IPS}} proofs for infinitely many $n$, where $\psi_n = $ {\rm ref}-$\mathcal{C}$-{\textup {IPS}}$(\mathcal{C}$-\textup{ckt}$(\perm_n, n^{O(1)}), N^{O(1)})$, with $N = |\mathcal{C}$-{\rm ckt}$(\perm_n, n^{O(1)})|$.
 \end{theorem}

Theorem \ref{thm:mainthm3} is shown by abstracting the argument of Theorem \ref{thm:mainthm2} and combining the resulting generalization with the Grochow-Pitassi implication from proof complexity lower bounds to circuit complexity lower bounds~\cite{GP18}.

\para{Existential depth amplification.}
The statements $\psi_{d,d',n}$ themselves refer to proof complexity lower bounds. Nevertheless, the lower bounds stated in the formulas and the lower bounds we get from our result are \emph{different}.

First, the formulas state hardness of proving circuit lower bounds, while we get hardness of proving \emph{proof complexity lower bounds}. 

Second, notice the quantification over depths in Theorem \ref{thm:mainthm2}: there is some {\it fixed} depth $d'$ such that if depth-$d'$ IPS lower bounds on certain circuit lower bounds for the Permanent hold, then we get super-polynomial lower bounds for proofs of {\it any} constant depth. This shows that we can escalate the depth-$d'$ IPS lower bounds stated in the formulas to \emph{any} constant-depth IPS lower bounds. In other words, our result does not merely repeat the stated lower bound we assumed against depth-$d'$ proofs of circuit lower bounds, but goes beyond it to rule out short proofs in \emph{any} depth of proof complexity lower bounds.

\subsubsection{Extension to Larger Fields}

While \cite{ST25} and  Theorem \ref{thm:mainthm2} crucially use fields of constant size, and the ability to efficiently encode computation over finite fields of constant size by CNF formulas, we show further how to encode and reason about larger and growing fields. This technical contribution may be interesting on its own right. 

Specifically, by reasoning about the bits of polynomial expressions with algebraic proofs, Theorem \ref{thm:mainthm2} can be extended to underlying fields of size polynomially bounded by $|\psi_{d,d',n}|$. Bit arithmetic in proof complexity was used before (cf. \cite{Goe90, Bus87, AGHT24, IMP20}). We show how to reason about iterated addition, iterated multiplication and modular arithmetic in constant-depth IPS over polynomial-size fields.

 \begin{theorem}[Informal Statement; \Cref{theorem: main theorem in Polynomial-size Finite Fields}]
     \label{thm:mainthm4}
     Let $\{p_n\}$ be any sequence of primes such that $p_n = O(2^n)$, and $\mathbb{F}_{p_n}$ be the field of size $p_n$. For all positive integers $d$, there is a positive integer $d'$ such that for all positive integers $d''$, there are no polynomial-size depth-$d''$ IPS refutations of the formula $\psi_{d,d',n} = {\rm ref\text{-}IPS}_{d'}({\rm ckt}_d(\perm_n, n^{O(1)}), N^{O(1)})$ for infinitely many $n$ over $\mathbb{F}_{p_n}$, where $N = |{\rm ckt}_d(\perm_n, n^{O(1)})|$ is $O(2^{n^{O(1)}})$.
 \end{theorem}

\vspace{-8pt}

\subsubsection{The Diagonalization Argument}


Here we explain informally the idea behind the proof of \Cref{thm:mainthm2}, showing that the diagonalizing CNFs have no short refutations. This is where most of the nontrivial technical work lies.
The key idea is to combine diagonalization with the known implication from proof complexity lower bounds to circuit complexity lower bounds \cite{GP18}.

The argument builds on the following three nontrivial technical points:
\begin{enumerate}
\item  There is a reasonable $\CNF$ encoding expressing that the permanent polynomial can be computed by bounded-depth small-size algebraic circuits ``$\VNP = \VAC^0$'' (this is %
    ${\rm ckt}_d(%
        \perm_n, n^{O(1)}%
                )%
    $ %
    from before).%
\footnote{We use $\VAC^0$ to denote Valiant’s analogue of $\AC^0$, in the same way that $\VP$ and $\VNP$ correspond to $\P$ and $\NP$. The class $\VAC^0$ consists of families of polynomials computable by constant-depth, polynomial-size algebraic circuits.}\vspace{-7pt}
\item  There is a reasonable $\CNF$ encoding of the statement that there are constant-depth $\IPS$ refutations of size $s$ for a $\CNF$ $\phi$ %
    (this is ${\rm ref\text{-}IPS}_{d}(\phi, t)$ from before).%
    \vspace{-7pt}
\item  If $\phi$ is unsatisfiable and there are short and constant-depth $\IPS$ refutations of ``constant-depth $\IPS$ efficiently refutates $\phi$", then there are short and constant-depth $\IPS$ refutations of ``$\VNP = \VAC^0$".
\label{it:intro:GP18-in-constant-depth}
\end{enumerate}

The work of \cite{ST25} formalized the Grochow-Pitassi implication from $\IPS$ lower bounds for $\CNF$s to $\VNP \neq \VP$ within $\IPS$. Assumption 3 above is a novel formalisation of a \textit{constant-depth} version of the Grochow-Pitassi implication within constant-depth $\IPS$. 

\newcommand{\ubb}{\ensuremath{\mathsf{ub}}}

\newcommand{\lbb}{\ensuremath{\mathsf{lb}}}

According to \cite{LST25}, the permanent polynomial cannot be computed by constant-depth  small-size algebraic circuits, which means ``$\VNP = \VAC^0$" is an unsatisfiable $\CNF$, using Assumption 1. Assume for the sake of contradiction that 
        $\psi_{d,d',n}$, namely, 
        $%
        {\rm ref\text{-}IPS}_{d'}%
        (%
                        ``\VNP = \VAC^0\text{''}, \poly
        )%
        $ %
(or more formally, 
        $%
        {\rm ref\text{-}IPS}_{d'}%
        ({\rm ckt}_d(%
                        \perm_n, n^{O(1)}
                    )%
        , N^{O(1)}%
        )%
        $ 
)
has polynomial-size refutations in constant-depth $\IPS$. Then by Assumption 3, ``$\VNP = \VAC^0$" has polynomial-size constant-depth $\IPS$ refutations. But this contradicts the soundness of $\IPS$, since $\IPS$ has (polynomial-size) refutations of the statement that ``$\VNP = \VAC^0$" has  polynomial-size IPS-refutations (formally, one has to account for instances of different sizes when following this argument, and we explain this more precisely below).
\smallskip

To clarify further the argument, we describe a slightly more detailed overview, highlighting the logic behind the argument. We show that infinitely often there are no polynomial-size constant-depth $\IPS$ refutations of the formula ${\rm ref\text{-}IPS}_{d^\prime}({\rm ckt}_d(\perm_n,n^c), n^{c^\prime})$, expressing that there exists a constant $c^\prime$ such that ${\rm ckt}_d(\perm_n,n^c)$ has $\IPS$ refutations of size bounded from above by the polynomial $n^{c'}$ and depth bounded by $d^\prime$.
\medskip

\newcommand{\io}{\text{i.o.}}
\newcommand{\ale}{\text{a.e.}}

\noindent\textit{Proof sketch of  \Cref{thm:mainthm2} (formally \Cref{theorem: main theorem}).}
Let \io\ abbreviate \emph{infinitely often}, and let \ale\ denote its converse \emph{almost everywhere}, that is, ``always except for finite many cases''. Let $G\sststile{d}{f(n)}1=0$ stands for an $\IPS$ refutation of $G$ of refutation-size bounded from above by $f(n)$ and refutation-depth bounded from below by $d$.
Our goal is to prove:\vspace{-12pt} 
\begin{equation}
    \forall\ c, d \ \exists \ c^\prime, d^\prime\ \forall \ c^{\prime\prime}, d^{\prime\prime}\
\io~\underbrace{{\rm ref\text{-}IPS}_{d^\prime}(\overbrace{{\rm ckt}_d(\perm_n,n^c)}^{\gamma}, |\gamma|^{c\prime}) }_{\lambda}\not \sststile{d^{\prime\prime}}{|\lambda|^{c^{\prime\prime}}} 1=0,
\end{equation}
meaning that for all constants $c,d$ there are constants $c^\prime,d^\prime$, such that for all constants $c^{\prime\prime},d^{\prime\prime}$, $\lambda$ does not have polynomial-size $|\lambda|^{c^{\prime\prime}}$ and constant depth $d^{\prime\prime}$ refutation, infinitely often.

Assume by way of contradiction that the converse holds, namely:\vspace{-8pt} 
\begin{equation}\label{eq:328:30}
\exists\ c, d \ \forall \ c^\prime, d^\prime\ \exists \ c^{\prime\prime} ,d^{\prime\prime}\ \ale~
\underbrace{{\rm ref\text{-}IPS}_{d^\prime}(\overbrace{{\rm ckt}_d(\perm_n,n^c)}^{\gamma}, |\gamma|^{c\prime}) }_{\lambda} \sststile{d^{\prime\prime}}{|\lambda|^{c^{\prime\prime}}} 1=0,
\end{equation}
Using the bounded-depth version of the argument of \cite{GP18} (i.e.,  \Cref{it:intro:GP18-in-constant-depth} above), we get that if $\VNP$ has polynomial-size depth-$d$ circuits, then depth-$d$ $\IPS$ is polynomially bounded, namely:
\begin{equation*}
\exists c_1, d_1 \  \underbrace{{\rm ckt}_d(\perm_{|\gamma|},|\gamma|^c)}_{\gamma^\prime} \sststile{d_1}{|\gamma^\prime|^{c_1}} {\rm ref\text{-}IPS}_{d}({\rm ckt}_d(\perm_n,n^c), |\gamma|^{c}) \sststile{d^{\prime\prime}}{|\lambda|^{c^{\prime\prime}}} 1=0,
\end{equation*}
where $|\lambda|$ is polynomially bounded by $|\gamma^\prime|$. The second part from left of the refutation above is given by the fact that ${\rm ref\text{-}IPS}_{d^\prime}({\rm ckt}_d(\perm_n,n^c), |\gamma|^{c^\prime})$ can be efficiently refuted in bounded-depth $\IPS$ for any $d^\prime$ and $c^\prime$. Hence, we take $d^\prime$ to be $d$ and $c^\prime$ to be $c$, yielding the second part of the refutation above.

Therefore, in particular, by combining the two proofs into one, we get \vspace{-8pt} 
\begin{equation*}
\exists c_2, d_2 \ {\rm ckt}_d(\perm_{|\gamma|},|\gamma|^c) \sststile{d_2}{|\gamma^\prime|^{c_2}}1=0.
\end{equation*}
Thus,\vspace{-8pt}
\begin{equation*}
\exists c_2,d_2,\ \underbrace{{\rm ref\text{-}IPS}_{d_2}({\rm ckt}_d(\perm_{|\gamma|},|\gamma|^c), |\gamma^\prime|^{c_2})}_{\Lambda} \end{equation*}
is a \emph{satisfiable} $\CNF$ formula.
Since \Cref{eq:328:30} holds almost everywhere, we can take $n$ to be $|\gamma|$. Then, by taking $d^\prime$ to be $d_2$ and $c^\prime$ to be $c_2$ in \Cref{eq:328:30},
\begin{equation*}\vspace{-8pt}
    \exists C, D \ \underbrace{{\rm ref\text{-}IPS}_{d_2}({\rm ckt}_d(\perm_{|\gamma|},|\gamma|^c), |\gamma^\prime|^{c_2})}_{\Lambda} \sststile{D}{|\Lambda|^C} 1=0,
\end{equation*}
which means $\Lambda$ is refutable. By the soundness of $\IPS$, $\Lambda$ is unsatisfiable which is a contradiction.
\bigskip

    \subsection{Relation to Previous Work and Conclusion}

 

    Theorem \ref{thm:mainthm1} gives a first instance in which a formula whose validity is unknown is shown unconditionally to have no short proofs in $\AC^0[p]$-Frege. Razborov~\cite{Razb96} and Kraj\'{i}\v{c}ek \cite{Krajicek04c} {\it conjecture} that all Boolean circuit lower bound formulas are hard for EF, but this is open even for $\AC^0$-Frege and even when we relax hardness to being non-easy (as we do in our work). Our lower bound is for {\it proof complexity lower bound} formulas rather than for circuit lower bound formulas, but our work is somewhat similar in spirit to \cite{Razb95}.

    From a technical point of view, our work is related to recent works by \cite{ST25,PS19}, which also use diagonalization ideas. The main result of \cite{ST25} is a {\it conditional} existence of hard formulas for  IPS, under the assumption that $\VNP \neq \VP$. In contrast, our result ruling out easy formulas for constant-depth IPS, i.e., Theorem \ref{thm:mainthm2}, is unconditional. One way to interpret our result is that it strengthens the equivalence between proof complexity lower bounds and circuit complexity lower bounds shown in \cite{ST25} to hold for constant-depth circuits, and then applies the recent breakthroughs constant-depth algebraic circuit lower bounds \cite{LST25, forbes2024low} to get unconditional proof complexity lower bounds.

    Unconditional lower bounds for conjectured tautologies are also shown in \cite{PS21} using a somewhat different diagonalization technique, however these are for highly non-explicit formulas\footnote{There are two sources of non-explicitness in \cite{PS21}. First, they consider a {\it distribution} on conjectured hard formulas rather than a fixed hard formula at every length. Second the formulas refer to a non-constructively defined proof system with non-uniform verification.}, and do not seem directly relevant to progress on lower bounds for tautologies. In contrast, Theorem \ref{thm:mainthm3} shows that Theorem \ref{thm:mainthm2} is a {\it necessary step} toward super-polynomial lower bounds for constant-depth IPS.

    One novel aspect in Theorem \ref{thm:mainthm1} is that a proof complexity lower bound for a propositional proof system is shown via {\it algebraic circuit lower bounds}. The theory of feasible interpolation in propositional proof complexity enables proof complexity lower bounds for weak systems such as Resolution and Cutting Planes to be derived from lower bounds on monotone circuit complexity \cite{BPR97,Kra97-Interpolation,Pud97}. However, there is cryptographic evidence against the applicability of feasible interpolation techniques to $\AC^0$-Frege and stronger proof systems \cite{BPR00}. Implications from average-case circuit lower bounds to proof size lower bounds for strong systems were conjectured by Razborov \cite{Razb15-annals} in the context of {\it proof complexity generators} \cite{ABSRW00}, but these conjectures are so far unproven.
    

    We note that the question of proving proof complexity lower bounds for constant-depth IPS has in itself been highlighted and studied in recent works \cite{AF22, GHT22, HLT24, EGLT25, BLRS25}. Andrews and Forbes \cite{AF22} show a super-polynomial constant-depth IPS lower bound for refuting certain sets of polynomial equations. However, their hard instances do not themselves have polynomial-size constant-depth circuits, and in particular are not CNFs. Govindasamy, Hakoniemi and Tzameret \cite{GHT22} give a super-polynomial multilinear constant-depth IPS lower bound for refuting polynomial equations expressible as small depth-2 algebraic circuits. Hakoniemi, Limaye and Tzameret \cite{HLT24} extend and strengthen these results from multilinear to low individual degree proofs  \cite{GHT22}. However, they also show that the lower bound framework of \cite{GHT22, HLT24} is incapable of proving lower bounds for CNFs. Elbaz, Govindasamy, Lu and Tzameret \cite{EGLT25} showed that any constant-depth IPS lower bounds over finite fields would lead to lower bounds for CNF formulas and thus \ACZ$[p]$-Frege lower bounds. Alas, the strongest constant-depth IPS lower bounds \cite{HLT24} are restricted to low individual degree refutations, for which the result of \cite{EGLT25} do not hold.
    Proving lower bounds for CNFs is essential for the application to $\AC^0[p]$-Frege lower bounds, and none of the previous works achieve this. Our work provides an explicit CNF formula with no short refutations and with some evidence for its unsatisfiability.

    \section{Preliminaries}

    \subsection{Algebraic Complexity}
    Let $\mathbb{F}$ be a field. Denote by $\mathbb{F}[\overline{x}]$ the field of polynomials with coefficients from $\mathbb{F}$ and variables $\overline{x}=[x_1,\cdots,x_n]$. A polynomial is a formal linear combination of monomials, where a monomial is a product of variables. Two polynomials are identical if all their monomials have the same coefficients. The degree of a polynomial is the maximum total degree of a monomial in it.

    \begin{definition}[Depth-$\Delta$ algebraic circuits and algebraic formulas]
        An \emph{algebraic circuit} over a field $\mathbb{F}$ is a finite directed acyclic graph. The leaves are called input nodes, which have in-degree zero. Each input node is labelled either with a variable or a field element in $\mathbb{F}$. All the other nodes have unbounded in-degree and are labelled by  $+$ or $\times$. The output of a $+$ (or $\times$) node computes the addition (product, \resp) of the polynomials computed by its incoming nodes. An algebraic circuit is called an algebraic \emph{formula} if the underlying directed acyclic graph is a tree (every node has at most one outgoing edge). The \emph{size} of an algebraic circuit $C$ is the number of nodes in it, denoted by $|C|$. The \emph{depth} of $C$ is the length of the longest directed path in it, denoted by $\depth(C)$. If $\depth(C)=\Delta$ we call $C$ a \emph{depth-$\Delta$ circuit}.
    \end{definition}

    The \emph{product-depth} of an algebraic circuit is the maximum number of product gates on a root-to-leaf path. The product-depth is, without loss of generality, equal to the depth up to a factor of two.
    \begin{definition}[Syntactic-degree $\sdeg(\cdot)$]
        Let $C$ be an algebraic circuit and $v$ a node in $C$. The \emph{syntactic-degree} $\sdeg(v)$ of $v$ is defined as follows:
        \begin{enumerate}
            \item If $v$ is a field element or a variable, then $\sdeg(v) \coloneqq 0$ and $\sdeg(v) \coloneqq 1$, respectively;
            \item If $v=\sum_{i=0}^t u_i$ then $\sdeg(v) \coloneqq \max\{\sdeg(u_0),\cdots, \sdeg(u_t)\}$;
            \item If $v=\prod_{i=0}^t u_i$ then $\sdeg(v) \coloneqq \sum_{i=0}^t \sdeg(u_i)$.
        \end{enumerate}
    \end{definition}
    
    \begin{definition}[$\VP$ 
    \cite{Val79:ComplClass}]
         Over a field $\mathbb{F}$, $\VP_\mathbb{F}$ is the class of families $f=(f_n)_{n=1}^{\infty}$ of polynomials $f_n$ such that $f_n$ has $\poly(n)$ input variables, is of $\poly(n)$ degree, and can be computed by algebraic circuits over $\mathbb{F}$ of $\poly(n)$ size.
    \end{definition}

    \begin{definition}[$\VNP$ \cite{Val79:ComplClass}]
         Over a field $\mathbb{F}$, $\VNP_\mathbb{F}$ is the class of families $g=(g_n)_{n=1}^{\infty}$ of polynomials $g_n$ such that $g_n$ has $\poly(n)$ input variables and is of $\poly(n)$ degree, and can be written as
        \[
            g_n(x_1, \cdots , x_{\poly(n)})=\sum_{\overline{e} \in \{0,1\}^{\poly(n)}} f_n(\overline{e},\overline{x})
        \]
        for some family $(f_n) \in \VP_\mathbb{F}$.
    \end{definition}

    \begin{definition}[$\VAC^0$]
        \label{def: VAC0}
        Over a field $\F$, $\VAC^0_\F$ is the class of families $f=(f_n)_{n=1}^{\infty}$ of polynomials $f_n$ such that $f_n$ has $\poly(n)$ input variables, is of $\poly(n)$ degree, and can be computed by algebraic circuits over $\mathbb{F}$ of $\poly(n)$ size and depth $O(1)$.
    \end{definition}
    Notice that $\VP$,$\VNP$ and $\VAC^0$ are nonuniform complexity classes.

    \begin{definition}[Projection reduction \cite{Val79:ComplClass}]
         A polynomial $f(x_1,\cdots, x_n)$ is a \emph{projection} of a polynomial $g(y_1,\cdots, y_m)$ if there is a mapping $\sigma$ from $\{y_1,\cdots,y_m\}$ to $
         \{0,1, x_1, \cdots, x_n\}$ such that $f(x_1,\cdots,x_n) = g(\sigma(y_1),\cdots,\sigma(y_m))$.
        A family of polynomials $(f_n)$ is a \emph{polynomial projection} or \emph{p-projection} of another family $(g_n)$ if there is a function $t(n)=n^{\Theta(1)}$ such that $f_n$ is a projection of $g_{t(n)}$ for all (sufficiently large) $n$.
        We say that $f$ is projection-reducible to $g$ if $f$ is a projection of $g$.
    \end{definition}

    The symmetric group, denoted by $S_n$, over $n$ elements $\{1,\dots,n\}$ is the group whose elements are all bijective functions from $[n]$ to $[n]$ and whose group operation is that of function composition. The sign $\sgn(\sigma)$ of a permutation $\sigma\in S_n$ is $1$, if the permutation can be obtained with an even number of transpositions (exchanges of two [not necessarily consecutive] entries); otherwise, it is $-1$.

    \begin{definition}[Determinant]
        The \emph{Determinant} of an $n\times n$ matrix $A$ is defined as
        \[
            \det(A) := \sum_{\sigma \in S_n} (\sgn(\sigma) \prod_{i=1}^n a_{i,\sigma(i)}).
        \]
    \end{definition}
    \begin{definition}[Permanent] 
        The \emph{Permanent} of an $n \times n$ matrix $A = (a_{ij})$ is defined as
        \[
            \perm(A) := \sum_{\sigma \in S_n} \prod_{i=1}^n a_{i,\sigma(i)}.
        \]
    \end{definition}

    It is known that the Permanent polynomial is complete under p-projections for $\VNP$ when the field $\mathbb{F}$ is a field of characteristic different from $2$. The Determinant polynomial is not known to be complete for $\VP$ under p-projections.
    
    \begin{theorem}[\cite{Val79:ComplClass}]
        For every field $\mathbb{F}$, every polynomial family on $n$ variables that is computable by an algebraic formula of size $u$ is projection reducible to the Determinant polynomial (over the same field) on $u+2$ variables.
        For every field $\mathbb{F}$, except those that have characteristic $2$, every polynomial family in $\VNP_\mathbb{F}$ is projection reducible to the Permanent polynomial (over the same field) with polynomially more variables.
    \end{theorem}

    \begin{definition}[Iterated Matrix Multiplication]
        Let $n$ and $d$ be such that $N=dn^2$. The \emph{Iterated Matrix Multiplication} $\IMM_{n,d}$ on $N=dn^2$ variables is defined as the following polynomial. The underlying variables are partitioned into $d$ sets $X_1,\cdots,X_d$ of size $n^2$, each of which is represented as an $n \times n$ matrix with distinct variable entries. Then $\IMM_{n,d}$ is defined to be the polynomial that is the $(1,1)$th entry of the product matrix $X_1\cdot X_2 \cdots X_d$.
    \end{definition}

    Iterated Matrix Multiplication is in $\VP$.

    \begin{theorem}[Super-polynomial lower bounds against constant-depth circuits over large field \cite{LST25}]\label{theorem: first-LST25}
        Assume $d \leq \frac{\log n}{100}$ and the characteristic of $\mathbb{F}$ is $0$ or greater than $d$. For any product-depth $\Delta \geq 1$, any algebraic circuit $C$ computing $\IMM_{n,d}$ of product-depth at most $\Delta$ must have size at least $n^{d^{\exp{-O(\Delta)}}}$. 
    \end{theorem}

    \cite{BDS24} improved the lower bound for $\IMM$ against constant-depth. Let $\mu(\Delta) = 1/ (F(\Delta)-1)$, where $F(n)= \Theta(\varphi^n)$ is the $n$th Fibonacci number (starting with $F(0) = 1$, $F(1) =2$) and $\varphi = (1 + \sqrt{5})/2$ is the golden ratio.
    \begin{theorem}[\cite{BDS24}]
        \label{theorem: BDS24}
        Fixed a field $\F$ of characteristic $0$ or greater than $d$. Let $N, d, \Delta$ be such that $d= O(\log / \log \log N)$. Then, any product-depth $\Delta$ circuit computing $\IMM_{n,d}$ on $N = dn^2$ variables must have size at least $N^{\Omega(d^{\mu(2\Delta)}/ \Delta)}$.
    \end{theorem}

    A recent result by Forbes \cite{forbes2024low} extended these results to \emph{any} field, including finite fields.
    
    \begin{corollary}[Super-polynomial lower bounds on constant-depth circuits over any field \cite{forbes2024low}]
        Let $\mathbb{F}$ be any field, and $d=o(\log n)$. Then the iterated matrix multiplication polynomial $\IMM_{n,d}$ where $X_i$ requires
         \[
             n^{\Theta(d^{\mu(2\Delta)}/ \Delta)}
         \]
         size algebraic circuits of product depth $\Delta$.
    \end{corollary}

    Since $\IMM_{n,d}$ is a p-projection of the Permanent polynomial with $\poly(n,d)$ many variables, it follows that the Permanent does not have constant-depth polynomial-size circuits over any field, in the following sense.
    
    \begin{theorem}[No polynomial-size constant depth circuits for the Permanent \cite{LST25,forbes2024low}]
        \label{theorem: LST25}
        Let $\Delta \geq 1$ and let $\mathbb{F}$ be any field.  
        There are no constants $c_1,c_2$, such that the Permanent polynomial $\perm(A)$ of the $n \times n$ symbolic matrix $A$ over the field $\mathbb{F}$ is computable by an algebraic circuit of size $c_1 n^{c_2}$ and depth $\Delta$ where $\Delta$ is independent with $n$, for sufficiently large $n$. 
    \end{theorem}

    \subsection{Proof Complexity}

    \begin{definition}[Propositional proof system, \cite{CR79}] 
        \label{def: propositional proof system}
        A propositional proof system is a polynomial-time computable relation $R(\cdot, \cdot)$ such that for each $x \in \{0,1\}^\ast$, $x\in \mathrm{TAUT}$, if and only if there exists $y\in \{0,1\}^\ast$ such that $R(x,y)$ holds. Given $x\in \mathrm{TAUT}$, any $y$ for which $R(x,y)$ holds is called an $R$-proof of $x$. A propositional proof system $R$ is polynomially bounded (\textbf{p-bounded}) if there exists a polynomial $p$ such that for each $x \in \mathrm{TAUT}$, there is an $R$-proof $y$ of $x$ of size at most $p(|x|)$ ($\ie \ |y| \leq p(|x|)$).
    \end{definition}

    \begin{definition}[p-simulation]
        \label{def: p-simulation}
        Let $P$ and $Q$ be propositional proof systems. We say that \emph{$P$ p-simulates $Q$}, written $Q \leq_p P$, if there exists a polynomial $p(\cdot)$ such that for every propositional tautology $\varphi$ and every $Q$-proof $\pi$ of $\varphi$ of size $s$, there exists a $P$-proof $\pi'$ of $\varphi$ whose size is at most $p(s)$.
    \end{definition}

    \begin{definition}[Frege, $\ACZ$-Frege and $\ACZP$-Frege]
        \label{def: Frege, AC-Frege and ACZP-Frege}
        A \emph{Frege rule} is an inference rule of the form: $B_1, \dots, B_n \implies B$, where $B_1, \dots, B_n, B$ are propositional formulas. If $n=0$, then the rule is an axiom. A \emph{Frege system} is specified by a finite set, $R$, of rules. Given a collection $R$ of rules, a derivation of a 3DNF formula $f$ is a sequence of formulas $f_1, \dots, f_m$ such that each $f_i$ is either an instance of an axiom scheme or follows from previous formulas by one of the rules in $R$ and such that the final formula $f_m$ is $f$.

        \emph{$\ACZ$-Frege} are Frege proofs but with the additional restriction that each formula in the proof has bounded depth.

        \emph{$\ACZP$-Frege} are bounded-depth Frege proofs that also allow unbounded fan-in $\mathrm{MOD}_p$ connectives, namely $\mathrm{MOD}_p^i$ for $i \in \{0, \dots, p-1\}$. $\mathrm{MOD}_p^i(x_1,\dots, x_k)$ evaluates to true if the number of $x_i$ that are true is congruent to $i$ $\mod(p)$ and evaluates to false otherwise.
    \end{definition}

    \begin{definition}[Polynomial Calculus \cite{CEI96}]
            \label{def: polynomial calculus}
            Given a field $\mathbb{F}$ and a set of variables, a \emph{polynomial calculus} ($\PC$) refutation of the set of axioms $P$ is a sequence of polynomials such that the last line is the polynomial $1$ and each line is either an axiom or is derived from the previous lines using the following inference rules:
            \begin{equation*}
                \frac{f \qquad g}{\alpha f + \beta g}
            \end{equation*}
            and
            \begin{equation*}
                \frac{f}{x\cdot f},
            \end{equation*}
            where $\alpha, \beta \in \mathbb{F}$ are any scalars and $x$ is an variable. The refutation has degree $d$ if all the polynomials in it have degrees at most $d$.

             The \emph{degree} of a $\PC$ proof is defined as the maximal degree of a polynomial appearing in it, and its \emph{size} is the number of different \emph{monomials} in this proof.
        \end{definition}
    
        \begin{definition}[Polynomial Calculus with Resolution \cite{MR1919962}]
            \label{def: PCR}
             Let $\mathbb{F}$ be a fixed field. \emph{Polynomial Calculus with Resolution} ($\PCR$) is the proof system whose lines are polynomials from $\mathbb{F}[x_1,\cdots,x_n,\overline{x_1},\cdots,\overline{x_n}]$, where $\overline{x_1},\cdots,\overline{x_n}$ are treated as new formal variables. $\PCR$ has all default axioms and inference rules of $\PC$ (including, of course, those that involve new variables $\overline{x_i}$), plus additional default axioms $x_i + \overline{x_i}=1$ $(i \in [n])$.
    
            For a clause $C$, denote by $\Gamma_C$ the monomial
            \begin{equation*}
                \Gamma_C \coloneqq \prod_{\overline{x} \in C} x \cdot \prod_{x \in C} \overline{x}
            \end{equation*}
            and for a $\CNF$ $\tau$, let $\Gamma_\tau \coloneqq \{\Gamma_C | C \in \tau \}$. (Note that $\tau$ is unsatisfiable if and only if the polynomials $\Gamma_\tau$ have no common root in $\mathbb{F}$ satisfying all default axioms of $\PCR$.) A \emph{$\PCR$ refutation} of a $\CNF$ $\tau$ is a $\PCR$ proof of the contradiction $1=0$ from $\Gamma_\tau$.
    
            The \emph{degree} of a $\PCR$ proof is defined as the maximal degree of a polynomial appearing in it, and its \emph{size} is the number of different \emph{monomials} in this proof.
        \end{definition}
    
        $\PC$ and $\PCR$ are equivalent with respect to the degree measure (via the linear transformation $\overline{x} \implies 1- x_i$).

     \subsection{Ideal Proof System}
    
     Given $f_1,\cdots, f_m \in \mathbb{F}[x_1, \cdots,x_n]$ over some field $\mathbb{F}$, Hilbert's Nullstellensatz shows that $f_1(\overline{x})=\cdots = f_m(\overline{x})=0$ is unsatisfiable (over the algebraic closure of $\mathbb{F}$) if and only if there are polynomials $g_1,\cdots,g_m \in \mathbb{F}[\overline{x}]$ such that $\sum_j g_j(\overline{x}) f_j(\overline{x})=1$ (as a formal identity), or equivalently, that 1 is in the ideal generated by the $\{f_j\}_j$.
    \begin{definition}[(Boolean) Ideal proof system ($\IPS$) \cite{GP18}] \label{def: Boolean ips}
        Let $f_1(\overline{x}),\cdots,f_m(\overline{x}), p(\overline{x})\in \mathbb{F}[x_1,\cdots,x_n]$ be a collection of polynomials. An $\IPS$ proof of $p(\overline{x})=0$ from $\{f_j(\overline{x})\}_{j=1}^m$, showing that $p(\overline{x})=0$ is semantically implied from the assumptions $\{f_j(\overline{x})=0\}_{j=1}^m$ over 0-1 assignments, is an algebraic circuit $C(\overline{x},\overline{y},\overline{z}) \in \mathbb{F}[\overline{x}, y_1,\cdots,y_m,z_1,\cdots,z_n]$, such that (the equalities in what follows stand for formal polynomial identities\footnote{That is, $C(\overline{x},\overline{0},\overline{0})$ computes the zero polynomial and $C(\overline{x},f_1(\overline{x}),\cdots,f_m(\overline{x}),x_1^2-x_1, \cdots, x_n^2 -x_n)$ computes the polynomial $p(\overline{x})$})
        \begin{enumerate}
            \item $C(\overline{x},\overline{0},\overline{0})=0$.
            \item $C(\overline{x},f_1(\overline{x}), \cdots, f_m(\overline{x}),x_1^2-x_1 , \cdots, x_n^2 -x_n)=p(\overline{x})$.
        \end{enumerate}
        
        The size of the $\IPS$ proof is the size of the circuit $C$. The variables $\overline{y},\overline{z}$ are sometimes called the placeholder variables since they are used as a placeholder for the axioms. An $\IPS$ proof $C(\overline{x},\overline{y},\overline{z})$ of $1=0$ from $\{f_j(\overline{x})=0\}_{j=1}^m$ is called an $\IPS$ refutation of $\{f_j(\overline{x})=0\}_{j=1}^m$. 
        If $C$ comes from a restricted class of algebraic circuits $\mathcal{C}$, then this is called a $\mathcal{C}$-$\IPS$ refutation.
    \end{definition}

    We shall also use the algebraic version of \IPS (which does not use the Boolean axioms):
    \begin{definition}[(Algebraic) Ideal proof system ($\IPS^{\alg}$) \cite{GP18}]
         Let $f_1(\overline{x}),\cdots,f_m(\overline{x}), p(\overline{x})\in \mathbb{F}[x_1,\cdots,x_n]$ be a collection of polynomials. An $\IPS^{\alg}$ proof of $p(\overline{x})=0$ from $\{f_j(\overline{x})\}_{j=1}^m$, showing that $p(\overline{x})=0$ is semantically implied from the assumptions $\{f_j(\overline{x})=0\}_{j=1}^m$ over assignments by field elements, is an algebraic circuit $C(\overline{x},\overline{y}) \in \mathbb{F}[\overline{x}, y_1,\cdots,y_m]$, such that 
        \begin{enumerate}
            \item $C(\overline{x},\overline{0})=0$.
            \item $C(\overline{x},f_1(\overline{x}), \cdots, f_m(\overline{x}))=p(\overline{x})$.
        \end{enumerate}
       The size and refutation are defined similarly to \Cref{def: Boolean ips}.
    \end{definition}

    Now, we introduce some notation we will use in the following sections. Let $\overline{\mathcal{F}}=\{f_i(\overline{x})=0\}_{i=1}^m$ be a collection of circuit equations, namely the $f_i$'s are written as algebraic circuits. We use $|\overline{\mathcal{F}}|$ to denote the total size of the circuit equations in $\overline{\mathcal{F}}$. We denote by $C: \overline{\mathcal{F}} \sststile{\IPS}{s,\Delta} 1=0$ the fact that $\overline{\mathcal{F}}$ has an $\IPS$ refutation $C$ of size at most $s$ and depth at most $\Delta$. If we do not care about the explicit size of the $\IPS$ refutation, we denote by $C: \overline{\mathcal{F}} \sststile{\IPS}{\ast, \Delta} 1=0$ the fact that $\overline{\mathcal{F}}$ has an $\IPS$ refutation $C$ of size polynomially bounded by $|\overline{\mathcal{F}}|$ and of depth $\Delta$.\par
    
    When we deal with algebraic $\IPS$ refutations, we will use the same notation as above, only using $\IPS^\alg$ instead of $\IPS$.

    Polynomial identities are proved for free in $\IPS$, which was observed in \cite{AGHT24}, and this also holds for constant-depth IPS proofs.

    \begin{proposition}
        \label{prop: polynomial identities that can be written as constant-depth circuits are proved for free in constant-depth IPS}
        If $C(\overline{x})$ is a Depth-$\Delta$ algebraic circuit in the variables $\overline{x}$ over the field $\mathbb{F}$ that computes the zero polynomial, then there is an Depth-$\Delta$ $\IPS$ proof of $C(\overline{x})=0$ of size $|F|$.
    \end{proposition}

    The following proposition can be regarded as a constant-depth analogue of Proposition A.5 in \cite{AGHT24}.

    \begin{proposition}[proof by boolean cases in bounded-depth $\IPS$]
        \label{prop: 0-1 implication completeness of IPS}
        Let $\F$ be a field. Let $\overline{\mathcal{F}}$ be a collection of $m$ many circuit equations over $n$ many variables $\overline{x}$. Assume that for every fixed assignment $\overline{\alpha} \in \{0,1\}^r$ where $0 \leq r \leq n$ we have 
        \[
            \sum_{i=1}^m G_i \cdot F_i + \sum_{i=1}^r L_i \cdot (x_i - \alpha_i) + \sum_{i=1}^n Q_i \cdot (x_i^2 -x_i)  = f(\overline{x})
        \]
        where each $G_i$, $L_i$ and $Q_i$ has size $s$ and depth-2 (in other words, each $G_i$, $L_i$, and $Q_i$ is just a summation of terms%
        \footnote{A term is a monomial multiplied by a field element.}), then there exist $G_i^\prime$ and $Q_i^\prime$ such that
        \[
            \sum_{i=1}^m G_i^\prime \cdot F_i + \sum_{i=1}^n Q_i^\prime \cdot (x_i^2 -x_i) = f(\overline{x})
        \]
        where each $G_i^\prime$, $Q_i^\prime$ has size $c^r \cdot s$ and depth-2 for some constant $c$ independent of $r$.
    \end{proposition}

    \begin{proof}
        We prove by induction on $r$.

        \emph{Base case:} $r=0$. Assume that 
        \[
            \sum_{i=1}^m G_i \cdot F_i + \sum_{i=1}^n Q_i \cdot (x_i^2 -x_i)  = f(\overline{x})
        \]
        where each $G_i$ and $Q_i$ has size $s$ and depth-2, then clearly there exist $G_i^\prime$ and $Q_i^\prime$ such that
        \[
            \sum_{i=1}^m G_i^\prime \cdot F_i + \sum_{i=1}^n Q_i^\prime \cdot (x_i^2 -x_i) = f(\overline{x})
        \]
        where each $G_i^\prime$, $Q_i^\prime$ has size $c^r \cdot s$ (which is $s$ in the base case) and depth-2 for some constant $c$ independent of $r$.

        \emph{induction step:} $r > 0$. Suppose for any assignment $\alpha \in \{0,1\}^r$, 
        \[
            \sum_{i=1}^m G_i \cdot F_i + \sum_{i=1}^r L_i \cdot (x_i - \alpha_i) + \sum_{i=1}^n Q_i \cdot (x_i^2 -x_i)  = f(\overline{x})
        \]
        where each $G_i$, $L_i$ and $Q_i$ has size $s$ and depth-2. We aim to show that there exist $G_i^\prime$ and $Q_i^\prime$ such that
        \[
            \sum_{i=1}^m G_i^\prime \cdot F_i + \sum_{i=1}^n Q_i^\prime \cdot (x_i^2 -x_i) = f(\overline{})
        \]
        where each $G_i^\prime$, $Q_i^\prime$ has size $c^r \cdot s$ and depth-2 for some constant $c$ independent of $r$.

        Then, by our assumption, we know that for every fixed assignment $\overline{\alpha} \in \{0,1\}^{r-1}$ we have:
        \begin{align}
            \sum_{i=1}^m G_i \cdot F_i + \sum_{i=2}^r L_i \cdot (x_i - \alpha_i) + M \cdot x_1 + \sum_{i=1}^n Q_i \cdot (x_i^2 - x_i) = f(\overline{x})\\
            \sum_{i=1}^m P_i \cdot F_i + \sum_{i=2}^r K_i \cdot (x_i - \alpha_i) + N \cdot (1-x_1) + \sum_{i=1}^n W_i \cdot (x_i^2 - x_i) = f(\overline{x})
        \end{align}
        where each $G_i$, $L_i$, $M$, $Q_i$, $P_i$, $K_i$, $N$ and $W_i$ is of size $s$ and depth $2$.
        By the induction hypothesis,
        \begin{align}
            \sum_{i=1}^m G_i^\prime \cdot F_i + M^\prime \cdot x_1 + \sum_{i=1}^n Q_i^\prime \cdot (x_i^2 -x_i) = f(\overline{x}) \label{eq: 01 implication completeness 1}\\
            \sum_{i=1}^m P_i^\prime \cdot F_i + N^\prime \cdot (1-x_1) + \sum_{i=1}^n W_i^\prime \cdot (x_i^2 - x_i) = f(\overline{x}) \label{eq: 01 implication completeness 2}
        \end{align}
        where each $G_i^\prime$, $M^\prime$, $Q_i^\prime$, $P_i^\prime$, $N^\prime$ and $W_i^\prime$ is of size $c^{r-1} \cdot s$ and depth $2$.
        
        By multiplying \Cref{eq: 01 implication completeness 1} and \Cref{eq: 01 implication completeness 2} with $1-x_1$ and $x_1$, respectively, we get
        \begin{align}
            \sum_{i=1}^m (1-x_1)\cdot G_i^\prime \cdot F_i + (1-x) \cdot M^\prime \cdot x_1 + \sum_{i=1}^n (1-x_1) Q_i^\prime \cdot (x_i^2 -x_i) = (1-x_1)\cdot f(\overline{x}) \label{eq: 01 implication completeness 3}\\
            \sum_{i=1}^m x_1 \cdot P_i^\prime \cdot F_i + x_1 \cdot N^\prime \cdot (1-x_1) + \sum_{i=1}^m x_1\cdot W_i^\prime \cdot (x_i^2 - x_i) = x_1\cdot f(\overline{x}) \label{eq: 01 implication completeness 4}
        \end{align}

        By summing \Cref{eq: 01 implication completeness 3} and \Cref{eq: 01 implication completeness 4}, we get
        \begin{align*}
            \sum_{i=1}^m [(1-x_1)  G_i^\prime + x_1  P_i^\prime] \cdot F_i + [ (1-x_1)  M^\prime + (1-x_1) Q_1^\prime + x_1 N^\prime + x_1 W_1^\prime] \cdot (x_1^2 - x_1) + \\
            \sum_{i=2}^n [(1-x_1) Q_i^\prime + x_1 W_i^\prime] \cdot (x_i^2 - x_i) = f(\overline{x})
        \end{align*}

        Note that both $(1-x_1)  G_i^\prime + x_1  P_i^\prime$, $(1-x_1)  M^\prime + (1-x_1) Q_1^\prime + x_1 N^\prime + x_1 W_1^\prime$ and $(1-x_1) Q_i^\prime + x_1 W_i^\prime$ can be computed by an algebraic circuit of size at most $6 \cdot c^{r-1} \cdot s \leq c^r \cdot s$ and depth $2$ for large enough $c$ independent with $r$. This concludes the proof of the proposition.
    \end{proof}

    The following theorem is from \cite{GP18}, and it already holds for $\IPS^\alg$.
    \begin{theorem}[Superpolynomial $\IPS$ lower bounds imply $\VNP \neq \VP$ \cite{GP18}]
        \label{theorem: IPS lower bound implies VNP neq VP}
        For any field $\mathbb{F}$, a superpolynomial lower bound on $\IPS^\alg$ (also $\IPS$) refutations over $\mathbb{F}$ for any family of $\CNF$ formulas implies $\VNP_\mathbb{F} \neq \VP_\mathbb{F}$. The same result holds if we assume that the $\IPS^\alg$ ($\IPS$) refutation size lower bound holds only infinitely often. 
    \end{theorem}
    \begin{lemma}[\cite{GP18}]
        \label{lemma: unsat CNF has VNP refutation}
        Every family of unsatisfiable $\CNF$ formulas $(\varphi_n)$ has a family of $\IPS^\alg$ (also $\IPS$) certificates $(C_n)$ in $\VNP_\mathbb{F}$.
    \end{lemma}

    \subsection{Encoding in Fixed Finite Fields}

    In the section, we are working in the finite field $\mathbb{F}_q$ where $q$ is a constant (independent of the size of the formulas and their number of variables). When we work with $\CNF$ formulas in $\IPS$ we assume that the $\CNF$ formulas are translated as follows:
    \begin{definition}[Algebraic translation of CNF formulas]
        \label{def: Algebraic translation of CNF formulas}
        Given a $\CNF$ formula in the variables $\overline{x}$, every clause $\bigvee_{i \in P}x_i \lor \bigvee_{j \in N} \neg x_j$ is translated into $\prod_{i \in P}(1-x_i) \cdot \prod_{j \in N}x_j=0$. (Note that these terms are written as algebraic circuits as displayed, where products are not multiplied out.)
    \end{definition}
    Notice that a $\CNF$ formula is satisfiable by 0-1 assignment if and only if the assignment satisfies all the equations in the algebraic translation of the $\CNF$.\par
    
    The following definitions are taken from \cite{ST25}, and we supply them here for completeness. 
    \begin{definition}[Algebraic extension axioms and unary bits \cite{ST25}] 
        \label{def: algebraic extension axioms and unary bits}
        Given a circuit $C$ and a node $g$ in $C$, we call the equation
        \begin{equation*}
            x_g = \sum_{i=0}^{q-1} i \cdot x_{g_i}
        \end{equation*}
        the algebraic extension axiom of $g$, with each variables $x_{g_i}$ being the $i$th unary-bit of $g$.
    \end{definition}
    
    \begin{definition}[Plain CNF encoding of bounded-depth algebraic circuit $\cnf(C(\overline{x}))$ \cite{ST25}]
        \label{def: plain CNF encoding of constant-depth algebraic circuit cnf(C(x))}
        Let $C(\overline{x})$ be a circuit in the variables $\overline{x}$. The plain $\CNF$ encoding of the circuit $C(\overline{x})$ denoted $\cnf(C(\overline{x}))$ consists of the following $\CNF$s in the unary-bit variables corresponding to   all the gates in $C$ and all the extra extension variables in \Cref{it:1914}:
        \begin{enumerate}
            \item If $x_i$ is an input node in $C$, the plain $\CNF$ encoding of $C$ uses the variables $x_{x_{i0}},\cdots,x_{x_{i(q-1)}}$ that are the unary-bits of $x_i$, and contains the clauses that express that precisely one unary-bit is $1$ and all other unary-bits are $0$:
            \begin{equation*}
                    \bigvee_{j=0}^{q-1} x_{x_i j} \land \bigwedge_{j \neq l \in \{0,\cdots,q-1\}} (\neg x_{x_i j} \lor \neg x_{x_i l}).
            \end{equation*}
            \item If $\alpha \in \mathbb{F}_q$ is a scalar input node in $C$, the plain $\CNF$ encoding of $C$ contains the $\{0,1\}$ constants  corresponding to the unary-bits of $\alpha$. These constants are used when fed to (translation of) gates according to the wiring of $C$ in item 4.
            \item \label{it:1914} For every node $g$ in $C(\overline{x})$ and every satisfying assignment $\overline{\alpha}$ to the plain $\CNF$ encoding, the corresponding unary-bit $x_{g_i}$ evaluates to $1$ if and only if the value of $g$ is $i \in \{0,\cdots, q-1\}$ (when the algebraic inputs $\overline{x} \in (\mathbb{F}_q)^\ast$ to $C(\overline{x})$ take on the values corresponding to the Boolean assignment $\overline{\alpha}$; "$\ast$" here means the Kleene star). This is ensured with the following equations: if $g$ is a $\circ \in \{+, \times\}$ node that has inputs $u_1,\cdots, u_t$. Then we consider the following equations:
                \begin{align*}
                     &u_1 \circ u_2 = v_1^g \\
                    &u_{i+2} \circ v_{i}^g = v_{i+1}^g, \qquad 1 \leq i \leq t-3 \\
                    &u_{t} \circ v_{t-2}^g = g.
                \end{align*}
                In other words, we add the \emph{extension variables} $v_i^g$ for each $+,\times$ gate, to sequentially compute the unbounded fan-in gate $g$ into a sequence of binary operations in the obvious way. 
                For simplicity, we denote each equation above by $x \circ y =z$. Then, for each $x\circ y=z$ we have a $\CNF$ $\phi$ in the unary-bits variables of $x, y,z$ that is satisfied by assignment precisely when the output unary-bits of $z$ get their correct values based on the (constant-size) truth table of $\circ$ over $\mathbb{F}_q$ and the input unary-bits of $x,y$ (we ensure that if more than one unary-bit is assigned $1$ in any of the unary-bits of  $x,y,z$ then the $\CNF$ is unsatisfiable). 
            \item For every unary-bit variable $x_{g_i}$, we have the Boolean axiom (recall we write these Boolean axioms explicitly since we are going to work with $\IPS^\alg$):
            \begin{equation*}
                x_{g_i}^2 - x_{g_i} = 0.
            \end{equation*}
        \end{enumerate}
        Therefore, we can see that the formula size of $\cnf(C(\overline{x})=0)$ is $\poly(q^2 \cdot |C|)$.
    \end{definition}

    Note that the only variables in a plain CNF encoding are unary-bit variables.
    
    \begin{definition}[Plain CNF encoding of a bounded-depth circuit equation $\cnf(C(\overline{x})=0)$ \cite{ST25}]
        \label{def: plain CNF encoding of a constant-depth circuit equation cnf(C(x)=0)}
         Let $C(\overline{x})$ be a circuit in the variables $\overline{x}$. The plain $\CNF$ encoding of the circuit equation $C(\overline{x})=0$ denoted $\cnf(C(\overline{x})=0)$ consists of the following $\CNF$ encoding from \Cref{def: plain CNF encoding of constant-depth algebraic circuit cnf(C(x))} in the unary-bits variables of all the gates in $C$ ( and only in the unary-bit variables), together with the equations:
         \begin{equation*}
                x_{g_{out}0} =1 \quad \text{and} \quad x_{g_{out}i} = 0, \quad \text{for all }i=1,\cdots,q-1,
            \end{equation*}
            which express that $g_{out} = 0$, where $g_{out}$ is the output node of $C$.
        
    \end{definition}

    \begin{definition}[Extended CNF encoding of a circuit equation (circuit, \resp); $\ecnf(C(\overline{x})=0)$ ($\ecnf(C(\overline{x}))$, \resp) \cite{ST25}]
        \label{def: Extended CNF encoding of a circuit equation cnf(C(x)=0)}
         Let $C(\overline{x})$ be a circuit in the variables $\overline{x}$ over the finite filed $\mathbb{F}_q$. The extended $\CNF$ encoding of the circuit equation $C(\overline{x})=0$ (circuit $C(\overline{x})$, \resp), in symbols $\ecnf(C(\overline{x})=0)$ ($\ecnf(C(\overline{x}))$, \resp), is defined to be a set of algebraic equations over $\mathbb{F}_q$ in the variables $x_g$ and $x_{g0},\cdots,x_{gq-1}$ which are the unary-bit variables corresponding to the node $g$ in $C$, that consist of:
         \begin{enumerate}
             \item the plain $\CNF$ encoding of the circuit equation $C(\overline{x})=0$ (circuit $C(\overline{x})$, \resp), namely, $\cnf(C(\overline{x})=0)$ ($\cnf(C(\overline{x}))$, \resp); and
             \item the algebraic extension axiom of $g$, for every gate $g$ in $C$.
         \end{enumerate}
         
    \end{definition}

    Since we work with extension variables for each gate in a given circuit equation $C(\overline{x})=0$, it is more convenient to express circuit equations as a set of equations that correspond to the straight line program of $C(\overline{x})$ (which is equivalent in strength formulation to algebraic circuits):
    
    \begin{definition}[Straight line program ($\SLP$)]
        \label{def: Straight line program SLP}
        An $\SLP$ of a circuit $C(\overline{x})$, denoted by $\SLP(C(\overline{x}))$, is a sequence of equations between variables such that the extension variable for the output node computes the value of the circuit assuming all equations hold. Formally, we choose any topological order $g_1,g_2, \cdots,g_i,\cdots,g_{|C|}$ on the nodes of the circuit $C$ (that is, if $g_j$ has a directed path to $g_k$ in $C$ then $j < k$) and define the following set of equations to be the $\SLP$ of $C(\overline{x})$:
        \begin{center}
            $g_i = g_{j1} \circ g_{j2} \circ \cdots \circ g_{jt}$ for $\circ \in \{+, \times \}$ iff $g_i$ is a $\circ$ node in $C$ with $t$ incoming edges from $g_{j1},\cdots,g_{jt}$.
        \end{center}
        An $\SLP$ representation of a circuit equation $C(\overline{x})=0$ means that we add to the $\SLP$ above the equation $g_{|C|}=0$, where $g_{|C|}$ is the output node of the circuit.
    \end{definition}
    The below lemma, which we refer to as the translation lemma in this paper, shows that we can derive the circuit equations from the extended $\CNF$ formulas encoding those circuit equations with some additional axioms and vice versa.
    \begin{lemma}[Translating between extended $\CNF$s and circuit equations in fixed finite fields \cite{ST25}] 
        \label{lemma: Translating between extended CNF formulas and circuit equations in Fixed finite fields}
        Let $\mathbb{F}_q$ be a finite field, and let $C(\overline{x})$ be a circuit of depth $\Delta$ in the $\overline{x}$ variables over $\mathbb{F}_q$. Then, the following both hold
        \begin{align}
            \ecnf(C(\overline{x})=0) \sststile{\IPS^\alg}{*,O(\Delta)} C(\overline{x})=0  
        \end{align}

        \begin{equation}
            \begin{array}{rcl}
            \left \{x_{g} = \sum_{i=0}^{q-1} i \cdot x_{gi} : g \text{ a node in } C \right \}, & & \\
            \left \{x_{gi}^2 - x_{gi} = 0 : g \text{ is a node in } C,\; 0 \le i < q\right \}, & 
             \scalebox{1.5}{$\sststile{\IPS^\alg}{*,O(\Delta)}$} &
            \mathrm{ecnf}\big(C(\overline{x}) = 0\big), \\[2mm]
        \left \{\sum_{i=0}^{q-1} x_{gi} = 1 : g \text{ is a node in } C\right \},\; C(\overline{x}) = 0, \; \SLP(C(\overline{x})) & &
            \end{array}
        \end{equation}

    \end{lemma}

    \begin{proposition}[Proposition 3.7 in \cite{ST25}]
        \label{prop: C <-> cnf <-> ecnf fixed finite fields}
        Let $C(\overline{x})=0$ be a circuit equation over $\mathbb{F}_q$ where $q$ is any constant prime. Then, $C(\overline{x})=0$ is unsatisfiable over $\mathbb{F}_q$ iff $\cnf(C(\overline{x})=0)$ is an unsatisfiable $\CNF$ iff $\ecnf(C(\overline{x})=0)$ is an unsatisfiable set of equations over $\mathbb{F}_q$.
    \end{proposition}

    Using results in \cite{EGLT25}, we could remove some extension axioms used in \cite{ST25} when working over fixed finite fields.
    We use $\UBIT_j(x)$ to denote the following Lagrange polynomial:
    \begin{equation}
        \label{UBITs}
        \UBIT_j(x):=
        \frac{\prod_{i=0, i \neq j}^{q-1}(x-i)}{\prod_{i=0, i \neq j}^{q-1}(j-i)}
    \end{equation}
    where $x$ can be a single variable or an algebraic circuit. Hence, it is easy to observe that
    \begin{equation*}
        \UBIT_j(x) = \begin{cases}
            1, \quad x = j,\\
            0, \quad \text{otherwise}.
        \end{cases}
    \end{equation*}
    Also, suppose $x$ has size $|x|$ and depth $\depth(x)$ (when $x$ is a single variable, it has size $1$ and depth $1$), $\UBIT_j(x)$ can be computed by an algebraic circuit of size $O(|x|^{q-1})$ and depth $\depth(x)+2$.

    \begin{definition}[Semi-CNF $\SCNF$ encoding of a bounded-depth circuit equation $\SCNF(C(\overline{x})=0)$]
        \label{def: semi-CNF encoding}
        Let $C(\overline{x})$ be a circuit in the variables $\overline{x}$. The semi-CNF encoding of the circuit equation $C(\overline{x})=0$ denoted $\SCNF(C(\overline{x}))$ is a substitution instance of the plain CNF encoding in \Cref{def: plain CNF encoding of a constant-depth circuit equation cnf(C(x)=0)} where each unary-bits $x_{uj}$ of all the gates and extra extension variables\footnote{These extension variables are used in Item 3 of \Cref{def: plain CNF encoding of constant-depth algebraic circuit cnf(C(x))} to help encode the circuit.} $u$ is substituted with $\UBIT_j(C_u)$ where $C_u$ is the bounded-depth algebraic circuit computed by $u$.\footnote{This $C_u$ can be constructed from $\SLP$s easily.}
    \end{definition}

    We call $x^q -x =0$ the \emph{field axiom} for the variable $x$.

    \begin{lemma}[Translate semi-CNFs from circuit equations in Fixed Finite Fields, \cite{EGLT25}]
        \label{lemma: Translate semi-CNFs from circuit equations in Fixed Finite Fields}
        Let $\mathbb{F}_q$ be a finite field, and let $C(\overline{x})$ be a circuit of depth $\Delta$ in the $\overline{x}$ variables over $\mathbb{F}_q$. Then, the following hold
        \begin{gather*}
            \{x^q -x =0: x \text{ is a variable in } C\}, C(\overline{x})=0  \sststile{\IPS^\alg}{*,O(\Delta)} \SCNF(C(\overline{x})=0)
        \end{gather*}
        
    \end{lemma}

    \begin{lemma}[Translate circuit equations from semi-CNFs in fixed finite fields, \cite{EGLT25}]
        \label{lemma: Translate circuit equations from semi-CNFs in Fixed Finite Fields}
        Let $\mathbb{F}_q$ be a finite field, and let $C(\overline{x})$ be a circuit of depth $\Delta$ in the $\overline{x}$ variables over $\mathbb{F}_q$. Then, the following hold
        \begin{gather*}
            \{x^q -x =0: x \text{ is a variable in } C\}, \SCNF(C(\overline{x})=0)  \sststile{\IPS^\alg}{*,O(\Delta)} C(\overline{x})=0
        \end{gather*}
    \end{lemma}

    We will use our new translation lemma for the next section, which is our main result in fixed finite fields. For polynomial-size finite fields, we have a different translation lemma that we will explain later.

    \begin{proposition}
        \label{prop: C <-> scnf fixed finite fields}
        Let $C(\overline{x})=0$ be a circuit equation over $\mathbb{F}_q$ where $q$ is any constant prime. Then, $C(\overline{x})=0$ is unsatisfiable over $\mathbb{F}_q$ iff $\scnf(C(\overline{x})=0)$ is an unsatisfiable $\SCNF$. 
    \end{proposition}

    \section{Universal Algebraic Circuits for Bounded-Depth}

    Here, we develop the necessary technical information regarding universal circuits. This is a novel adaptation of the work of Raz \cite{Raz10} to the bounded-depth setting, in which both the universal circuit and the circuits it encodes are of bounded depth.

    In this section, we will work with algebraic circuits whose edges can be labelled by field elements. This does not make too much difference, as we can easily replace them with a multiplication, which only increases the depth and size of a circuit up to a factor of 2.
    
    For general algebraic circuits, we have the following. 

    \begin{theorem}[Existence of universal circuits for homogeneous polynomials \cite{Raz10}]
        \label{theorem: Existence of universal circuits for homogeneous polynomials}
        Let $\mathbb{F}$ be a field and $\overline{x}$ be $n$ variables, and let $C_{s,d}^{\hom}$ denote the class of all homogeneous polynomials of total degree exactly $d$ in $\mathbb{F}[\overline{x}]$ that have algebraic circuits of size at most $s$. Then there is a circuit $U(\overline{x},\overline{w}) \in \mathbb{F}[\overline{x},\overline{w}]$ of size $O(d^2 s^8)$ and syntactic-degree $O(d)$ such that $\overline{w}$ are $K_{s,d} = O(d^2 s^8)$ many variables which are disjoint from $\overline{x}$, that is universal for $C_{s,d}^{\hom}$ in the following sense: if $f(\overline{x}) \in C_{s,d}^{\hom}$, then there exists $\overline{a} \in \mathbb{F}^{K_{s,d}}$ such that $U(\overline{x},\overline{a})=f(\overline{x})$. Notice that given $s$ and $d$, $K_{s,d}$ can be computed efficiently.
    \end{theorem}

    We use the following simple adaptation from \cite{ST25}:
    \begin{definition}[Universal circuits for polynomials \cite{ST25}]  
        The universal circuit for degree $d$ and size $s$ circuits is defined as:
        \begin{equation*}
            U(\overline{x},\overline{w})=\sum_{i=0}^d U_i(\overline{x},\overline{w}),
        \end{equation*}
        where $U_i(\overline{x},\overline{w})$ is the universal circuit for homogeneous $\overline{x}$-polynomials of $i$ degree $C_{s,i}^{\hom}$ and where the $\overline{w}$-variables in each distinct $U_i(\overline{x},\overline{w})$ are pairwise disjoint.
    \end{definition}
    The size of $U(\overline{x},\overline{w})$ is $\sum_{i=0}^d O(i^4s^8)=O(d^5s^8)$.\par
    \begin{definition}[Circuit-graph \cite{Raz10}]
        Let $\Phi$ be an algebraic circuit. We denote by $G_\Phi$ the underlying graph of $\Phi$, together with the labels of all nodes. That is, the entire circuit, except for the labels of the edges. We call $G_\Phi$, the circuit-graph of $\Phi$.
    \end{definition}

    We will need universal circuits for bounded-depth circuits, where the universal circuits are bounded-depth themselves. We say a circuit-graph $G$ is depth-$\Delta$ if $\depth(G) \leq \Delta$.
    \begin{definition}[Normal-Depth-Form]
        Let $G$ be a depth-$\Delta$ circuit-graph. We say that $G$ is in Normal-Depth-Form if it satisfies:
        \begin{enumerate}
            \item All edges from the leaves are to $+$ nodes.
            \item All output-nodes are $+$ nodes.
            \item The nodes of $G$ are alternating. That is, if $v$ is a $+$ node and $(u,v)$ is an edge, then $u$ is either a leaf or a $\times$ node and if $v$ is a $\times$ node and $(u,v)$ is an edge then $u$ is a $+$ node.
            \item The out-degree of every node is at most $1$.
            \item The depth of every leaf is the same.
            
        \end{enumerate}
        We say that an algebraic circuit $\Phi$ is in \emph{normal-depth-form} if the circuit-graph $G_\Phi$ is in normal-depth-form.
    \end{definition}

    \begin{lemma}[Existence of normal-depth-form algebraic circuits for bounded-depth algebraic circuits]
        \label{lemma: existence of n-d-2d-f}
        Let $\mathbb{F}$ be a field and $\Delta$ be a constant. Let $\Phi$ be a depth-$\Delta$ algebraic circuit of size $s$ for a polynomial $g \in \mathbb{F}[x_1,\cdots,x_n]$. Then, there exists an algebraic circuit $\Phi^\prime$ that computes $g$ such that $\Phi^\prime$ is a normal-depth-form, and the number of nodes in $\Phi^\prime$ is $\poly(s)$. Moreover, given $\Phi$ (as an input), $\Phi^\prime$ can be efficiently constructed.
    \end{lemma}
     \begin{proof}
         First, we turn our depth-$\Delta$ algebraic circuit $\Phi$ into a depth-$\Delta$ algebraic formula $\varphi$ of size $\poly(s \cdot 2^\Delta)$. Since $\Delta$ is a constant, $\varphi$ is of size $\poly(s)$.
        
         Then, by merging nodes and adding dummy nodes ($+$ nodes or $\times$ nodes such that only have one input and one output), we can construct $\Phi^\prime$ from $\varphi$. To be specific, for an edge $(u,v)$ in $\varphi$, if both $u$ and $v$ are $+$ nodes (\resp, $\times$ nodes), we merge them to one $+$ node (\resp, one $\times$ node). Then, if $(u,v)$ is an edge after merging $u$ is a $\times$ node, and $v$ is a leaf, then we add a $+$ node $o$ between $u$ and $v$ such that $(u,o)$ and $(o,v)$ are labelled with $1$. If the output node $u$ is a $\times$ node, we add a $+$ node $v$ above it and label $(u,v)$ with $1$. Also, if the depth of a leaf is smaller than the maximum depth, by alternatively adding dummy $+$ nodes and $\times$ nodes, we can make the depth of every leaf the same.
         
         We can see that $\Phi^\prime$ has bounded-depth and is of size $\poly(s)$. $\Phi^\prime$ is in normal-depth-form.
     \end{proof}

    \begin{remark}
        Note that if we consider the universal circuit for bounded-depth circuits with size polynomially bounded by the number of variables, we can further bound the maximum fan-in of multiplication gates in the normal-depth-form. 
        We bound the multiplication fan-in of $\varphi$ by replacing each $\times$ node with polynomial many $\times$ nodes. This can be achieved since the size of the circuit is polynomially bounded. After this replacement, $\varphi$ is a bounded-depth algebraic formula with bounded multiplication fan-in of size $\poly(s)$.
    \end{remark}

     \begin{theorem}[Existence of bounded-depth universal circuits for polynomials computed by bounded-depth circuits]
        \label{theorem: Existence of universal circuits for constant-depth polynomials}
         Let $\mathbb{F}$ be a field and $\overline{x}$ be $n$ variables, and let $\VAC^0_{s,\Delta}$ denote the class of all polynomials in $\mathbb{F}[\overline{x}]$ that have algebraic circuits of size at most $s$ and depth at most $\Delta$, where $\Delta$ is a constant. Then there is a circuit $U(\overline{x},\overline{w}) \in \mathbb{F}[\overline{x},\overline{w}]$ of size $\poly(s)$ and depth $\poly(\Delta)$ such that $\overline{w}$ are $K_{s,d}$ variables that are disjoint from $\overline{x}$, that is universal for $\VAC^0_{s,\Delta}$ in the following sense: if $f(\overline{x}) \in \VAC^0_{s,\Delta}$, then there exists $\overline{a} \in \mathbb{F}^{K_{s,d}}$ such that $U(\overline{x},\overline{a})=f(\overline{x})$. Notice that given $s$ and $d$, $K_{s,d}$ can be computed efficiently and is bounded by $\poly(s)$. Also, the universal circuit preserves the maximum multiplication fan-in.
     \end{theorem}
     \begin{proof}
         Let $\mathbb{F}$ be a field. A polynomial $g \in \mathbb{F}[\overline{x}]$ is computed by an algebraic circuit of size $s$ and depth at most $\Delta$ where $\Delta$ is a constant. Then, by \Cref{lemma: existence of n-d-2d-f}, there exists a bounded-depth algebraic circuit $\Phi^\prime$ that computes $g$ such that $\Phi^\prime$ is in normal-depth-form of size $\poly(s)$ and depth $2\Delta^\prime-1$, which is a constant. Let $t$ be the maximum multiplication fan-in in $\Phi^\prime$.
         
         We partition the nodes in $\Phi^\prime$ into $2\Delta^\prime$ levels by their depths as follows:
         \begin{itemize}
             \item For every $i \in \{1,\cdots, \Delta^\prime-1 \}$, level $2i$ contains the $\times$ nodes where each of them has a length $2i-1$ path to the output node.
             \item For every $i \in \{1. \cdots, \Delta^\prime\}$, level $2i-1$ contains the $+$ nodes where each of them has a length $2i-2$ path to the output node.
             \item Level $2\Delta^\prime$ contains all the leaves.
         \end{itemize}
         Now, we construct the universal circuit of depth $2\Delta^\prime-1$ as follows:
         \begin{itemize}
             \item For every $+$ node in level $2i-1$ ($i \in \{1,\cdots, \Delta^\prime\}$), the children of it are all the nodes in level $2i$. There is only one $+$ node in level $1$ as the output node.
             \item For every $\times$ node in level $2i$ ($i \in \{1. \cdots, \Delta^\prime-1\}$), the children of it are $+$ nodes in level $2i+1$. All the $\times$ nodes in level $2i$ ($i \in \{1, \cdots, \Delta^\prime-1\}$) are partitioned into $t$ many groups, where $\times$ nodes in group $j$ ($1 \leq j \leq t$) has $j$ many children in level $2i$. Also, all $\times$ nodes in the same group have distinct children.
         \end{itemize}
         Since the out-degree of every $+$ node in $\Phi^\prime$ is at most $1$, it is sufficient to have $\Size(\Phi^\prime)/k$ many $\times$ nodes with in-degree $k$ in the above level. Therefore, there are at most $\Size(\Phi^\prime)+\lceil \frac{\Size(\Phi^\prime)}{2} \rceil+ \lceil \frac{\Size(\Phi^\prime)}{3} \rceil+ \cdots + \lceil \frac{\Size(\Phi^\prime)}{t} \rceil \leq \Size(\Phi^\prime)^2$ many $\times$ nodes in the above level. Therefore, there are $O(\Size(\Phi^\prime)^3)$ many edges between $+$ nodes in level $2i$ and $\times$ nodes in level $2i+1$.

        Hence, we get a depth $2\Delta^\prime-1$ universal circuit of size $\poly(\Size(\Phi^\prime)) = \poly(s)$. Also, note that such universal circuit has the same maximum multiplication fan-in as $\Phi^\prime$.
     \end{proof}

\section{Extracting Coefficients and $\IPS$ Refutation Formula}
    Let $f(\overline{x},\overline{w}) \in \mathbb{F}[\overline{x}, \overline{w}]$ be a polynomial, and let $M=\prod_{i \in I} x_i^{\alpha_i} \cdot \prod_{j \in J}w_j^{\beta_j}$ be a monomial in $f(\overline{x},\overline{w})$, for some $\alpha_i,\beta_i \in \mathbb{N}$ (where $ 0 \in \mathbb{N}$). Then, we call $\Sigma_{i \in I}\alpha_i$ the $\overline{x}$-degree of $M$.
    \begin{definition}[$\coeff_M(\cdot)$ \cite{ST25}]
        \label{def: coefficient circuit}
        Let $f(\overline{x},\overline{w})$ be a polynomial in $\mathbb{F}[\overline{x},\overline{w}]$ in the disjoint sets of variables $\overline{x}, \overline{w}$. Let $M$ be an $\overline{x}$-monomial of degree $j$. Then, $\coeff_M(f(\overline{x},\overline{w}))$ is the (polynomial) coefficient in $\mathbb{F}[\overline{w}]$ (that is, in the $\overline{w}$-variables only) of $M$ in $f(\overline{x},\overline{w})$.
    \end{definition}
    Note that $f(\overline{x},\overline{w})=\Sigma_{M_i} M_i \cdot \coeff_{M_i}(f(\overline{x},\overline{w}))$, where the $M_i$'s are all possible $\overline{x}$-monomials of degree at most $d$, for $d$ the maximal $\overline{x}$-degree of a monomial in $f(\overline{x},\overline{w})$.

    \begin{proposition}[Computation of coefficients in general circuits \cite{ST25}]
        \label{prop: computation of coefficients in general circuits}
        Let $f(\overline{x},\overline{w}) \in \mathbb{F}[\overline{x},\overline{w}]$ be a polynomial in $\mathbb{F}[\overline{x},\overline{w}]$ in the disjoint sets of variables $\overline{x},\overline{w}$. Suppose that $M$ is an $\overline{x}$-monomial of degree $d$, and assume that there is an algebraic circuit computing $f(\overline{x},\overline{w})$ of size $s$ and syntactic-degree $l$. Then, there is a circuit of size $O(7^d\cdot s)$ computing $\coeff_M(f(\overline{x},\overline{w}))$ of syntactic-degree $l^{O(1)}$.
    \end{proposition}
    
    While \cite{ST25} gave the above proposition about the computation of the coefficient of an $\overline{x}$-monomial in general algebraic circuits, we present the computation of the coefficient of an  $\overline{x}$-monomial in \emph{bounded-depth} circuits.

    Since we can decrease the maximum multiplication fan-in to $n$ in each polynomial-size bounded-depth circuit with at most a polynomial-size blow up and depth blowing up by at most a constant factor, we can assume that the maximum multiplication fan-in in each polynomial-size bounded-depth circuit is $n$.
    
    \begin{proposition}[Computation of coefficients in bounded-depth circuits] 
        \label{prop: computation of coefficients in constant-depth circuits}
        Let $f(\overline{x},\overline{w}) \in \mathbb{F}[\overline{x},\overline{w}]$ be a polynomial in $\mathbb{F}[\overline{x},\overline{w}]$ in the disjoint sets of variables $\overline{x},\overline{w}$. Suppose that $M$ is an $\overline{x}$-monomial of degree $d$, and assume that there is an algebraic circuit $C(\overline{x},\overline{w})$ computing $f(\overline{x},\overline{w})$ of maximum multiplication fan-in $t$, size $s$, syntactic-degree $l$ and depth $\Delta$ where $\Delta$ is a constant such that
        \begin{enumerate}
            \item All edges from the leaves are to $+$ nodes.
            \item All output-nodes are $+$ nodes.
            \item The nodes of $G$ are alternating. That is, if $v$ is a $+$ node and $(u,v)$ is an edge, then $u$ is a $\times$ node, and if $v$ is a $\times$ node and $(u,v)$ is an edge then $u$ is either a leaf or a $+$ node.
        \end{enumerate}
        Then, there is a bounded-depth algebraic circuit of depth $\Delta$, size $O(2^{(t+d)d} \cdot s)$ computing $\coeff_M(f(\overline{x},\overline{w}))$ of syntactic-degree $l^{O(1)}$.
    \end{proposition}
    \begin{proof}
        For a variable $x_i$, we show how to construct a circuit, which is in the same depth as $C(\overline{x},\overline{w})$, computing a polynomial $g(\overline{x},\overline{w})$ such that $f=x_i \cdot g + h$ with $h$ having no occurrences of $x_i$. Then, using $d$ such iterations for each of the $d$ variables in $M$, we shall get the circuit $D$ that computes the coefficient of $M$ in $f(\overline{x},\overline{w})$ of the same depth. Then, by assigning zeros to all $\overline{x}$-variables in $D$, we can eliminate all the monomials in $D$ in both $\overline{x}$ and $\overline{w}$ variables. Now, we prove the following claim. The following claim shows how to construct the circuit that extracts the coefficient of a single variable. To construct the circuit that extracts the coefficient of a monomial of degree at most $d$, we apply $d$ iterations of the following claim.  
        We denote by $C(\overline{x},\overline{w}) \restriction _{x_i = 0}$ the polynomial $C(\overline{x},\overline{w})$ where $x_i$ is assigned $0$.\par
        \bigskip
        \begin{claim}\label{claim:01} Let $C(\overline{x},\overline{w})$ be an algebraic circuit over the field $\mathbb{F}$ of maximum addition fan-in $t_1$, maximum multiplication fan-in $t_2$, syntactic-degree $l$ and depth $\Delta$ such that 
        \begin{enumerate}
            \item All edges from the leaves are to $+$ nodes.\label{it:2127:1}
            \item The output node is a $+$ node. \label{it:2127:2}
            \item The nodes of $G$ are alternating. That is, if $v$ is a $+$ node and $(u,v)$ is an edge, then $u$ is a $\times$ node, and if $v$ is a $\times$ node and $(u,v)$ is an edge then $u$ is either a leaf or a $+$ node.\label{it:2127:3}
        \end{enumerate}
        Then, for every variables $x_i$, there is a depth-$\Delta$ circuit in the same form as $C$ (i.e., with \Cref{it:2127:1}-\Cref{it:2127:3} holding) of maximum addition fan-in $t_1\cdot(2^{t_2}-1)$, maximum multiplication fan-in $t_2+1$, size $O(2^{t_2}|C|)$ and syntactic-degree $l^{O(1)}$ that computes the polynomial $g(\overline{x},\overline{w})$, such that $C(\overline{x}) = x_i \cdot g(\overline{x},\overline{w}) + C(\overline{x},\overline{w}) \restriction_{x_i = 0}$.
        \end{claim}

        \noindent\textit{Proof of claim.}
        The proof is obtained by induction on circuit size. Denote by $p$ the polynomial computed by $C$ and for every gate $v$ in $C$ denote by $p_v$ the polynomial computed at gate $v$.\par
        Denote by $P_{x_i}(p_v)$ the unique polynomial such that $p_v = x_i \cdot P_{x_i}(p_v) + p_v \restriction_{x_i = 0}$. For a $\times$ gate $v$ with fan-in $t$ in $C$, we add at most $2^t$ new gates. Each gate $v$ itself is duplicated twice so that the first duplicate computes $P_{x_i}(p_v)$ and the second duplicate computes $p_v\restriction_{x_i = 0}$.\par
        Base case:\par
        \textbf{Case 1:} $p_v = x_i$. Then,  $P_{x_i}(p_v) := 1$ and $p_v\restriction_{x_i = 0} := 0$.\par
        \textbf{Case 2:} $p_v = x_j$, for $j \neq i$. Then, $P_{x_i}(p_v) \coloneqq 0$ and $p_v\restriction_{x_i= 0} \coloneqq x_j$.\par
        \textbf{Case 3:} $p_v = \alpha$, for $\alpha \in\mathbb{F}$. Then, $P_{x_i}(p_v) = 0$ and $p_v\restriction_{x_i = 0} = \alpha$.\par
        Induction step:\par
        \textbf{Case 1}: $p_v = \Sigma_{j=1}^{t_1} u_j$. Then, $P_{x_i}(p_v) =  \Sigma_{j=1}^{t_1} P_{x_i}(p_{u_j}) $ and $p_v \restriction_{x_i = 0} = \Sigma_{j=1}^{t_1} p_{u_j}\restriction_{x_i =0}$.\par
        \textbf{Case 2}: $p_v = \prod_{j=1}^{t_2} u_j$. Then, 
        \begin{equation*}
            P_{x_i}(p_v) = \sum_{j=1}^{t_2} \left ( \sum_{\substack{S \subseteq [t_2] \\ |S| = k} } \prod_{k \in S } P_{x_i}(p_{u_k}  )  \prod_{l \notin S} p_{u_l}\restriction_{x_i =0}\right) x_i^{j-1}
        \end{equation*}
        which is equivalent to
        $
            {\prod_{j=1}^{t_2} (x_i P_{x_i}(p_{u_j}) + p_{u_j} \restriction_{x_i = 0}) - \prod_{j=1}^{t_2} p_{u_j} \restriction_{x_i =0}}
        $
        ``divided by $x_i$'', namely when we decrease by $1$ the power of every $x_i^b$ in every monomial in this polynomial  (noting that $x_i$ appears with a positive power $b\ge 1$ in every monomial),
        and
        \begin{equation*}
            p_v \restriction_{x_i = 0} = \prod_{j=1}^{t_2} p_{u_j}\restriction_{x_i = 0}.
        \end{equation*}
        Note that by expanding  brackets, $P_{x_i}(p_v)$ is written explicitly with $2^{t_2}-1$ many terms as depth-2 circuits that have one $+$ node at the top and $2^{t_2}-1$ many $\times$ nodes as children.\par
        Then, by merging the $+$ node in $P_{x_i}(p_v) =  \Sigma_{j=1}^{t_1} P_{x_i}(p_{u_j}) $ and $ P_{x_i}(p_v) = \sum_{j=1}^{t_2} \left ( \sum_{\substack{S \subseteq [t_2] \\ |S| = k} } \prod_{k \in S } P_{x_i}(p_{u_k}  )  \prod_{l \notin S} p_{u_l}\restriction_{x_i =0}\right) x_i^{j-1}$ 
        (Since this formula is written explicitly and the circuit is alternating, which means there is always a $+$ node above any $\times$ node), our circuit has the same depth as $C$ and is in the following form:
         \begin{enumerate}
            \item All edges from the leaves are to $+$ nodes.
            \item The output nodes is a $+$ node.
            \item The nodes of $G$ are alternating. That is, if $v$ is a $+$ node and $(u,v)$ is an edge, then $u$ is a $\times$ node, and if $v$ is a $\times$ node and $(u,v)$ is an edge then $u$ is either a leaf or a $+$ node.
        \end{enumerate}
        Moreover, by computing the $x_i^2, x_i^3, \cdots, x_i^{t_2}$ at the bottom of the circuit using a trivial depth-2 circuit with maximum multiplication fan-in $t_2$, the maximum addition fan-in is $t_1\cdot (2^{t_2}-1)$ and the maximum multiplication fan-in is $t_2+1$. The size of the circuit after one iteration is $2\cdot (2^{t_2}-1) \cdot |C|$. \qed$_\text{claim}$
    
    \medskip
     As we showed above, we can construct a circuit of the same depth as $C$ that extracts the coefficient of a single variable in size $2\cdot (2^{t_2}-1) \cdot |C|$. To construct a circuit of the same depth as $C$ that extracts the coefficient of a monomial of degree at most $d$, we just need to do $d$ iterations of the above claim. Therefore, after $d$ iterations, the size of the circuits after the last iteration is $|C|\cdot \prod_{i=0}^{d-1} (2^{t_2+i}-1) \cdot 2^d=O(2^{(t_2+d)d}|C|)$ and the syntactic-degree is $l^{O(1)}$.
    This concludes the proof of \Cref{prop: computation of coefficients in constant-depth circuits}.
    \end{proof} 

    \begin{definition}[Bounded-depth $\IPS$ refutation predicate $\IPS_{\refute}(s,\Delta,l,\overline{\mathcal{F}})$]
        \label{def: Constant-depth IPS refutation predicate}
         Let $\overline{\mathcal{F}}$ be a $\CNF$ formula with m clauses and n variables $\overline{x}$ written as a set of polynomial equations according to \Cref{def: Algebraic translation of CNF formulas}. Let $U(\overline{x},\overline{y},\overline{w})$ be the bounded-depth universal circuit for depth $\Delta$ and size $s$ circuits in the $\overline{x}$ variables and the $m$ placeholder variables $\overline{y}$, and the $K_{s,\Delta}$ edge label variables $\overline{w}$. We formalize the existence of a size $s$, depth $\Delta$ circuit that computes the $\IPS$ refutation of $\overline{\mathcal{F}}$ in degree at most $l$, denoted $\IPS_{\refute}(s,\Delta,l,\overline{\mathcal{F}})$, with the following set of circuit equations (in the $\overline{w}$ variables only):
            \begin{gather*}
                \coeff_{M_i}(U(\overline{x},\overline{0},\overline{w})) = 0\\
                \coeff_{M_i}(U(\overline{x},\overline{\mathcal{F}},\overline{w}))= 
                    \begin{cases}
                        1,\ M_i =1  \ (\ie, \text{the constant 1 monomial});\\
                        0,\ \text{otherwise},
                    \end{cases}
            \end{gather*}
        where $i \in [N]$ so that $\{M_i\}_{i=1}^N$ are the set of all possible $\overline{x}$-monomials of degree at most $l$, and $N=\sum_{j=0}^l \binom{n+j-1}{j}=2^{O(n+l)}$ is the number of monomials of total degree at most $l$ over n variables, and $\overline{0}$ is the all-zero vector of length m.

        The size of $\IPS_{\refute}(s,\Delta,l,\overline{\mathcal{F}})$ is $O(2^{(t+l)l}\cdot |U(\overline{x},\overline{y},\overline{w})| \cdot |\overline{\mathcal{F}}| \cdot N)$ where $t$ is maximum multiplication fan-in in $U(\overline{x},\overline{y},\overline{w})$.
    \end{definition}

    \begin{definition}[Formalization of $\VNP = \VAC^0$]
        \label{def: formalization of VNP = c.d. poly-size} The formalization of $\VNP=\VAC^0(n,s,l,\Delta)$ denoted ``\,$\VNP=\VAC^0(n,s,l,\Delta)$", expressing that there is a bounded-depth universal circuit for size $s$ and depth $\Delta$ circuits that compute the Permanent polynomial of dimension $n$ (with $\overline{x}$ being the $n^2$ variables of the Permanent), is the following set of polynomial equations (in the $\overline{w}$-variables only):
         \[
            \{\coeff_{M_i} (U(\overline{x},\overline{w})) = b_i : 1\leq i \leq N\},
        \]
        where $\overline{b}=\coeff(\perm(\overline{x})) \in \mathbb{F}^N$ is the coefficient vector of the polynomial $\perm(\overline{x})$ of dimension $n$, $U(\overline{x},\overline{w})$ is the constant-depth universal circuit for polynomials of depth at most $\Delta$ and have circuits of size at most $s$, $\overline{w}$ are the $K_{s,\Delta}$ edge variables, $\{M_i\}_{i=1}^N$ is the set of all possible $\overline{x}$-monomials of degree at most $l$, and $N=\Sigma_{j=0}^l \binom{n^2+j-1}{j}=2^{O(n^2+l)}$ is the number of monomials of total degree at most $l$ over $n^2$ variables. Then, the size of the above set of polynomial equations is $O(2^{(t+l)l}\cdot |U(\overline{x},\overline{w})| \cdot N)$ where $t$ is maximum multiplication fan-in in $U(\overline{x},\overline{w})$.\par 
    \end{definition}


\section{No Short $\AC^0[p]$-Frege Proofs for Diagonalizing DNF Formulas}\label{sec:main-theorem-sec}

    This section presents our main result. we show unconditionally that constant-depth $\IPS$ cannot efficiently refute certain constant-depth $\IPS$ upper bounds. As a corollary, we obtain the same result for $\ACZ[p]$-Frege, since this proof system is simulated by constant-depth $\IPS$ over $\F_p$ (\Cref{thm:cdIPS simulates AC0p-Frege}). More precisely, we prove that constant-depth $\IPS$ does not admit polynomial-size refutations of the diagonalizing $\CNF$ formulas $\Phi_{t,l,\Delta',n,s,\Delta}$, infinitely often. These formulas express the existence of size-$t$, depth-$\Delta'$ $\IPS$ refutation of the statement that the Permanent polynomial is computable by size-$s$, depth-$\Delta$ algebraic circuits.

    In addition, this section contains the proof of Theorem~\ref{thm:mainthm3}, which shows that ruling out short proofs for the diagonalizing formulas is a necessary step towards constant-depth IPS lower bounds. We begin by introducing the formulas that will be used throughout the argument.

Let  $N = \sum_{j=0}^l \binom{n^2+j-1}{j}=2^{O(n^2+l)}$ be the number of monomials of total degree at most $l$ over $n^2$ variables. We shall work over a finite field $\F_q$ here to enable the encoding of circuit equations as CNF formulas.
    
    \begin{itemize}
        \item $\VNP=\VAC^0(n,s,l,\Delta)$: circuit equations expressing that there is an algebraic circuit of size $s$ and depth $\Delta$ that \textit{agrees} with the Permanent polynomial of dimension $n$ on all the coefficients of monomials of degree at most $l$ (i.e., every monomial $M$ computed by the algebraic circuit has the same coefficient as in the Permanent polynomial).       
            \begin{itemize}
                \item Type: circuit equations;
                \item Number of variable: $K_{s, \Delta}=\poly(s, \Delta);$
                \item Size: $O(2^{(n+l)l} \cdot \poly(s, \Delta) \cdot N).$ 
            \end{itemize}
        \item $\varphi_{n,s,l,\Delta}^\cnf$: the $\CNF$ encoding of $\VNP=\VAC^0(n,s,l,\Delta)$ based on \Cref{def: plain CNF encoding of a constant-depth circuit equation cnf(C(x)=0)}.
            \begin{itemize}
                \item Type: $\CNF$ formula;
                \item Number of variable: $O(2^{(n+l)l} \cdot \poly(s, \Delta) \cdot N)$; 
                \item Size: $O(q \cdot 2^{(n+l)l} \cdot \poly(s, \Delta) \cdot N)$. 
            \end{itemize}
        \item $\varphi_{n,s,l,\Delta}^\scnf$: the $\SCNF$ encoding of $\VNP=\VAC^0(n,s,l,\Delta)$ based on \Cref{def: semi-CNF encoding} together with the field axioms ($x^q -x =0$) for all variables.
            \begin{itemize}
                \item Type: $\SCNF$ formula;
                \item Number of variable: $O(2^{(n+l)l} \cdot \poly(s, \Delta) \cdot N)$.
            \end{itemize}
        \item $\IPS_{\refute}(t,\Delta,l,\overline{\mathcal{F}})$: circuit equations expressing that there exists an algebraic circuit for size $t$ and depth $\Delta$ that agrees with the $\IPS$ refutation of $\overline{\mathcal{F}}$ on all the coefficients of monomials of degree at most $l$. 
            \begin{itemize}
                \item Type: circuit equations;
                \item Number of variable: $K_{t, \Delta}=\poly(t, \Delta)$; 
                \item Size: $O(2^{(n+l)l} \cdot \poly(t, \Delta) \cdot |\overline {\mathcal{F}}| \cdot N)$.
            \end{itemize}
        \item  $\Phi_{t,l,\Delta^\prime,n,s,\Delta}$: \emph{The diagonalizing CNF formula}. More precisely, the $\CNF$ encoding of $\IPS_\refute(t,\Delta^\prime,l,\varphi_{n,s,l,\Delta}^\scnf)$ expressing that there is an algebraic circuit (the purported \IPS\ refutation) of size $t$ and depth $\Delta^\prime$ that agrees with the $\IPS$ refutation of $\varphi_{n,s,l,\Delta}^\scnf$ on all the coefficients of monomials of degree at most $l$.
            \begin{itemize}
                \item Type: $\CNF$ formulas;
                \item Number of variable: $K_{t, \Delta}$ which is $\poly(t, \Delta)$;
                \item Size: $O(q \cdot 2^{(n+l)l} \cdot \poly(t, \Delta) \cdot |\mathcal{F}| \cdot N)$.
            \end{itemize}
    \end{itemize}

    The following lemma easily follows from \Cref{lemma: unsat CNF has VNP refutation}.
    \begin{lemma}
        \label{lemma: unsat eCNF formulastar has VNP refutation}
        Every family of unsatisfiable formulas $(\varphi_n)$, which contains a set of unsatisfiable $\CNF$ formulas, has a family of $\IPS^\alg$ (also $\IPS$) certificates $(C_n)$ in $\VNP_\mathbb{F}$.
    \end{lemma}

    Since Semi-CNFs are substitution instances of CNFs, we get the following lemma by substituting the occurrence of Boolean variables with their corresponding $\UBIT$.
    \begin{lemma}
        \label{lemma: unsat SCNF formula has VNP refutation}
        $\{\varphi_{n,s,l,\Delta}^\scnf\}_n$ is a family of unsatisfiable $\SCNF$ formulas and has a family of $\IPS^\alg$ (also $\IPS$) certificates $(C_n)$ in $\VNP_\mathbb{F}$.
    \end{lemma}

    We fix $l: \mathbb{N} \to \mathbb{N}$ to be a (monotone) size function $l(r) = \lceil r^\epsilon\rceil$ for some constant $\epsilon$.

    The main result of this section is the following.

    \begin{theorem}[Main; no short proofs over fixed finite field]
        \label{theorem: main theorem}
        Let $q$ be a constant prime.
        The $\CNF$ family $\{\Phi_{t,l,\Delta^\prime,n,s,\Delta}\}$ does not have constant-depth polynomial-size $\IPS$ refutations infinitely often over $\mathbb{F}_q$, in the following sense: 
        for every constant $\Delta$ there exist constants $c_1$ and  $\Delta^\prime$,  such that for every sufficiently large constant $c_2$ and every constants $\Delta^{\prime\prime}$ and $c_0$, for infinitely many $n, t(n), s(n) \in \mathbb{N}$, such that $t(n) > |\varphi_{n,s,l,\Delta}^\scnf|^{c_1}$ and $n^{c_1} < s(n) < n^{c_2}$, $\Phi_{t,l,\Delta^\prime,n,s,\Delta}$ has no $\IPS$ refutation of size at most $|\Phi_{t,l,\Delta^\prime,n,s,\Delta}|^{c_0}$ and depth at most $\Delta^{\prime\prime}$.
    \end{theorem}

    Before proving this theorem, we provide an overview of its proof. 

\medskip

    \noindent\textit{Proof overview:} 
    By way of contradiction, we assume that there exists a constant $\Delta$ such that for every constant $\Delta^\prime$, for every sufficiently large $n$, and for every $t(n), s(n) \in \mathbb{N}$, such that $t(n) > |\varphi_{n,s,l,\Delta}^\scnf|^{c_1}$ and $n^{c_1} < s(n) < n^{c_2}$,
        $\Phi_{t,l,\Delta^\prime,n,s,\Delta} \coloneqq \cnf(\IPS_\refute(t,\Delta^\prime,l,\varphi_{n,s,l,\Delta}^\scnf))$ has a small depth-$O(1)$ refutation. Then, by substituting the occurrences of Boolean variables with the correspondence $\UBIT$s, we can assume that $\scnf(\IPS_\refute(t,\Delta^\prime,l,\varphi_{n,s,l,\Delta}^\scnf))$ has a small and depth-$O(\Delta^\prime)$ refutation.

        \begin{enumerate}
            
            \item \label{it:2292}
            By applying \Cref{lemma: Translate circuit equations from semi-CNFs in Fixed Finite Fields}, from $\varphi_{m,t,l,\Delta}^\scnf$ where $m,t,l$ are parameters, that meet the conditions in \Cref{theorem: main theorem}, we can derive the circuit equations
            \[
                \VNP=\VAC^0(m,t,l,\Delta),
            \]
            in $O(\Delta)$ depth and polynomial-size $\IPS$. Recall that $\VNP=\VAC^0(m,t,l,\Delta)$ is the set of circuit equations expressing that there is an algebraic circuit of size $t$ and depth $\Delta$ \ that agrees with the Permanent polynomial of dimension $m$ on all the coefficients of monomials of degree at most $l$.

            \item \label{it:2301}
             In \Cref{lemma: grochow-pitassi formalization in cdIPS} we prove (the contrapositive of) the bounded-depth version of \Cref{theorem: IPS lower bound implies VNP neq VP} \emph{within bounded-depth $\IPS$}.
            The contrapositive of the bounded-depth version of \Cref{theorem: IPS lower bound implies VNP neq VP} expresses that if the Permanent polynomial can be computed by bounded-depth polynomial-size circuits, then bounded-depth $\IPS$ can efficiently refute any family of $\CNF$ formulas. 
            To be more specific, from $\VNP=\VAC^0(m,t,l,\Delta)$ we derive $\IPS_\refute(t,\Delta,l,\varphi_{n,s,l,\Delta}^\scnf)$ in depth $O(\Delta)$ and polynomial size. Note that $\IPS_\refute(t,\Delta,l,\varphi_{n,s,l,\Delta}^\scnf)$ is the set of circuit equations expressing that $\IPS$ can refute $\varphi_{n,s,l,\Delta}^\scnf$ in size $t$ and depth $\Delta$.
        
            \item \label{it:2306} Applying \Cref{lemma: Translate semi-CNFs from circuit equations in Fixed Finite Fields}, we can derive $\scnf(\IPS_\refute(t,\Delta,l,\varphi_{n,s,l,\Delta}^\scnf))$ from $\IPS_\refute(t,\Delta,l,\varphi_{n,s,l,\Delta}^\scnf)$ in depth $O(\Delta)$ polynomial size $\IPS$.
            Since we assumed that for any constant $\Delta^\prime$, $\scnf(\IPS_\refute(t,\Delta^\prime,l,\varphi_{n,s,l,\Delta}^\scnf))$ has a small and bounded-depth refutation, it follows that $\scnf(\IPS_\refute(t,\Delta,l,\varphi_{n,s,l,\Delta}^\scnf))$ has a small and bounded-depth refutation, particularly when we take $\Delta^\prime = \Delta$.
        \end{enumerate}

        Hence, we get a small and bounded-depth refutation of $\varphi_{m,t,l,\Delta}^\ast$, as follows:
        \begin{align*}
            \varphi_{m,t,l,\Delta}^\ast &\sststile{\IPS^\alg}{\ast, O(\Delta)} \VNP=\VAC^0(m,t,l,\Delta) \qquad  \text{\Cref{lemma: Translate circuit equations from semi-CNFs in Fixed Finite Fields} (see \Cref{it:2292})}\\ 
            & \sststile{\IPS^\alg}{\ast, O(\Delta)} \IPS_\refute(t,\Delta,l,\varphi_{n,s,l,\Delta}^\scnf) \qquad \text{\Cref{lemma: grochow-pitassi formalization in cdIPS} (see \Cref{it:2301})}\\
            & \sststile{\IPS^\alg}{\ast, O(\Delta)} \scnf(\IPS_\refute(t,\Delta,l,\varphi_{n,s,l,\Delta}^\scnf)) \qquad \text{\Cref{lemma: Translate semi-CNFs from circuit equations in Fixed Finite Fields} (see \Cref{it:2306})}\\
             & \sststile{\IPS^\alg}{\ast, O(\Delta)}  1=0 \qquad \text{by assumption}.
        \end{align*}

    Then we know that for some constant $\Delta^{\prime\prime\prime}$ and large enough $w$, the system of circuit equations $\IPS_\refute^\alg(w,\Delta^{\prime\prime\prime},l,\varphi_{m,t,l,\Delta}^\ast)$ is satisfiable. Hence, by \Cref{prop: C <-> cnf <-> ecnf fixed finite fields}, $\cnf(\IPS_\refute^\alg(w,\Delta^{\prime\prime\prime},l,\varphi_{m,t,l,\Delta}^\ast))$ is satisfiable. 
        
    By assumption, for any constant $\Delta^\prime$, for all sufficiently large $n$, for all proper $t(n), s(n) \in \mathbb{N}$, $\Phi_{t,l,\Delta^\prime,n,s,\Delta} \coloneqq \cnf(\IPS_\refute(t,\Delta^\prime,l,\varphi_{n,s,l,\Delta}^\scnf))$ has a small and depth-$O(1)$ refutation. Taking $t=w$, $\Delta^\prime = \Delta^{\prime\prime\prime}$, $n=m$, $s=t$, we know that $\cnf(\IPS_\refute^\alg(w,\Delta^{\prime\prime\prime},l,\varphi_{m,t,l,\Delta}^\scnf))$ is refutable which implies it is not satisfiable. This is a contradiction.

    \bigskip

    \begin{proof}[Proof of \Cref{theorem: main theorem}]
    
        For the sake of contradiction, we assume that there exists a constant $\Delta$ for all constants $c_1$, $\Delta^\prime$ such that there exist constants $c_2$, $\Delta^{\prime\prime}$ and $c_0$ for all $n, t(n), s(n) \in \mathbb{N}$ such that $t(n)> |\varphi_{n,s,l,\Delta}^\scnf|^{c_1}$ and $n^{c_1} < s(n) < n^{c_2}$,
        \begin{equation}
            \underbrace{\cnf(\IPS_\refute^\alg(t,l ,\Delta^\prime,\varphi_{n,s,l,\Delta}^\scnf))}_{\lambda} \sststile{\IPS}{|\lambda|^{c_0}, \Delta^{\prime\prime}} 1=0.
        \end{equation}

        By substituting the occurrence of Boolean variables with their corresponding $\UBIT$s, we get that
        there exists a constant $\Delta$ for all constants $c_1$, $\Delta^\prime$ such that there exist constants $c_2$, $\Delta^{\prime\prime}$ and $c_0$ for all $n, t(n), s(n) \in \mathbb{N}$, if $t(n)> |\varphi_{n,s,l,\Delta}^\scnf|^{c_1}$ and $n^{c_1} < s(n) < n^{c_2}$,
        \begin{equation}
            \underbrace{\scnf(\IPS_\refute^\alg(t,l ,\Delta^\prime,\varphi_{n,s,l,\Delta}^\scnf))}_{\lambda} \sststile{\IPS^\alg}{|\lambda|^{c_0}, \Delta^{\prime\prime}} 1=0.
        \end{equation}
        Since the Boolean axioms of each $\UBIT$ can be easily derived, as shown in the proof of \Cref{lemma: Translate semi-CNFs from circuit equations in Fixed Finite Fields}, $\IPS$ can be replaced by $\IPS^\alg$ here. 
        
        We take $\Delta^\prime = \Delta$, which gives us the following:
        \begin{equation}
            \underbrace{\scnf(\IPS_\refute^\alg(t,l, \Delta,\varphi_{n,s,l,\Delta}^\scnf))}_{\lambda} \sststile{\IPS^\alg}{|\lambda|^{c_0}, \Delta^{\prime\prime}} 1=0.
        \end{equation}

        Let $m = 6(L^\prime+P^\prime)$ where $L^\prime$ is the number of variables in $ \varphi_{n,s,l,\Delta}^\scnf$ and $P^\prime$ is the number of placeholders needed for axioms in $ \varphi_{n,s,l,\Delta}^\scnf$. Let $\gamma \coloneqq \varphi_{m,t,l,\Delta}^\scnf$ be the Semi-CNF formulas $\varphi_{m,t,l,\Delta}^\scnf$ together with the field axioms for the variables in $\IPS_\refute^\alg(t,l,\Delta,\varphi_{n,s,l,\Delta}^\scnf)$.
        
        By \Cref{lemma: grochow-pitassi formalization in cdIPS}, there exist constants $c_3$ and $\Delta^{\prime\prime\prime}$ such that
        \[
            \overbrace{\varphi_{m,t,l,\Delta}^\scnf}^\gamma \sststile{\IPS^\alg}{|\gamma|^{c_3}, \Delta^{\prime\prime\prime}} \scnf(\IPS_\refute^\alg(t,l, \Delta, \varphi_{n,s,l,\Delta}^\scnf)).
        \]
        Combining the above equation with the assumption, we know that

        \begin{equation}
            \varphi_{m,t,l, \Delta}^\scnf \sststile{\IPS^\alg}{|\gamma|^{c_3}, \Delta^{\prime\prime\prime}}  \underbrace{\scnf(\IPS_\refute^\alg(t,l, \Delta,\varphi_{n,s,l,\Delta}^\scnf))}_{\lambda} \sststile{\IPS^\alg}{|\lambda|^{c_0},\Delta^{\prime\prime}} 1=0.
        \end{equation}
        Since $|\lambda|$ is polynomially bounded by $|\gamma|$, there exists a constant $c_1$ such that
        \[
            |\gamma|^{c_3} + |\lambda|^{c_0} < |\gamma|^{c_1}.
        \]
        We pick a big enough $c_1$ such that
        \[
            \underbrace{\varphi_{m,t,l, \Delta}^\scnf}_{\gamma}  \sststile{\IPS^\alg}{|\gamma|^{c_1}, \Delta^{\prime\prime\prime}+\Delta^{\prime\prime}} 1 =0.
        \]

        From the above equation, we can conclude that $\IPS_\refute^\alg(w,l,\Delta^{\prime\prime\prime}+\Delta^{\prime\prime}, \varphi_{m,t,l,\Delta}^\scnf)$ is satisfiable for some $w \geq |\gamma|^{c_1}$ over $\mathbb{F}_q$ and is polynomially bounded by $|\gamma|$. Hence, by \Cref{prop: C <-> cnf <-> ecnf fixed finite fields}, $\cnf(\IPS_\refute^\alg(w,l,\Delta^{\prime\prime\prime}+\Delta^{\prime\prime}, \varphi_{m,t,l,\Delta}^\scnf))$ is also satisfiable over $\mathbb{F}_q$. However, by our assumption, when we take $\Delta^\prime = \Delta^{\prime\prime\prime}+\Delta^{\prime\prime}$, we know that for all big enough $w$ and $t$, $\IPS_\refute^\alg(w,l,\Delta^{\prime\prime\prime}+\Delta^{\prime\prime}, \varphi_{m,t,l,\Delta}^\scnf)$ is refutable. By the soundness of $\IPS$, $\cnf(\IPS_\refute^\alg(w,l,\Delta^{\prime\prime\prime}+\Delta^{\prime\prime}, \varphi_{m,t,l,\Delta}^\scnf))$ should be unsatisfiable which is a contradiction.
    \end{proof}

    It remains to prove the following. 

\begin{lemma}[Constant-depth version of  Grochow-Pitassi formalization in $\IPS^\alg$]
    \label{lemma: grochow-pitassi formalization in cdIPS}
    There are constants $c_3$ and $\Delta^{\prime\prime\prime}$ such that under the above notation and parameters:
    \begin{equation*}
        \overbrace{\varphi_{m,t,l,\Delta}^\scnf}^\gamma \sststile{\IPS^\alg}{|\gamma|^{c_3}, \Delta^{\prime\prime\prime}} \scnf(\IPS_\refute^\alg(t,l, \Delta, \varphi_{n,s,l,\Delta}^\scnf)).
    \end{equation*}
\end{lemma}

\begin{proof}
    Recall that $\gamma \coloneqq \varphi_{m,t,l,\Delta}^\scnf$ is the Semi-CNF formula $\varphi_{m,t,l,\Delta}^\scnf$ together with field axioms for all variables in it.
    Let $N^\prime$ be the number of $\overline{x}$-monomials with degree at most $l$ over $m$ variables.
    According to
    \Cref{lemma: Translate circuit equations from semi-CNFs in Fixed Finite Fields}, from $\varphi_{m,t,l,\Delta}^\scnf$, we can derive the following circuit equations in polynomial-size $\IPS^\alg$ 
    \[
        \{\coeff_{M_i} (U(\overline{x},\overline{w})) = b_i : 1\leq i \leq N^\prime\},
    \]
    where $\overline{b}=\coeff(\perm(\overline{x})) \in \mathbb{F}_q^{N^\prime}$ is the coefficient vector of the Permanent polynomial $\perm(\overline{x})$, and $U(\overline{x},\overline{w})$ is the universal circuit for algebraic circuits of depth at most $\Delta$ and of size at most $t$.
    Formally, there exist constants $c_4$ and $\Delta_1$ such that
    \begin{equation*}
            \lambda \sststile{\IPS^\alg}{|\lambda|^{c_4}, \Delta_1}  \{\coeff_{M_i} (U(\overline{x},\overline{w})) = b_i : 1\leq i \leq N^\prime\}
    \end{equation*}
    which is
        \[
            \lambda \sststile{\IPS^\alg}{|\lambda|^{c_4}, \Delta_1} \VNP = \VAC^0(m,t,l,\Delta).
        \]

        \bigskip
        \bigskip

    Now, we will show that from $\VNP = \VAC^0(m,t,l,\Delta)$, there is a depth-$O(\Delta)$ polynomial size $\IPS$ derivation of $\IPS_\refute^\alg(t,l,\Delta, \varphi_{n,s,l,\Delta}^\scnf)$.

    \begin{claim}
        \label{claim: main part of GP in IPS}
        Suppose $M_i$ is an $\overline{x}$-monomial with degree at most $l$. Let $\overline{a}$ be any possibly partial substitution of polynomials (including field elements) for the variables $\overline{x}$. Given $\VNP = \VAC^0(m,t,l,\Delta) \coloneqq \{\coeff_{M_i} (U(\overline{x},\overline{w})) = b_i : 1\leq i \leq N^\prime\}$, we can derive
        \begin{equation}
            \coeff_M(U(\overline{x}\restriction_{\overline{a}}, \overline{w})) = \coeff_M(\perm(\overline{x}\restriction_{\overline{a}}))
        \end{equation}
        in depth-$O(\Delta)$ polynomial size $\IPS^\alg$ for every $\overline{x}$-monomial $M$ with degree at most $l$.
    \end{claim}

    \begin{proof}[Proof of \Cref{claim: main part of GP in IPS}]
        First, it is easy to get the following polynomial identity,
    \begin{equation}
        \label{equa: assignment can be taken out}            U(\overline{x}\restriction_{\overline{a}}, \overline{w}) = U(\overline{x},\overline{w})\restriction_{\overline{a}} \,.
    \end{equation}
    And by \Cref{def: coefficient circuit}, we have the following polynomial identity,
    \begin{equation}
        \label{equa: by definition 1}
        U(\overline{x},\overline{w})\restriction_{\overline{a}} = 
        \left(
            \sum_{i\in [N^\prime]}    \coeff_{M_i}(U(\overline{x}, \overline{w}))\cdot M_i%
        \right)\restriction_{\overline{a}}\,.
    \end{equation}
    
    And since $\coeff_{M_i}(U(\overline{x},\overline{w}))$ only contains $\overline{w}$ variables, we have the following polynomial identity,
    \begin{equation}
        \label{equa: coeff only have w variables}
        \left (\sum_{i\in [N^\prime]} \coeff_{M_i}(U(\overline{x}, \overline{w})) \cdot M_i \right )\restriction_{\overline{a}} = \sum_{i\in [N^\prime]} \coeff_{M_i}(U(\overline{x}, \overline{w})) \cdot (M_i\restriction_{\overline{a}})\,.
    \end{equation}
    
    Hence, combining three polynomial identities in \Cref{equa: assignment can be taken out}, \Cref{equa: by definition 1} and \Cref{equa: coeff only have w variables} above, we have the following polynomial identity,
    \begin{equation}
        \label{equa: first combination}
        U(\overline{x}\restriction_{\overline{a}}, \overline{w}) = \sum_{i\in [N^\prime]} \coeff_{M_i}(U(\overline{x}, \overline{w})) \cdot (M_i\restriction_{\overline{a}}) \,.
    \end{equation}
        
    Let $M$ be any $\overline{x}$-monomial with degree at most $l$. By the above polynomial identity, we have the following polynomial identity,
    \begin{equation}
        \label{equa: by definition 2}
        \coeff_M( U(\overline{x}\restriction_{\overline{a},} \overline{w})) = \coeff_M(\sum_{i\in [N^\prime]} \coeff_{M_i}(U(\overline{x}, \overline{w})) \cdot (M_i\restriction_{\overline{a}}))\,.
    \end{equation}
    
    Since $\coeff_M$ is a linear operator, we have the following polynomial identity,
        \begin{equation}
            \label{equa: coefficient of an algebraic circuit is the same as the sum of all the coefficients of subcircuit}
            \coeff_M(\sum_{i\in [N^\prime]} \coeff_{M_i}(U(\overline{x}, \overline{w})) \cdot (M_i\restriction_{\overline{a}})) = \sum_{i\in [N^\prime]} \coeff_M(\coeff_{M_i}(U(\overline{x}, \overline{w})) \cdot (M_i\restriction_{\overline{a}})) \,.
        \end{equation}
        
    Since $\coeff_M(\coeff_{M_i}(U(\overline{x}, \overline{w})) \cdot (M_i\restriction_{\overline{a}})) = \coeff_{M_i}(U(\overline{x}, \overline{w})) \cdot \coeff_M(M_i\restriction_{\overline{a}})$, we have the following polynomial identity,
        \begin{equation}
            \label{equa: coefficient for subcircuits in disjoint variables}
            \sum_{i\in [N^\prime]} \coeff_M(\coeff_{M_i}(U(\overline{x}, \overline{w})) \cdot (M_i\restriction_{\overline{a}})) = \sum_{i\in [N^\prime]} \coeff_{M_i}(U(\overline{x}, \overline{w})) \cdot \coeff_M( (M_i\restriction_{\overline{a}}))\,.
        \end{equation}
        Hence, combining the three polynomial identities in \Cref{equa: by definition 2}, \Cref{equa: coefficient of an algebraic circuit is the same as the sum of all the coefficients of subcircuit} and \Cref{equa: coefficient for subcircuits in disjoint variables}, we have the following polynomial identity,
        \begin{equation}
            \label{equa: first polynomial identity in IPS}
            \coeff_M( U(\overline{x}\restriction_{\overline{a},} \overline{w})) = \sum_{i\in [N^\prime]} \coeff_{M_i}(U(\overline{x}, \overline{w})) \cdot \coeff_M( (M_i\restriction_{\overline{a}}))\
        \end{equation}
        in depth-$O(\Delta)$ linear-size $\IPS^\alg$ (we increased the depth of $U(\overline{x}, \overline{w})$ by a constant factor using \Cref{prop: computation of coefficients in constant-depth circuits}). 
        
        \bigskip
        
        Also, since we have already derived $\{\coeff_{M_i} (U(\overline{x},\overline{w})) = b_i : 1\leq i \leq N^\prime\}$ from $\lambda$, by multiplying $\coeff_M(M_i\restriction_{\overline{a}})$ to each circuit equation and adding them, we can derive the following circuit equation in depth-$O(\Delta)$ polynomial-size $\IPS^\alg$,

        \begin{equation}
            \label{equa: second polynomial in IPS}
            \sum_{i\in [N^\prime]} \coeff_{M_i}(U(\overline{x}, \overline{w})) \cdot \coeff_M( (M_i\restriction_{\overline{a}})) = \sum_{i \in [N^\prime]} b_i \cdot \coeff_M(M_i\restriction_{\overline{a}})\,.
        \end{equation}

        \bigskip

        By definition, we have the following polynomial identity,
        \begin{equation}
            \label{equa: perm identity}
            \sum_{i\in [N^\prime]} b_i \cdot (M_i\restriction_{\overline{a}}) = \perm(\overline{x}\restriction_{\overline{a}})\,.
        \end{equation}
        
        Therefore, for any $\overline{x}$-monomial $M$ with degree at most $l$, we have the following polynomial identity,
        \begin{equation}
            \label{equa: apply coefficient to perm identity}
            \coeff_M(\sum_{i\in [N^\prime]} b_i \cdot (M_i\restriction_{\overline{a}})) = \coeff_M(\perm(\overline{x}\restriction_{\overline{a}}))\,.
        \end{equation}
        
        Also, since $b_i$ are just field elements and linearity of $\coeff_M$, 
        \begin{equation}
            \label{equa: take bi out}
            \coeff_M(\sum_{i \in [N^\prime]} b_i \cdot (M_i\restriction_{\overline{a}}))  = \sum_{i \in [N^\prime]} b_i \cdot \coeff_M(M_i\restriction_{\overline{a}}) \,.
        \end{equation}
        
        Combining the two polynomial identities in \Cref{equa: apply coefficient to perm identity} and \Cref{equa: take bi out} above and using \Cref{prop: polynomial identities that can be written as constant-depth circuits are proved for free in constant-depth IPS} again, we prove the following polynomial identities,

        \begin{equation}
            \label{equa: third polynomial identity in IPS}
            \sum_{i \in [N^\prime]} b_i \cdot \coeff_M(M_i\restriction_{\overline{a}}) = \coeff_M(\perm(\overline{x}\restriction_{\overline{a}}))
        \end{equation}
        in depth-$O(\Delta)$ polynomial size $\IPS^\alg$.

        \bigskip
        
        Hence, by combining those three circuit equations in \Cref{equa: first polynomial identity in IPS}, \Cref{equa: second polynomial in IPS} and \Cref{equa: third polynomial identity in IPS},
        \begin{align*}
            \coeff_M( U(\overline{x}\restriction_{\overline{a},} \overline{w})) &= \sum_{i\in [N^\prime]} \coeff_{M_i}(U(\overline{x}, \overline{w})) \cdot \coeff_M( (M_i\restriction_{\overline{a}}))\\
            \sum_{i\in [N^\prime]} \coeff_{M_i}(U(\overline{x}, \overline{w})) \cdot \coeff_M( (M_i\restriction_{\overline{a}})) &= \sum_{i \in [N^\prime]} b_i \cdot \coeff_M(M_i\restriction_{\overline{a}})\\
            \sum_{i \in [N^\prime]} b_i \cdot \coeff_M(M_i\restriction_{\overline{a}}) &= \coeff_M(\perm(\overline{x}\restriction_{\overline{a}})),
        \end{align*}
        we can derive
        
        \begin{equation}
            \label{equa: final circuit equation}
            \coeff_M(U(\overline{x}\restriction_{\overline{a}}, \overline{w})) = \coeff_M(\perm(\overline{x}\restriction_{\overline{a}}))
        \end{equation}
        in depth-$O(\Delta)$ polynomial size $\IPS^\alg$. This concludes the proof of \Cref{claim: main part of GP in IPS}.
    \end{proof}
    
        \bigskip
        Now, we divide $\overline{x}$ into $\overline{x^\prime}$ and $\overline{y}$ two disjoint parts of variables, where $\overline{y}$ represents the placeholder for axioms and $\overline{x^\prime}$ represents the rest. 
        
        By \Cref{theorem: LST25}, we know that $\varphi_{n,s,l,\Delta}^\scnf$, which contains $\varphi_{n,s,l,\Delta}^\scnf$ that denotes the $\SCNF$ encoding of the circuit equation $\VNP=\VAC^0(n,s,l,\Delta)$ is unsatisfiable. Therefore, by \Cref{lemma: unsat SCNF formula has VNP refutation}, there exists a $\VNP$-$\IPS^\alg$ refutation for $\varphi_{n,s,l,\Delta}^\scnf$.

        Note that the Permanent polynomial $\perm(\overline{x})$ of dimension $m= 6(L^\prime+P^\prime)$ is complete for $\VNP$ with $m/6 = L^\prime+P^\prime$ variables which is the number of variables in $\IPS_\refute^\alg(t,l,\Delta, \varphi_{n,s,l,\Delta}^\scnf)$ \cite{Val79:ComplClass}. Therefore, there exists an substitution $\overline{\alpha}$ such that $\perm(\overline{x}\restriction_{\overline{\alpha}})$ computes exactly a $\VNP$-$\IPS^\alg$ refutation of $\varphi_{n,s,l,\Delta}^\scnf$.
 
        Notice that under substitution $\overline{\alpha}$, $\overline{y}$ must be unassigned because they are the variables representing placeholders which means $\perm(\overline{x}\restriction_{\overline{\alpha}})=\perm(\overline{x^\prime}\restriction_{\overline{\alpha}}, \overline{y})$.
        
        We use $\overline{y}\restriction_{\overline{0}}$ to represent that all $\overline{y}$ variables are assigned zero, and $\overline{y}\restriction_{\varphi_{n,s,l,\Delta}^\scnf}$ to represent that each $y_i$ is replaced by the corresponding $i$th axiom from $\varphi_{n,s,l,\Delta}^\scnf$.

        By \Cref{claim: main part of GP in IPS}, from 
        \[
            \VNP = \VAC^0(m,t,l,\Delta),
        \]
        we can derive the following circuit equation,
        \begin{equation}
            \label{equa: equation for refutation predicate 1}
            \coeff_M(U(\overline{x^\prime}\restriction_{\overline{\alpha}}, \overline{y}\restriction_{\overline{0}},\overline{w})) = \coeff_M(\perm(\overline{x^\prime}\restriction_{\overline{\alpha}},\overline{y}\restriction_{\overline{0}}))
        \end{equation}
        where $\overline{\alpha}$ is the substitution such that $\perm(\overline{x}\restriction_{\overline{\alpha}})$ computes exactly the $\VNP$-$\IPS^\alg$ refutation of $\varphi_{n,s,l,\Delta}^\scnf$.

        Therefore, $\perm(\overline{x^\prime}\restriction_{\overline{\alpha}},\overline{y}\restriction_{\overline{0}})$ is the $\IPS^\alg$ refutation of $\varphi_{n,s,l,d}^\scnf$ with placeholder variables replaced with all zero. By the definition of $\IPS^\alg$, we have the following polynomial identity,
        \begin{equation}
            \label{equa: equation for refutation predicate 2}
            \coeff_M(\perm(\overline{x^\prime}\restriction_{\overline{\alpha}},\overline{y}\restriction_{\overline{0}})) = 0\,.
        \end{equation}

        Combining \Cref{equa: equation for refutation predicate 1} and \Cref{equa: equation for refutation predicate 2}, we have
        \begin{equation}
            \label{equa: first half of refutation predicate}
            \coeff_M(U(\overline{x^\prime}\restriction_{\overline{\alpha}}, \overline{y}\restriction_{\overline{0}},\overline{w})) = 0\,.
        \end{equation}
        Again, by \Cref{claim: main part of GP in IPS}, we have
        \begin{equation}
            \label{equa: equation for refutation predicate 3}
            \coeff_M(U(\overline{x^\prime}\restriction_{\overline{\alpha}},\overline{y}\restriction_{\varphi_{n,s,l,\Delta}^\scnf}, \overline{w})) = \coeff_M(\perm(\overline{x^\prime}\restriction_{\overline{\alpha}}, \overline{y}\restriction_{\varphi_{n,s,l,\Delta}^\scnf}))\,.
        \end{equation}

        Note that $\perm(\overline{x^\prime}\restriction_{\overline{\alpha}}, \overline{y}\restriction_{\varphi_{n,s,l,\Delta}^\scnf})$ is the $\IPS^\alg$ refutation of $\varphi_{n,s,l,\Delta}^\scnf$ with placeholder variables replaced with all axioms in $\varphi_{n,s,l,\Delta}^\scnf$. By the definition of $\IPS^\alg$, we have the following polynomial identity,
        \begin{equation}
            \label{equa: equation for refutation predicate 4}
            \coeff_M(\perm(\overline{x^\prime}\restriction_{\overline{\alpha}}, \overline{y}\restriction_{\varphi_{n,s,l,\Delta}^\scnf})) 
            =  \begin{cases}
                 1,\ M_i =1  \ (\ie, \text{the constant 1 monomial});\\
                0,\ \text{otherwise},
            \end{cases}\,.
        \end{equation}

        Combining \Cref{equa: equation for refutation predicate 3} and \Cref{equa: equation for refutation predicate 4}, we have
        \begin{equation}
            \label{equa: second half of refutation predicate}
            \coeff_M(U(\overline{x^\prime}\restriction_{\overline{\alpha}},\overline{y}\restriction_{\varphi_{n,s,l,\Delta}^\scnf}, \overline{w})) = \begin{cases}
                 1,\ M_i =1  \ (\ie, \text{the constant 1 monomial});\\
                0,\ \text{otherwise}.
            \end{cases}
        \end{equation}

        Note that \Cref{equa: first half of refutation predicate} and \Cref{equa: second half of refutation predicate} are exactly the circuit equations in $\IPS_\refute^\alg(t,l, \Delta,\varphi_{n,s,l,\Delta}^\scnf)$. Also, note that we can prove \Cref{equa: first half of refutation predicate} and \Cref{equa: second half of refutation predicate} for every monomial $M$ in parallel, and there are only a constant many polynomial identities and derivations in the proof for each $M$. Moreover, the depth of each polynomial identity and derivation is bounded by $O(\Delta)$. Hence, we can conclude that there exist constants $c_5$ and $\Delta_2= O(\Delta)$ such that
        \begin{equation}
            \underbrace{\varphi_{m,t,l, \Delta}^\scnf}_{\gamma} \sststile{\IPS^\alg}{|\gamma|^{c_5}, \Delta_2} \IPS_\refute^\alg(t,l,\Delta, \varphi_{n,s,l,\Delta}^\scnf)\,.
        \end{equation}
        Note that $\varphi_{m, t,l,\Delta}^\scnf$ already includes all the field axioms of variables in $ \IPS_\refute^\alg(t,l,\Delta, \varphi_{n,s,l,\Delta}^\scnf)$. By \Cref{lemma: Translate semi-CNFs from circuit equations in Fixed Finite Fields}, we can derive $\scnf(\IPS_\refute^\alg(t,l,\Delta, \varphi_{n,s,l,\Delta}^\scnf))$ in depth-$O(\Delta)$ polynomial-size $\IPS^\alg$.
        
        We can conclude that there exist constants $c_3$ and $\Delta^{\prime\prime\prime}$ such that \begin{equation}
            \underbrace{\varphi_{m,t,l, \Delta}^\scnf}_{\gamma} \sststile{\IPS^\alg}{|\gamma|^{c_3}, \Delta^{\prime\prime\prime}} \scnf(\IPS_\refute^\alg(t,l,\Delta, \varphi_{n,s,l,\Delta}^\scnf))\,.
        \end{equation}
\end{proof}

Theorem \ref{thm:mainthm2} in the introduction which uses a fixed prime field of size $p$ follows from Theorem \ref{theorem: main theorem} by setting $\psi_{d, d',n} = \Phi_{t,l,\Delta^\prime,n,s,\Delta}$, where $d = \Delta, d' = \Delta'$ and $s$ and $t$ are chosen to be appropriate polynomially bounded functions as in the statement of Theorem \ref{theorem: main theorem}. 


To get the no-short proof result against \ACZP-Frege we use the following simulation:
\begin{theorem}[Depth-preserving simulation of Frege systems by the Ideal proof system \cite{GP18}]\label{thm:cdIPS simulates AC0p-Frege}
    \label{Depth-preserving simulation of Frege systems by the Ideal proof system}
    Let $p$ be prime and $\mathbb{F}$ any field of characteristic $p$. Then $\IPS_\mathbb{F}$ p-simulates  $\ACZP$-Frege in such a way that depth-$d$ $\ACZP$-Frege proofs are simulated by depth-$O(d)$ $\IPS_\mathbb{F}$ proofs. In particular, $\ACZP$-Frege is p-simulated by bounded-depth $\IPS_\mathbb{F}$.
\end{theorem}

Therefore, a corollary of \Cref{theorem: main theorem} is (see the argument in \Cref{sec:framework-results} that comes after \Cref{thm:mainthm2}):

\begin{corollary}\label{cor:ACZ-p-stuff}
    For every prime $p$ there is an explicit sequence $\{\phi_n\}$ of DNF formulas (of unknown validity) such that there are no polynomial-size \ACZP-Frege proofs of $\{\phi_n\}$.
\end{corollary}




\subsection{Ruling Out Easiness for Diagonalizing CNFs is Necessary}

    We now turn to establishing Theorem \ref{thm:mainthm3}. We first define the notion of a ``reasonable" circuit class. 

\begin{definition}
We say $\mathcal{C}$ is a \emph{reasonable}
algebraic circuit class if $\mathcal{C}$-$\IPS$ can efficiently prove the $\mathcal{C}$-$\IPS$ analogue of Lemma \ref{lemma: grochow-pitassi formalization in cdIPS}, i.e., that $\mathcal{C}$ lower bounds for Permanent follow from $\mathcal{C}$-$\IPS$ lower bounds for $\CNF$ formulas. 
\end{definition}

The main result of \cite{ST25} can be seen as showing that the class of general algebraic circuits is reasonable, and \Cref{theorem: main theorem} shows that the class of constant-depth algebraic circuits is reasonable as well. The proof of \Cref{theorem: main theorem}  establishes that every natural algebraic class intermediate in power between constant-depth circuits and general circuits is reasonable as well.

We observe that Lemma \ref{lemma: unsat CNF has VNP refutation} can be generalised for an arbitrary algebraic circuit class $\mathcal{C}$.

\begin{theorem}[Grochow-Pitassi for $\mathcal{C}$]
    \label{theorem: grochow-pitassi for C}
    Let $\mathcal{C}$ be any algebraic circuit class. For any field $\mathbb{F}$, if $\mathcal{C}$-$\IPS$ is not p-bounded, namely there exists a super-polynomial lower bound on algebraic $\mathcal{C}$-$\IPS$ refutations (hence also on $\mathcal{C}$-$\IPS$ refutations) over $\mathbb{F}$ for a family of unsatisfiable $\CNF$ formulas $\{\phi_n\}_n$, then $\VNP_\mathbb{F} \neq \mathcal{C}_\mathbb{F}$.
\end{theorem}

\begin{proof}
    By \Cref{lemma: unsat CNF has VNP refutation}, any family of unsatisfiable $\CNF$ formulas $\{\phi_n\}_n$ has an $\IPS$ refutation that is computable in $\VNP$. However, since $\mathcal{C}$-$\IPS$ is not p-bounded, there exists a family of unsatisfiable $\CNF$ formulas $\{\phi_n\}_n$ that require super-polynomial size $\mathcal{C}$-$\IPS$ proofs. Hence no $\IPS$ refutation for the family $\{\phi_n\}_n$ is in $\mathcal{C}$. Thus we have that the $\VNP$ refutation is not in $\mathcal{C}$, and hence that $\VNP \neq \mathcal{C}$.
\end{proof}

Now, we show that the $\mathcal{C}$-analogue of  Theorem \ref{thm:mainthm2} is necessary to prove $\mathcal{C}$-$\IPS$ lower bounds for unsatisfiable CNFs, when $\mathcal{C}$ is a reasonable algebraic circuit class.

    We use $\{\Phi^\mathcal{C}_n\}$ to denote the $\mathcal{C}$ analogue of $\{\Phi_{t,l,\Delta^\prime,n,s,\Delta}\}$ from \Cref{theorem: main theorem}.

    \begin{theorem}[Necessity of the main theorem]
        \label{theorem: necessity}
        Let $\mathcal{C}$ be any reasonable algebraic circuit class. Let $p$ be a sequence of primes, and $\mathbb{F}_p$ be the prime field. If $\mathcal{C}$-$\IPS$ is not p-bounded over $\F_p$, which means there is a super-polynomial lower bound on algebraic $\mathcal{C}$-$\IPS$ refutations (hence also on $\mathcal{C}$-$\IPS$ refutations) over $\mathbb{F}_{p}$ for a family of unsatisfiable $\CNF$ formulas $\{\phi_n\}_n$, then the $\CNF$ family $\{\Phi^\mathcal{C}_n\}$ does not have polynomial-size $\mathcal{C}$-$\IPS$ refutations infinitely often over $\mathbb{F}_{p}$ in the following sense: 
        there exists a constant $c_1$ such that for every sufficiently large constant $c_2$ and every constant $c_0$, for infinitely many $n, t(n), s(n) \in \mathbb{N}$, $t(n) > 2^{(n^{c_1})}$ and $n^{c_1} < s(n) < n^{c_2}$, $\Phi^\mathcal{C}_n$ has no $\mathcal{C}$-$\IPS$ refutation over $\mathbb{F}_{p}$ of size at most $|\Phi^\mathcal{C}_n|^{c_0}$.
    \end{theorem}

    \begin{proof}
    Our definition of a reasonable algebraic circuit class $\mathcal{C}$ abstracts out the properties of $\mathcal{C}$ required to prove an {\it implication} from $\VNP \neq \mathcal{C}$ to super-polynomial lower bounds on $\Phi^\mathcal{C}_n$ against $\mathcal{C}$-$\IPS$, using essentially the same proof as for Theorem \ref{theorem: main theorem}. By Theorem \ref{theorem: grochow-pitassi for C}, super-polynomial $\mathcal{C}$-$\IPS$ lower bounds for any sequence $\{\phi_n\}_n$ of unsatisfiable CNFs implies $\VNP \neq \mathcal{C}$. The desired result follows from combining these two implications. 
    \end{proof}

    Theorem \ref{theorem: necessity} is the formal version of Theorem \ref{thm:mainthm3} in the Introduction.

\section{Supporting Evidence for the Diagonalizing CNFs as Unsatisfiable}\label{sec:supporting-evidence}

In this section, we present two results providing supporting evidence that the diagonalizing CNF $\Phi$ is unsatisfiable. 

\subsection{Tensor Rank Hardness Entails that $\Phi$ is Unsatisfiable}

Here we show that a lower bound against constant-depth IPS refutations of a formula expressing a tensor with tensor rank $m$ cannot be decomposed to the summation of $n$ rank-$1$ tensors, implying that the diagonalization formula $\Phi$ is unsatisfiable.

It is now known from \cite{LST25} that basic linear-algebraic operations such as matrix rank, determinant, and plausibly solving systems of linear equations cannot be efficiently carried out by constant-depth algebraic circuits, i.e., $\VAC^0$. The rank of an \(n \times n\) matrix over a field is only known to be computable in uniform \NCTwo\ via Mulmuley’s parallel algorithm~\cite{Mul86} (in fact, integer determinant is computable already in $\#\SAC^1$ \cite{Coo85}). In contrast, determining the rank of a 3-dimensional tensor is known to be NP-hard~\cite{Has90, Shi16}.

    The following is the rank principle as studied in \cite{GGRT25} (cf.~\cite{SU04,Kra09,GGPS23}).

    \begin{definition}[Rank Principle]
        \label{def: rank principle}
        Let $\F$ be a field and $m,n$ be two positive integers such that $m > n$. We denote $\RankP^m_n(A)$ the system of degree-$2$ polynomial equations stating that the rank of an $m \times m$ matrix $A$ is at most $n$ over $\F$. More precisely, the \emph{rank principle} $\RankP_n^m(A)$ is defined over $2mn$ variables arranged into two matrices $X \in \F^{m \times n}$ and $Y \in \F^{n \times m}$, For every $i,j \in [m]$, there is an equation in $\RankP_n^m(A)$ stating that the $(i,j)$th entry of the product $XY$ is equal to $A_{i,h}$. That is
        \begin{equation}
            \label{equa: rank principle}
            \sum_{k \in [n]} x_{i,k} y_{k,j} - A_{i,j}, \qquad i,j \in [m].
        \end{equation}
    \end{definition}
    By linear algebra, when the rank of $A$ exceeds $n$, $\RankP_n^m(A)$ is unsatisfiable.
    The work of \cite{GGLST25} considers a generalisation of the rank principle, which they call the tensor rank principle.

    \begin{definition}[Tensor Rank Principle]
    \label{def: tensor rank principle}
    Let $\F$ be a field and $m, n$ be two positive integers such that $m > n$. We denote $\TRankP_{m,n}^r(A)$ the system of degree-$r$ polynomial equations stating that the tensor rank of the $r$-tensor $A \in \F^{(\overbrace{m \times \cdots \times m}^{r \text{ times}})}$ is at most $n$ over $\F$. 
    More precisely, the \emph{$r$th-order tensor rank principle} $\TRankP_{m,n}^r(A)$ is defined over $rmn$ variables arranged into $r$ matrices $X_1, \dots, X_r \in \{0,1\}^{m \times n}$ where each matrix $X_i$ can be viewed as $n$ vectors $\x_{j,k} \in \{0,1\}^m$ for $k \in [n]$ and $j \in [r]$. For every $i_1, \dots, i_r \in [m]$ (not necessarily distinct), there is an equation in $\TRankP_{m,n}^r (A)$ stating that the $(i_1, \dots, i_r)$th entry of the summation of the tensor product $\bigotimes_{j=1}^r \x_{j,k}$ is equal to $A_{i_1, \dots, i_r}$ (where $\otimes$ denotes the tensor (outer) product of vectors). That is,
    \begin{equation}
        \label{equa: tensor rank principle}
        \sum_{k \in [n]} \prod_{j \in [r]} x_{j,i_j, k} = A_{i_1, \dots, i_r}, \qquad i_1, \dots, i_r \in [m],
    \end{equation}
    where $x_{j,i_j, k}$ denote the $(i_j,k)$th entry of the matrix $X_j$.

    Additionally, for every $i \in [m]$, $k \in [n]$ and $j \in [r]$, we have a Boolean axiom $x_{j,i,k}^2 - x_{j,i,k}$.
    Namely, $\TRankP_{m,n}^r(A)$ is the set of polynomial equations stating that
    \[
        \sum_{k \in [n]} \bigotimes_{j=1}^r \x_{j,k} = A.
    \]
\end{definition}
\bigskip 
    
    Note that $\bigotimes_{j=1}^r \x_{j,k}$ is the $r$th order tensor of rank $1$ obtained by the outer product of all the $k$th columns in $X_1,\dots,X_r$. Hence, by basic algebra, it follows that the sum of $n$ rank 1 tensors $\sum_{k \in [n]} \bigotimes_{j=1}^r \x_{j,k}$ has tensor rank at most $n$. Thus, whenever the tensor rank of $A$ exceeds $n$, the formula $\TRankP_{m,n}^r(A)$ is unsatisfiable. For standard background on tensor rank and tensor decompositions, see the survey by Kolda and Bader~\cite{KB09}. As noted earlier, the tensor rank principle generalises the matrix equation $XY = A$, which corresponds to the rank principle. In fact, the rank principle $\RankP_n^m(A)$ is precisely the case $r=2$ of the tensor rank principle, i.e., $\TRankP_{m,n}^2(A)$.

    \cite{GGLST25} proved that $\TRankP_{m,n}^r(A)$ requires $\exp(\Omega(n))$-size $\PCR$ refutations over the two-element field.

    \begin{theorem}[\cite{GGLST25}]
        \label{theorem: exp lower bound for TRankP}
        Every $\PCR$ refutation over $\F_2$ of $\TRankP_{m,n}^r(A)$ requires $2^{cn}$ monomials for some constant $c$.
    \end{theorem}
    
    They also exhibit a reduction from the tensor rank principle to bounded-depth algebraic circuit upper bounds statements (so that the latter are at least as hard as the former), specifically the statement “$\perm\in\VAC^0$.”
    

     \begin{definition}[Bounded-depth algebraic circuit upper bound formulas \mbox{$\AUB$}]
        \label{def: Constant-depth algebraic circuit upper bound formula}
         Let $f(\overline{x})\in \mathbb{F}[\overline{x}]$ be a polynomial with $n$-variables and degree at most $l$. The following set of polynomial equations $\AUB(f,s,l,\Delta)$ (in the $\overline{w}$-variables only) express that the polynomial $f(\overline{x})$ can be computed by a bounded-depth algebraic circuit of size $s$ and depth $\Delta$: 
         \[
            \{\coeff_{M_i} (U(\overline{x},\overline{w})) = b_i : 1\leq i \leq N\},
        \]
        where $\overline{b}=\coeff(f(\overline{x})) \in \mathbb{F}^N$ is the coefficient vector of the polynomial $f$ of dimension $N$, $U(\overline{x},\overline{w})$ is the constant-depth universal circuit for polynomials of depth at most $\Delta$ and size at most $s$, $\overline{w}$ are the $K_{s,\Delta}$ (circuit) edge variables, $\{M_i\}_{i=1}^N$ is the set of all possible $\overline{x}$-monomials of degree at most $l$, and $N=\Sigma_{j=0}^l \binom{n+j-1}{j}=2^{O(n+l)}$ is the number of monomials of total degree at most $l$ over $n$ variables. The size of the above set of polynomial equations is $O(2^{(t+l)l}\cdot |U(\overline{x},\overline{w})| \cdot N)$ where $t$ is the maximum multiplication fan-in in $U(\overline{x},\overline{w})$.\par 
    \end{definition}

    \begin{theorem}[Constant-depth algebraic circuit upper bounds are at least as hard as tensor rank principle \cite{GGLST25}]
        \label{theorem: reduction to TRankP}
        Suppose $\SLP(\AUB(f,s,\ell,\Delta))$ admits size-$S$ depth-$\Delta^\prime$ $\IPS$ refutations where $\Delta > 4$ and $s > n2 \ell^4$.
        Then, $\TRankP_{N,\sqrt{s}}^n(A)$ admits size-$O(Nn \cdot S + NK \cdot |\SLP(\AUB(f, s, l, \Delta))|)$ depth-$(\Delta^\prime+5)$ $\IPS$ refutations where
        \begin{itemize}
            \item $N = \binom{n+\ell}{\ell}$ is the number of monomials in $n$ variables and total degree at most $\ell$.
            \item $A$ is the $n$-tensor such that the $(\calM, \dots, \calM)$th entry of $A$ is $\coeff_\calM(p)$ where $\calM$ is a monomial in $n$ variables and total degree at most $\ell$. The rest of $A$ are all zeros. 
        \end{itemize}
    \end{theorem}

    \begin{corollary}[\cite{GGLST25}]
    \label{Cor:6.6}
        Suppose $\TRankP_{m,n}^r(A)$ requires $2^{n^\delta}$-size depth-$\Delta^\prime$ $\IPS$ refutations, for some constant $\delta>0$. Then, $\SLP(\AUB(f,s,\ell,\Delta))$ requires $|\SLP(\AUB(f,s,\ell,\Delta))|^{\omega(1)}$-size, depth-$(\Delta^\prime - 5)$ $\IPS$ refutation.
    \end{corollary}

    In particular, this shows that the hardness of the tensor rank principle against constant-depth IPS entails the unsatisfiability of the diagonalizing formula $\Phi$ from \Cref{sec:main-theorem-sec}.

    \subsection{Unconditional $\PCR$ Lower Bounds for Algebraic Circuit Upper Bound Formulas}

    In the previous section, we showed that the tensor rank principle $\TRankP_{m,n}^r(A)$ can be reduced to the \emph{bounded-depth} algebraic circuit upper bound formulas $\AUB(f,s,l,\Delta)$. In this section, we show that a variant of the rank principle can be reduced to \emph{general} (unbounded depth) algebraic circuit upper bound formulas $\AUB(f,s,l)$. This yields an \emph{unconditional} $\PCR$ lower bound for $\AUB(f,s,l)$. 

    The definition of algebraic circuit upper bound formulas is similar to the bounded-depth case, except that the universal circuit for bounded-depth is replaced with the universal circuit for general algebraic circuits, as defined in~\cite{Raz10} (see~\cite{ST25}). Accordingly, we adopt the same notation for general algebraic circuit upper bound formulas as for the bounded-depth case, except that the depth parameter is omitted.

    \begin{definition}[Algebraic circuit upper bound formula]
        \label{def: Algebraic circuit upper bound formula}
         Let $f(\overline{x})\in \mathbb{F}[\overline{x}]$ be a polynomial with $n$-variables and degree at most $l$. The following set of polynomial equations $\AUB(f,s,l)$ (in the $\overline{w}$-variables only) expressing that the polynomial $f(\overline{x})$ can be computed by an algebraic circuit of size $s$: 
         \[
            \{\coeff_{M_i} (U(\overline{x},\overline{w})) = b_i : 1\leq i \leq N\},
        \]
        where $\overline{b}=\coeff(f(\overline{x})) \in \mathbb{F}^N$ is the coefficient vector of the polynomial $f$ of dimension $N$, $U(\overline{x},\overline{w})$ is the universal circuit for polynomials of degree at most $l$ and circuit size at most $s$, $\overline{w}$ are the $K_{s,l}$ edge variables, $\{M_i\}_{i=1}^N$ is the set of all possible $\overline{x}$-monomials of degree at most $l$, and $N=\Sigma_{j=0}^l \binom{n+j-1}{j}=2^{O(n+l)}$ is the number of monomials of total degree at most $l$ over $n$ variables. The size of $\AUB(f,s,l)$ is $O(7^l\cdot |U(\overline{x},\overline{w})| \cdot N)$.\par
    \end{definition}

    Now, we define an \emph{iterated} version of the rank principle. Let $\F_L$ be a finite field of characteristic $L$. We will denote by $\F_L^{\leq n}$ the set of vectors over $\F_L$ of length at most $n$. Similarly, $\F_L^{< n}$ denotes the set of vectors over $\F_L$ of length less than $n$. Let $L, n, K  \in \mathbb{N}^+$. For every vector $\pi \in \F_L^{\leq n}$, let $X^{\pi} = (x_{i,k}^\pi )_{i \in [LK], k \in [K]}$ be variable matrices (unique for each different $\pi$). Let $Y$ be a $K \times LK$ matrix in variables $y_{k,j}$ for $k \in [K], j \in [LK]$. For every vector $\pi \in \F_L^n$, let $A^\pi $ be an $LK \times LK$ matrix over $\F_L$, and $\{A^\pi\}$ be the set consisting of all $A^\pi$ over all vectors $\pi$.

    \begin{definition}[Iterated Rank Principle \cite{GGRT25}]
        \label{def: Iterations of Rank Principle}
        Let $L,n,K$ be parameters in $\mathbb{N}^+$. The \emph{Iterated Rank Principle} is $\IRankP_{L,n,K}(\{A^\pi\}) \coloneqq$
        \begin{align*}
            &\sum_{k \in [K]} x_{i,k}^\pi y_{k,j} - x_{i, j- (\lceil \frac{j}{K} \rceil-1) K}^{\pi (\lceil \frac{j}{K} \rceil-1)}, \qquad \forall \pi \in \F_L^{< n}, i \in [LK], j \in [LK],\\
            &\sum_{k \in [K]} x_{i,k}^\pi y_{k,j} - A_{i,j}^\pi, \qquad \forall \pi \in \F_L^n ,i \in [LK], j \in [LK].
        \end{align*}
        Namely, $\IRankP_{L,n,K}(\{A^\pi\})$ contains all the degree-$2$ polynomial equations in the following matrix multiplications (where $\pi b$, for $b\in\F_L$, denotes concatenation of $b$ to $\pi$):
        \begin{align}
             & X^{\pi} Y = [X^{\pi  0} X^{\pi 1} \cdots X^{\pi  (L-1)}], \quad \forall \pi \in \F_L^{\leq n}, \label{eq:2810}\\
             & X^{\pi} Y = A^{\pi}, \quad \forall \pi \in \F_L^n . \label{eq:2811}
        \end{align}
    \end{definition}
    Note that \Cref{eq:2811} are instances of the rank principle, hence the Iterated Rank Principle is no stronger than the rank principle. Intuitively, \Cref{eq:2810} provides auxiliary axioms that allow us to access the nodes of a tree, in which we can embed a circuit.

    \begin{definition}[Iterated Rank Principle with Extension Variables]
        \label{def: iterations of rank principle with extension variables}
        Let $L,n,K$ be parameters in $\mathbb{N}^+$. The \emph{Iterated Rank Principle with Extension Variables} denoted $\IRankPE_{L,n,K}(\{A^\pi\}) $ consists of the equations in $ \IRankP_{L,n,K}(\{A^\pi\}) $ together with
        \begin{align*}
            z_{i,k,j}^\pi - x_{i,k}^\pi y_{k,j}, \qquad \forall \pi \in \F_L^{< n}, i \in [LK], k \in [K], j \in [LK].
        \end{align*}
    \end{definition}

    \begin{theorem}[\cite{GGRT25}]
        \label{theorem: exp lower bound for IRankP}
        Every $\PCR$ refutation over $\F_2$ of $\IRankPE_{L, n ,K}(\{A^\pi\})$ requires $2^{K^\delta}$ monomials for some constant $\delta>0$.
    \end{theorem}
    
    Note that the size of $\IRankPE_{L, n ,K}(\{A^\pi\})$ is exponential in $n$. Hence, if we choose $K$ large enough (e.g., $K\ge n^{100}$) this theorem gives a super-polynomial $\PCR$ lower bound.
    \cite{GGLST25} further showed, through a reduction from the iterated rank principle with extension variables to (the SLP \Cref{def: Straight line program SLP} version of) $\AUB(f,s,l)$, that the latter admits an unconditional size lower bound against $\PCR$.

    \begin{theorem}[\cite{GGLST25}]
        \label{theorem: super-poly lower bound for AUP}
        Every $\PCR$ refutation over $\F_2$ of $\SLP(\AUB(f,s, l))$ where $s$ is polynomially bounded requires $|\SLP(\AUB(f,s, l))|^{\omega(1)}$ many monomials.
    \end{theorem}

\section{No Short Bounded-Depth IPS Refutations for Diagonalizing CNF Formulas: the Polynomial-Size Fields Case}

In this section, we work over finite fields whose characteristic is polynomially bounded by the instance size, in order to obtain a more general result. This contrasts with the previous section, where the characteristic of the finite field was a fixed global constant, independent of the instance size. Here we encode binary string arithmetic—including addition, multiplication, and modular computation into $\CNF$ formulas. Within this setting, we also establish the translation lemma, which yields a version of \Cref{theorem: main theorem} over finite fields of polynomially bounded characteristic. In other words, we show that no efficient provability result holds against constant-depth $\IPS$ over polynomial-size finite fields. Since Forbes~\cite{forbes2024low} extended \cite{LST25} to arbitrary fields, the results of this section generalize \Cref{theorem: main theorem}. The techniques here, however, are somewhat more involved.

We will continue to use notations such as $\cnf(C(\overline{x})=0)$ and $\ecnf(C(\overline{x})=0)$ from the fixed-field setting, but with a different interpretation: in this section they refer to the $\CNF$ and Extended $\CNF$ encodings of $C(\overline{x})=0$ obtained via bit-level arithmetic, which we describe below.

\subsection{Bit Arithmetic}
Field elements are encoded in standard binary representation. We work over the finite field $\mathbb{F}_q$. Note that the characteristic of $\mathbb{F}_q$ is not constant. More precisely, we work over a finite field $\mathbb{F}_q$ where $q$ may grow polynomially (with the input size). 
    \begin{definition}[The encoding of binary value $\VAL$]
        \label{def: binary value VAL}
        Given a bit vector $w_{k-1},\dots,w_0$ of variables $w$, denoted $\overline{w}$, ranging over 0-1 values, define the following algebraic formulas:
        \begin{align*}
            w&=\VAL(\overline{w})\\
            \VAL(\overline{w}) &= \Sigma_{i=0}^{k-1} 2^i \cdot w_i,
        \end{align*}
        where $\VAL(\overline{w})$ is an extension variable. The size of a $\VAL$ is $O(k)$. 
    \end{definition}
    \begin{definition}[Arithmetization operation $\arit(\cdot)$] 
        \label{def: arithmetization operation arit}
        For a variable $y_i$, $\arit(y_i)=y_i$. For the truth value false $\bot$ and true $\top$ we put $\arit(\bot) \coloneqq  0$ and $\arit(\top) \coloneqq  1$. For logical connectives we define $\arit(A \land B) \coloneqq \arit(A) \cdot \arit(B)$, $\arit(A \lor B) \coloneqq 1- (1-\arit(A))\cdot (1-\arit(B))$, and for XOR operation we define $\arit(A\oplus B) \coloneqq \arit(A)+\arit(B)-2 \cdot \arit(A) \cdot \arit(B)$.
    \end{definition}
    In this way, for every Boolean formula $F(\overline{x})$ with $n$ variables and a Boolean substitution $\overline{\alpha} \in \{0,1\}^n$,
    $\arit(F)(\overline{\alpha})=1$ if and only if $F(\overline{\alpha})=\top$.

    We will present the $\CNF$ encoding of unbounded fan-in algebraic circuits using bit arithmetic. However, for simplicity, we present the algebraic encoding $\phi$ that is "equivalent`` to the $\CNF$ formula $ F$. For "equivalent", we mean that $\phi$ can be derived from $F$ in constant depth and constant size, and vice versa. The reason why this can be achieved is that each formula has only a constant number of variables. Therefore, by the implicational completeness of $\IPS$ over 0-1 assignment, $\IPS$ can derive all formulas simultaneously in constant depth.  It is easy to see that all formulas we give below can be written as $\CNF$s. Also, we write $\phi-\psi =0$ as $\phi= \psi$.

    We divide the addition bit arithmetic in a finite field into two big steps:
    \begin{itemize}
        \item Addition Step: for two $k$ length binary representations $\overline{x}=x_{k-1},\dots,x_0$ and $\overline{y}=y_{k-1},\dots,y_0$, we do the general addition bit arithmetic which outputs a $k+1$ length binary representations $\add_k,\dots, \add_0$.
        \item Modular Step: we turn this $k+1$ length binary representations $\add_k,\dots, \add_0$ into a $k$ length binary representation $\add_{k-1}^\prime ,\dots,\add_0^\prime$, which represents the same number in finite field $\mathbb{F}_q$. In other words, $\sum_{i=0}^k 2^i \add_i \equiv \sum_{i=0}^{k-1} 2^i \add_i^\prime \mod{q}$. Note that, it is possible that $\sum_{i=0}^{k-1} 2^i \add_i^\prime \geq q$. 
    \end{itemize}
    For the convenience of later use of the Modular step in multiplication bit arithmetic in a finite field, we generalize our Modular step into the following:
    \begin{itemize}
        \item Given a $k$ length binary representations $\overline{x} = x_{k-1}, \dots, x_0$ and an extra bit $x^\prime$ which is the $t+1$th bit where $t \geq k$. We do the addition bit arithmetic of $0,x_{k-2} ,\dots, x_0$ and the $k$ length binary representation of $(2^t x^\prime +2^{k-1} x_{k-1}) \mod{q}$, which outputs a $k+1$ length binary representation $\overline{y} = y_k, \dots, y_0$.
        \item Then, we do the addition bit arithmetic of $y_{k-1}, \dots, y_0$ and $(2^k y_k \mod{q})$, which outputs a $k-1$ length binary representation. The reason why the output must be smaller than $2^k$ is as follows: 
    
        Notice that, both two items $\sum_{i=0}^{k-2} 2^i x_i $ and $(2^t x^\prime +2^{k-1} x_{k-1}) \mod{q}$ are smaller than $q$ since $\sum_{i=0}^{k-2} 2^i x_i < 2^{k-1} \leq q$. If $\sum_{i=0}^{k-1} 2^i y_i + 2^k y_k < 2^k $, then we are done. Otherwise, suppose $\sum_{i=0}^{k-1} 2^i y_i + 2^k y_k \geq 2^k $. Since $\sum_{i=0}^{k-1} 2^i y_i \leq 2^k -1$, $\sum_{i=0}^{k-1} 2^i y_i + 2^k y_k \geq 2^k $ implies that $y_k =1$. Hence,
        \begin{align*}
            \sum_{i=0}^{k-1} 2^i y_i + (2^k y_k \mod{q})  &\leq \sum_{i=0}^{k-1} 2^i y_i + 2^k y_k -q\\
            & = \sum_{i=0}^{k-2} 2^i x_i+ (2^t x^\prime +2^{k-1} x_{k-1} \mod{q}) -q\\
            & < 2q- q \\
            & = q\\
            & < 2^k
        \end{align*}
        Therefore, we know that the output of this step must be a  $k-1$ length binary representation.
    \end{itemize}

    Now, we define the formula for the addition step.
    
    \begin{definition}[The carry bit $\carry_i$, the addition bit $\add_i$ and the encoding of carry lookahead addition $\addition$] 
        \label{def: carry-add}
        Suppose we have two $k$ length binary representations $\overline{x}=x_{k-1},\dots,x_0$ and $\overline{y}=y_{k-1},\dots,y_0$. We define the carry bit $\carry_i$, the addition bit $\add_i$ and the encoding of carry lookahead addition $\addition(\overline{x},\overline{y},\overline{\add})$ as follows, together with the Boolean axioms for each variable:
        \begin{equation*}
            \carry_i(\overline{x},\overline{y}) = \begin{cases}
                \arit((x_{i-1}\land y_{i-1}) \lor ((x_{i-1}\lor y_{i-1})\land \carry_{i-1}(\overline{x},\overline{y}))), &i=1,\dots, k; \\
                \arit(\bot), &i=0,
            \end{cases}
        \end{equation*}
        and
        \begin{align*}
            \add_i(\overline{x},\overline{y}) &= \arit(x_i \oplus y_i \oplus \carry_i(\overline{x},\overline{y})), \qquad i=0,\dots, k-1\\
            \add_k(\overline{x},\overline{y}) &= \carry_k(\overline{x},\overline{y})
        \end{align*}
        The size of the encoding $\addition(\overline{x},\overline{y},\overline{\add})$ is also $O(k)$. The encoding $\addition(\overline{x},\overline{y},\overline{\add})$ represents the addition of two $k$ length binary representations $\overline{x}$ and $\overline{y}$, which gives a $k+1$ length binary representation $\overline{\add}$.
        
        Also, we denote $\addition^\prime$ as encoding of carry lookahead addition without $\carry_k$ and $\add_k$. In other words, $\addition^\prime(\overline{x},\overline{y},\overline{\add})$ represents the addition of two $k$ length binary representations $\overline{x}$ and $\overline{y}$, which gives a $k$ length binary representation $\overline{\add}$.
    \end{definition}
    
    Notice that both $\carry_i$ and $\add_i$ are extension variables. Now, we define the encoding of modular.
    \begin{definition}[The encoding of modular $\modular$]
        \label{def: modular}
        Suppose we have a $k$ length binary representation $\overline{x} = x_{k-1} ,\dots, x_0$ and another binary bit $x^\prime$ which is the $t+1$th bit in a binary representation, the encoding of modular $\modular^t(\overline{x},x^\prime, \overline{\add^\prime})$ is defined as follows, together with the Boolean axioms for each variable:
        \[
            \addition(\overbrace{0, x_{k-2},\dots,x_0}^{\mathclap{\text{the last $k$ bit of $x$ with $x_{k-1}$ exchanged to $0$}}}, \underbrace{\overline{m}}_{\mathclap{\text{the $k$ length binary representation of $2^t x^\prime +2^{k-1} x_{k-1} \mod{q}$ }}}, \overline{u})
        \]
        \[
            \addition^\prime(\overbrace{u_{k-1},\dots,u_0}^{\mathclap{\text{the last $k$ bit of $u$}}}, \underbrace{\overline{w}}_{\mathclap{\text{the $k$ length binary representation of $2^k u_k$}}},\overline{\add^\prime})
        \]
        
        $\overline{m}$ and $\overline{w}$ are defined as follows. Suppose $\overline{a}$ is the binary representation of $2^t \mod{q}$, $\overline{b}$ is the binary representation of $2^{k-1} \mod{q}$ and $\overline{c}$ is the binary representation of $2^t + 2^{k-1} \mod{q}$, each bit $m_i$ in $\overline{m}$ can be computed by a Boolean function $f_i^t(x^\prime, x_{k-1})$ whose truth table is as follows:\\
        \begin{center}
            \begin{tabular}{c|c|c}
                \hline
                $x^\prime$  &  $x_{k-1}$ & $f_i^t(x^\prime, x_{k-1})$ \\
                \hline
                $0$ & $0$ & $0$\\
                \hline
                $0$ & $1$ & $b_i$ \\
                \hline
                $1$ & $0$ & $a_i$ \\
                \hline
                $1$ & $1$ & $c_i$ \\ 
                \hline
            \end{tabular}
        \end{center}

        Therefore, $f_i^t(x^\prime, x_{k-1})$ can be represented by a $\CNF$ with variables $x^\prime$ and $x_{k-1}$ since all $\overline{a},\overline{b}$ and $\overline{c}$ are fixed. We denote such $\CNF$ as $\cnf(f_i^t(x^\prime, x_{k-1}))$. Hence, $\overline{m}$ are defined as $m_i \coloneqq \arit(\cnf(f_i^t(x^\prime, x_{k-1})))$.

        Similarly, $\overline{w}$ are defined as $w_i \coloneqq \arit(\cnf(g_i(u_k)))$ where $g_i$ is another Boolean function only depends on $u_k$.
    \end{definition}

    With the encoding of carry lookahead addition $\addition$ and the encoding of modular $\modular$, we define the $\CNF$ encoding of $x+y=z$.
    \begin{definition}[The $\CNF$ encoding of bit arithmetic for addition in a finite field]
        \label{def: cnf addition}
        For any $\SLP$ formula $x+y-z=0$, the $\CNF$ encoding of bit arithmetic for addition in a finite field, denoted as $\CNF$-$\add(\overline{x},\overline{y},\overline{z})$, is as follows:
        \begin{itemize}
            \item Addition step: $\addition(\overline{x},\overline{y}, \overline{\add})$
            \item Modular step: $\modular^{k}(\add_{k-1}, \dots, \add_0, \add_k, \overline{\add^\prime})$
            \item Connection step: for each $0 \leq i \leq k-1$, we have $\add_i^\prime = z_i$.
        \end{itemize}
    \end{definition}
    \bigskip

    For the multiplication bit arithmetic in a finite field, we divide it into three big steps:
    \begin{itemize}
        \item Multiplication Step: For two $k$ length binary representations $\overline{x}= x_{k-1}, \dots, x_0$ and $\overline{y}=y_{k-1},\dots,y_0$, we do the general multiplication bit arithmetic which outputs $k$ many $k$ binary representations $\overline{s_0}, \dots, \overline{s_{k-1}}$ where $s_{ij}=x_j \land y_i$.
        \item Modular Step: for $\overline{s_i}$, we do the generalized modular step for $i$ times, which gives us a $k$ length binary representation whose value is the same as $\sum_{j=i}^{i+k} 2^j s_{ij}$ in the finite field $\mathbb{F}_q$.
        \item Addition Step: After the Modular Step, we have $k$ many $ k$-length binary representations. Using the addition we showed above, we can add them together in the finite field $\mathbb{F}_q$ and get a $k$-length binary representation.
    \end{itemize}
    
    \begin{definition}[The encoding of multiplication $\mult_i$, $\mult$] 
        \label{def: prod, mult}
        Suppose we have two binary representations $\overline{x}=x_{k-1},\dots,x_0$ and $\overline{y}=y_{k-1},\dots,y_0$, we define the encoding of multiplication $\mult(\overline{x},\overline{y},\overline{s_0},\cdots, \overline{s_{k-1}})$ as follows, together with the Boolean axioms for each variable:
        \begin{equation*}
            \mult_i(\overline{x},\overline{y},\overline{s_i}) \coloneqq \begin{cases}
                s_{i,j} = \arit(x_{j-i} \land y_i), &  i \leq j \leq k-1+i\\
                s_{i,j} = \arit(\bot), & 0 \leq j < i.
            \end{cases},\qquad  0 \leq i \leq k-1
        \end{equation*}
        where $\overline{s_i}$ is a $k+i$ length 0-1 vector.
        The size of the encoding of multiplication $\mult(\overline{x},\overline{y},\overline{s_0},\cdots, \overline{s_{k-1}})$ is $O(k^2)$.
    \end{definition}

    \begin{definition}[The $\CNF$ encoding of bit arithmetic for multiplication in a finite field]
         \label{def: cnf multiplication}
        For any $\SLP$ formula $x\times y-z=0$, the $\CNF$ encoding of bit arithmetic for multiplication in a finite field denoted as $\CNF$-$\mult(\overline{x},\overline{y},\overline{z})$, is as follows:
        \begin{itemize}
            \item Multiplication step: $\mult(\overline{x},\overline{y},\overline{s_0},\cdots, \overline{s_{k-1}})$
            \item Modular step: for each $\overline{s_i}$ such that $1 \leq i \leq k-1$, we have the following $i$ many modular formula:
            \begin{align*}
                &\modular^k(s_{i,k-1},\dots,s_{i,0},s_{i,k}, \overline{u_{i,k}}) \\
                &\modular^{k+1}(\overline{u_{i,k}}, s_{i,k+1}, \overline{u_{i,k+1}}) \\
                &\vdots\\
                &\modular^j(\overline{u_{i,j-1}}, s_{i,j}, \overline{u_{i,j}}), \qquad k+1 \leq j \leq i+k-1.
            \end{align*}
            All $\overline{u_{i,j}}$ is a $k$ length binary representation.
            \item Addition step: Now for each $\overline{s_i}$, we have a $k$ length binary representation $\overline{u_{i,i+k-1}}$ such that $\sum_{j=0}^{i+k-1}2^j s_{i,j}=\sum_{j=0}^{k-1} 2^j u_{i,i+k-1,j}$ in the finite field $\mathbb{F}_q$. Then, using the $\CNF$ encoding of bit arithmetic for addition in a finite field as we defined above in \Cref{def: cnf addition}, we have the following:
            \begin{align*}
                &\CNF\text{-}\add(\overline{s_0},\overline{u_{1,k}},\overline{v_1})\\
                &\CNF\text{-}\add(\overline{v_1},\overline{u_{2,k+1}},\overline{v_2})\\
                &\vdots\\
                &\CNF\text{-}\add(\overline{v_i},\overline{u_{i+1,k+i}},\overline{v_{i+1}}),\qquad 1 \leq i \leq k-2
            \end{align*}
            \item Connection step: for each $0 \leq j \leq k-1$, we have $v_{k-1,j} = z_j$.
        \end{itemize}
        The size of $\CNF$-$\mult(\overline{x},\overline{y},\overline{z})$ is $O(k^3)$.
    \end{definition}
    \begin{definition}[$\CNF$ encoding of unbounded fan-in algebraic circuits; $\cnf(C(\overline{x}))$] 
        \label{def: cnf encoding of algebraic circuit}
        Let $C(\overline{x})$ be an (unbounded fan-in) algebraic circuit in variables $\overline{x}$. The $\CNF$ encoding of $C(\overline{x})$ denoted $\cnf(C(\overline{x}))$ consists of the following $\CNF$s in the binary representation bits variables of all the nodes in $C$ and extra extension variables (and only in the binary representation bits variables):
        \begin{itemize}
            \item If $\alpha \in \mathbb{F}$ is a scalar input node in $C$, the $\CNF$ encoding of $C$ contains the $\{0,1\}$ constant corresponding to the binary representation bits of $\alpha$. These constants are used when fed to nodes according to the wiring of $C$.
            \item For every node $g$ in $C(\overline{x})$, suppose $g$ is a $+$ node that has inputs $u_1, \dots, u_t$. Then, firstly, we have the formula $\CNF\text{-}\add$ for each of the following equations:
            \begin{align*}
                &u_1+ u_2 = v_1^g \\
                &u_{i+2} + v_{i}^g = v_{i+1}^g, \qquad 1 \leq i \leq t-3 \\
                &u_{t} + v_{t-2}^g = g.
            \end{align*}
            Same for $\times$ nodes. Suppose $g$ is a $\times$ node that has inputs $u_1, \dots, u_t$. Then, we have the multiplication formula $\CNF$-$\mult$ for each following equations:
            \begin{align*}
                &u_1\times u_2 = v_1^g \\
                &u_{i+2} \times v_{i}^g = v_{i+1}^g, \qquad 1 \leq i \leq t-3 \\
                &u_{t} \times v_{t-2}^g = g.
            \end{align*}
            \item For every Boolean variables $u$, we have the Boolean axiom
            \begin{equation*}
                u_i^2 - u_i =0
            \end{equation*}
        \end{itemize}
        
        We call variables $v_i^g$ the \emph{intermediate nodes} which are nodes that do not exist in $C$ but are used to help to encode and the $\SLP$ formula $u_{i+2} + v_{i}^g = v_{i+1}^g$ or $u_{i+2} \times v_{i}^g = v_{i+1}^g$ the \emph{intermediate $\SLP$ formulas}. The size of the $\CNF$ encoding of unbounded fan-in algebraic circuits $\cnf(C(\overline{x}))$ is $O(tk^3|C|)=O(k^3|C|^2)$ where $t$ is the maximum fan-in.
    \end{definition}
    Notice that, in the above $\CNF$ encoding, we can only guarantee that the output of the circuit, which is a $k$ length binary representation $\overline{g_{out}}$, is equal to zero in the finite field $\mathbb{F}_q$. In other words, we can guarantee that $\sum_{i=0}^{k-1} 2^i g_{out,i} =0$. However, this does not mean that $g_{out,i}=0$ for all bits. Instead of guaranteeing each $g_{out,i}$ is equal to $0$, we guarantee that each $g_{out,i}=q_i$ where $\overline{q}$ is the $k$ length binary representation of $q$. Therefore, to encode the output of the algebraic circuit is equal to $0$, we add $q$ to the output, which is 
    \begin{equation*}
        \CNF\text{-}\add(\overline{g_{out}}, \underbrace{\overline{q}}_{\mathclap{\text{the binary representation of $q$}}}, \overline{g_{out}}).
    \end{equation*}

    \begin{definition}[$\CNF$ encoding of unbounded fan-in algebraic circuit equations; $\cnf(C(\overline{x}))=0$]
        \label{def: cnf encoding of algebraic circuit equations}
         Let $C(\overline{x})=0$ be a circuit equation in the variables $\overline{x}$ over a finite field $\mathbb{F}_q$. The $\CNF$ encoding of it denoted $\cnf(C(\overline{x})=0)$ consists of the $\CNF$ encoding of $C(\overline{x})$ from \Cref{def: cnf encoding of algebraic circuit} together with the equations
        \begin{itemize}
            \item $\CNF=\add(\overline{g_{out}}, \underbrace{\overline{q}}_{\mathclap{\text{the binary representation of $q$}}}, \overline{g_{out}}) $
            \item $ g_{out,i} =q_i, \quad 0 \leq i \leq k-1$ which express that the output node $g_{out} = 0$ and $g_{out,i}$ are the binary representation bits of $g_{out}$. We call these formulas the connection formulas for the output node. 
        \end{itemize}
         The size of the $\CNF$ encoding of unbounded fan-in algebraic circuits equation $\cnf(C(\overline{x})=0)$ is $O(tk^3|C|)=O(k^3|C|^2)$ where $t$ is the maximum fan-in.
    \end{definition}
    \begin{definition}[Extended $\CNF$ encoding of unbounded fan-in algebraic circuit equation (circuit \resp); $\ecnf(C(\overline{x})=0)$ ($\ecnf(C(\overline{x}))$, \resp)]
        \label{def: ecnf encoding of algebraic circuit (equations)}
         Let $C(\overline{x})$ be a circuit in the $\overline{x}$ variables over the field $\mathbb{F}_q$. Then the extended $\CNF$ encoding of the circuit equation $C(\overline{x})=0$ (circuit $C(\overline{x})$, \resp), in symbols $\ecnf(C(\overline{x})=0)$ ($\ecnf(C(\overline{x}))$, \resp), is defined to be the following:
        \begin{itemize}
            \item the $\CNF$ encoding of circuit equation $C(\overline{x})=0$ (circuit $C(\overline{x})$, \resp), namely, $\cnf(C(\overline{x})=0) (\cnf(C(\overline{x})), \resp)$; and
            \item the binary value formula for each binary representation bits vector.
        \end{itemize}
        The size of the Extended $\CNF$ encoding of unbounded fan-in algebraic circuits equation $\ecnf(C(\overline{x})=0)$ is $O(tk^3|C|)=O(k^3|C|^2)$ where $t$ is the maximum fan-in.
    \end{definition}
    Same as the argument given in \cite{ST25}, we have the following propositions.
    \begin{proposition}
        \label{prop: C <-> cnf(C) <-> ecnf(C)}
        Let $C(\overline{x})$ be a circuit equation over $\mathbb{F}$. Then, $C(\overline{x})$ is unsatisfiable over $\mathbb{F}$ if and only if $\cnf(C(\overline{x})=0)$ is an unsatisfiable $\CNF$ if and only if $\ecnf(C(\overline{x}))= 0$ is an unsatisfiable set of equations over $\mathbb{F}$.
    \end{proposition}
    \begin{corollary}
        \label{cor: unsat ecnf(C) has VNP IPS refutation}
        If $\ecnf(C(\overline{x})=0)$ is unsatisfiable over $\mathbb{F}_q$ then it has an $\IPS^\alg$ refutation in $\VNP_{\mathbb{F}_q}$.
    \end{corollary}

    Now, we are ready to show the translation lemma for bit arithmetic encoding, similarly to \Cref{lemma: Translate semi-CNFs from circuit equations in Fixed Finite Fields} and \Cref{lemma: Translate circuit equations from semi-CNFs in Fixed Finite Fields}.

    \begin{lemma}
        \label{lemma: CNF ADD(x,y,z)}
        Let $\mathbb{F}_q$ be a finite field. Let $\overline{x} = x_{k-1}, \dots, x_0$, $\overline{y} = y_{k-1},\dots, y_0$ and $\overline{z} = z_{k-1},\dots, z_0$ be three $k$ length binary representations. Then,
        \begin{align*}
            \CNF\text{-}\add(\overline{x},\overline{y},\overline{z}) \sststile{\IPS^\alg}{O(2^k\cdot \poly(k)), O(1)} \VAL(\overline{x}) + \VAL(\overline{y}) = \VAL(\overline{z})
            \end{align*}
    \end{lemma}

    \begin{lemma}
        \label{lemma: Translate addition}
        Let $\mathbb{F}_q$ be a finite field. Let $\overline{x} = x_{k-1}, \dots, x_0$ and $\overline{y} = y_{k-1},\dots, y_0$ be two $k$ length binary representations. Let $\overline{\add} = \add_{k},\add_{k-1},\dots, \add_0$ be a $k+1$ length binary representation. Then,
        \begin{align*}
            \addition(\overline{x},\overline{y},\overline{\add}) \sststile{\IPS^\alg}{O(k), O(1)}
            \VAL(\overline{x}) +\VAL(\overline{y})= \VAL(\overline{\add_i(\overline{x},\overline{y})}) .
        \end{align*}
        In other words, from $\addition(\overline{x},\overline{y},\overline{\add})$, there is a $O(1)$-depth, $O(k)$-size $\IPS^\alg$ proof of $\VAL(\overline{x}) +\VAL(\overline{y})= \VAL(\overline{\add_i(\overline{x},\overline{y})}) $.
    \end{lemma}

    \begin{proof}[Proof of \Cref{lemma: Translate addition}]
        By \Cref{def: carry-add}, $\addition(\overline{x},\overline{y},\overline{\add})$ includes
        \begin{align*}
            \add_i(\overline{x},\overline{y}) &= \arit(x_i \oplus y_i \oplus \carry_i(\overline{x},\overline{y})), \qquad i=0,\cdots, k-1\\
            \add_k(\overline{x},\overline{y}) &= \carry_k(\overline{x},\overline{y}) .
        \end{align*}

        By simple substitution, it suffices to show
        \begin{align*}
            \addition(\overline{x},\overline{y},\overline{\add}) \sststile{\IPS^\alg}{\ast, O(1)} \sum_{i=0}^{k-1} (x_i+y_i) \cdot 2^i = \sum_{i=0}^{k-1} (x_i \oplus y_i \oplus \carry_i(\overline{x},\overline{y})) \cdot 2^i +\carry_k(\overline{x},\overline{y}) \cdot 2^k .
        \end{align*}
         
        For each $0 \leq i \leq k-1$, we aim to show that there is a constant-depth, constant-size $\IPS^\alg$ proof of the following equation: 
         \[
            2^i\cdot (x_i+y_i) = 2^i(x_i\oplus y_i \oplus \carry_i(\overline{x},\overline{y})) + 2^{i+1} \carry_{i+1}(\overline{x},\overline{y}) - 2^i \carry_i(\overline{x},\overline{y}).
        \]

        By substituting $\carry_{i+1}$ with $(x_{i}\lor y_{i})\land \carry_{i}(\overline{x},\overline{y}))$ according to \Cref{def: carry-add}, for each $0 \leq i \leq k-1$, the above equation becomes
        \[
            2^i\cdot (x_i+y_i) = 2^i(x_i\oplus y_i \oplus \carry_i(\overline{x},\overline{y})) + 2^{i+1} ((x_{i}\land y_{i}) \lor ((x_{i}\lor y_{i})\land \carry_{i}(\overline{x},\overline{y}))) - 2^i \carry_i(\overline{x},\overline{y}).
        \]
        Note that the above equation, which has only three Boolean variables, is true under any 0-1 assignment. Therefore, by \Cref{prop: 0-1 implication completeness of IPS}, there is a constant-depth, constant-size $\IPS^\alg$ proof of the above equation.

        Observe that all the equations above can be proved simultaneously. By summing over all these equations, there is a constant-depth, $O(k)$-size $\IPS^\alg$ proof of
        \begin{align*}
            \sum_{i=0}^{k-1} (x_i+y_i) \cdot 2^i = \sum_{i=0}^{k-1} (x_i \oplus y_i \oplus \carry_i(\overline{x},\overline{y})) \cdot 2^i +\carry_k(\overline{x},\overline{y}) \cdot 2^k .
        \end{align*}
        
        We can conclude
        \begin{align*}
            \addition(\overline{x},\overline{y},\overline{\add}) \sststile{\IPS^\alg}{O(k), O(1)}
            \VAL(\overline{x}) +\VAL(\overline{y})= \VAL(\overline{\add_i(\overline{x},\overline{y})}).
        \end{align*}
    \end{proof}

    \begin{lemma}
        \label{lemma: Translate modular}
        Let $\mathbb{F}_q$ be a finite field. Let $\overline{x} = x_{k-1}, \dots, x_0$ and $\overline{y} = y_{k-1},\dots, y_0$ be two $k$ length binary representations. Let $z$ be a Boolean variable. Then,
        \begin{align*}
            \modular^t(\overline{x},z, \overline{y}) \sststile{\IPS^\alg}{O(2^k\cdot \poly(k)),O(1)} \VAL(\overline{x}) + 2^t z = \VAL(\overline{y}) .
        \end{align*}
        In other words, from $ \modular^l(\overline{x},z, \overline{y})$, there is a $O(1)$-depth, $O(k)$-size $\IPS^\alg$ proof of $\VAL(\overline{x}) + 2^t z = \VAL(\overline{y}) $.
    \end{lemma}

    \begin{proof}[Proof of \Cref{lemma: Translate modular}]

        Recall \Cref{def: modular}, $\modular^t(\overline{x},z, \overline{y})$ contains
        
        \[
            \addition(0, \overbrace{x_{k-2},\dots,x_0}^{\mathclap{\text{the last $k-1$ bit of $\overline{x}$}}}, \underbrace{\overline{m}}_{\mathclap{\text{the $k$ length binary representation of $2^t z +2^{k-1} x_{k-1} \mod{q}$ }}}, \overline{u}),
        \]
        \[
            \addition^\prime(\overbrace{u_{k-1},\dots,u_0}^{\mathclap{\text{the last $k$ bit of $u$}}}, \underbrace{\overline{w}}_{\mathclap{\text{the $k$ length binary representation of $2^k u_k$}}},\overline{y}).
        \]

        By \Cref{lemma: Translate addition},
        \[
            \addition(0, x_{k-2},\dots, x_0, \overline{m}, \overline{u}) \sststile{\IPS^\alg}{O(k), O(1)} \sum_{i=0}^{k-2} 2^i x_i + \VAL(\overline{m}) = \VAL(\overline{u}).
        \]

        By \Cref{prop: 0-1 implication completeness of IPS}, there is a constant-depth, constant-size $\IPS^\alg$ proof of $\VAL(\overline{m})=2^t z+2^{k-1}x_{k-1}$ since there are only two Boolean variables in it after replacing each bit $m_i$ with $\arit(\cnf(f_i^t(z, x_{k-1})))$.

        Hence, by adding $\VAL(\overline{m})=2^t z+2^{k-1} x_{k-1}$ and $\sum_{i=0}^{k-2} 2^i x_i + \VAL(\overline{m}) = \VAL(\overline{u})$ together, we have $\VAL(\overline{x}) + 2^t z=\VAL(\overline{u})$.

        Now, we aim to show that 
        \begin{align*}
            \addition^\prime(u_{k-1},\cdots,u_0, \overline{w},\overline{y}) \sststile{\IPS^\alg}{O(2^k\cdot \poly(k)), O(1)} \VAL(\overline{u}) = \VAL(y)
        \end{align*}

        Same as the proof of \Cref{lemma: Translate addition}, for $0 \leq i \leq k-2$, there is a constant-depth, constant-size $\IPS^\alg$ proof of
        \[
            2^i\cdot (u_i+w_i) = 2^iy_i(u_{k-1},\dots,u_0,\overline{w}) + 2^{i+1} \carry_{i+1}(u_{k-1},\dots,u_0,\overline{w}) - 2^i \carry_i(u_{k-1},\dots, u_0,\overline{w}).
        \]

        For the $k$th bit, we aim to prove
        \[
            2^{k-1}(u_{k-1}+w_{k-1}) = 2^{k-1}y_{k-1}(u_{k-1},\dots,u_0,\overline{w})- 2^{k-1} \carry_{k-1}(u_{k-1},\dots,u_0,\overline{w}),
        \]
        which is
        \[
            2^{k-1}(u_{k-1}+w_{k-1}) = 2^{k-1}(u_{k-1}\oplus w_{k-1} \oplus \carry_{k-1}(u_{k-1},\dots,u_0,\overline{w}))- 2^{k-1} \carry_{k-1}(u_{k-1},\dots,u_0,\overline{w}).
        \]
        By replacing $u_{k-1}$, $w_{k-1}$ and $\carry_{k-1}(u_{k-1},\dots,u_0,\overline{w})$ by constant-depth, $\poly(k)$-size formulas $U(\overline{x},y)$, $W(\overline{x},y)$ and $CARRY(\overline{x},z)$ correspondingly, we get
        \[
            2^{k-1}(U+W) = 2^{k-1}(U\oplus W \oplus CARRY)- 2^{k-1} CARRY.
        \]
        The above equation holds over all Boolean assignments of $\overline{x},z$ according to our discussion above. By \Cref{prop: 0-1 implication completeness of IPS}, we get a depth-2, $O(2^k\cdot \poly(k))$-size $\IPS^\alg$ of the above equation. By adding the above equations, we have $\sum_{i=0}^{k-1} 2^i u_i +\VAL(\overline{w})=\VAL(\overline{y})$. 
        
        Since each bit in $\overline{w}$ can be replaced by the corresponding $\arit(\cnf(g_i(u_k)))$, there is a constant-depth, constant-size $\IPS^\alg$ proof of $\VAL(\overline{w})=2^k u_k$ since there is only one Boolean variable (i.e. $u_k$) in the formula after replacing. From $\sum_{i=0}^{k-1} 2^i u_i +\VAL(\overline{w})=\VAL(\overline{y})$ and $\VAL(\overline{w})=2^k u_k$, we get $\VAL(\overline{u})=\VAL(\overline{y})$.
        
        Now, we can conclude that
        \[
            \modular^t(\overline{x},z, \overline{y}) \sststile{\IPS^\alg}{O(2^k\cdot \poly(k)),O(1)} \VAL(\overline{x}) + 2^t z = \VAL(\overline{y}) .
        \]
    \end{proof}
    
    \begin{proof}[Proof of Lemma \ref{lemma: CNF ADD(x,y,z)}]
        Recall the \Cref{def: cnf addition}, $\CNF$-$\add(\overline{x},\overline{y},\overline{z})$ is as follows:
        \begin{itemize}
            \item Addition step: $\addition(\overline{x},\overline{y}, \overline{\add})$
            \item Modular step: $\modular^{k}(\add_{k-1}, \dots, \add_0, \add_k, \overline{\add^\prime})$
            \item Connection step: for each $0 \leq i \leq k-1$, we have $\add_i^\prime = z_i$.
        \end{itemize}
        \medskip

        We aim to show
        \begin{align*}
            \CNF\text{-}\add(\overline{x},\overline{y},\overline{z}) \sststile{\IPS^\alg}{O(k), O(1)} \VAL(\overline{x}) + \VAL(\overline{y}) = \VAL(\overline{z}).
        \end{align*}

        By \Cref{lemma: Translate addition},
        \[
            \addition(\overline{x},\overline{y}, \overline{\add}) \sststile{\IPS^\alg}{O(k), O(1)} \VAL(\overline{x}) +\VAL(\overline{y}) = \VAL(\overline{\add}).
        \]

        By \Cref{lemma: Translate modular},
        \[
            \modular^{k}(\add_{k-1}, \dots, \add_0, \add_k, \overline{\add^\prime}) \sststile{\IPS^\alg}{O(2^k\cdot \poly(k)), O(1)} \VAL(\overline{\add}) = \VAL(\overline{\add^\prime}).
        \]
        Note that $\VAL(\add_{k-1},\dots, \add_0) + 2^k \add_k = \VAL(\overline{\add})$.
         
        Since for each $0\leq i \leq k-1$, $\add^\prime_i = z_i$. It is easy to show that $\VAL(\overline{\add^\prime})=\VAL(\overline{z})$. We can conclude that
        \[
            \CNF\text{-}\add(\overline{x},\overline{y},\overline{z}) \sststile{\IPS^\alg}{O(2^k\cdot \poly(k)), O(1)} \VAL(\overline{x}) + \VAL(\overline{y}) = \VAL(\overline{z}).
        \]

    \end{proof}

     \begin{lemma}
        \label{lemma: CNF MULT(x,y,z)}
        Let $\mathbb{F}_q$ be a finite field. Let $\overline{x} = x_{k-1}, \dots, x_0$, $\overline{y} = y_{k-1},\dots, y_0$ and $\overline{z}=z_{k-1},\dots, z_0$ be three $k$ length binary representations. Then,
        \begin{align*}
            \CNF\text{-}\mult(\overline{x},\overline{y},\overline{z}) \sststile{\IPS^\alg}{O(2^k \cdot \poly(k)), O(1)} \VAL(\overline{x}) \cdot \VAL(\overline{y}) = \VAL(\overline{z})
            \end{align*}
    \end{lemma}

    \begin{lemma}
        \label{lemma: translate multiple}
        Let $\mathbb{F}_q$ be a finite field. Let $\overline{x} = x_{k-1},\dots, x_0$ and $\overline{y} = y_{k-1}, \dots, y_0$ be two $k$ length binary representations. Let $\overline{s_{k-1}}, \dots, \overline{s_0}$ be $k$ many binary representations with different length where $\overline{s_i}= s_{i, i+k-1}, \dots, s_{i,0}$ is a $i+k$ length binary representation for $0 \leq i \leq k-1$. Then,
        \begin{align*}
              \mult(\overline{x},\overline{y},\overline{s_0},\dots, \overline{s_{k-1}}) \sststile{\IPS^\alg}{O(k^2), O(1)} \VAL(\overline{x}) \cdot \VAL(\overline{y}) = \sum_{i=0}^{k-1} \VAL(\overline{s_i}).
        \end{align*}
        In other words, from $ \mult(\overline{x},\overline{y},\overline{s_{k-1}},\dots, \overline{s_0})$, there is a $O(1)$-depth, $O(k+i)$ size $\IPS^\alg$ proof of $\VAL(\overline{x}) \cdot \VAL(\overline{y}) = \sum_{i=0}^{k-1} \VAL(\overline{s_i})$.
        
    \end{lemma}

    \begin{proof}[Proof of \Cref{lemma: translate multiple}]
        Recall Definition \ref{def: prod, mult}, $\mult(\overline{x},\overline{y},\overline{s_0},\cdots, \overline{s_{k-1}})$ includes
        \begin{equation*}
            \mult_i(\overline{x},\overline{y},\overline{s_i}) \coloneqq \begin{cases}
                s_{i,j} = \arit(x_{j-i} \land y_i), &  i \leq j \leq k-1+i\\
                s_{i,j} = \arit(\bot), & 0 \leq j < i.
            \end{cases},\qquad  0 \leq i \leq k-1.
        \end{equation*}

        First, we aim to show that from each $\mult_i(\overline{x},\overline{y},\overline{s_i})$, there is a $O(k+i)$-size, $O(1)$-depth $\IPS^\alg$ proof of
        \[
            \VAL(\overline{x}) \cdot 2^i y_i = \VAL(\overline{s_i}).
        \]
    
        For each $j$ such that $i \leq j \leq k-1+i$, from $s_{i,j} = x_{j-i} \land y_i$ and Boolean axioms, there is a $O(1)$-size, $O(1)$-depth $\IPS^\alg$ proof of $2^{j-i} x_{j-i} \cdot 2^i y_i = 2^j s_{i,j}$. For each $j$ such that $0 \leq j  < i$, from $s_{i,j} =0$, there is a $O(1)$-size, $O(1)$-depth $\IPS^\alg$ proof of $2^j s_j = 0$. By summing up all these equations, there is a $O(k+i)$-size, $O(1)$-depth $\IPS^\alg$ proof of
        \[
            \VAL(\overline{x}) \cdot 2^i y_i = \VAL(\overline{s_i}).
        \]

        Again, by summing up all these equations for each $\mult_i$, there is a $O(k^2)$-size, $O(1)$-depth $\IPS^\alg$ proof of
        \[
        \sum_{j=0}^{k-1} 2^j x_j \times \sum_{i=0}^{k-1} 2^i y_i = \sum_{i=0}^{k-1} \sum_{w=0}^{k-1+i} 2^w s_{iw}.
    \]
    which is $\VAL(\overline{x})\cdot \VAL(\overline{y})=\sum_{i=0}^{k-1} \VAL(\overline{s_i})$.
    \end{proof}

    \begin{proof}[Proof of Lemma \ref{lemma: CNF MULT(x,y,z)}]
        Recall the \Cref{def: cnf multiplication}, $\CNF$-$\mult(\overline{x},\overline{y},\overline{z})$ is defined as follows:
    \begin{itemize}
        \item Multiplication step: $\mult(\overline{x},\overline{y},\overline{s_0},\cdots, \overline{s_{k-1}})$.
        \item Modular step: for each $\overline{s_i}$ such that $1 \leq i \leq k-1$, we have the following $i$ many modular formula:
        \begin{align*}
            &\modular^k(s_{i,k-1},\dots,s_{i,0},s_{i,k}, \overline{u_{i,k}}) \\
            &\modular^{k+1}(\overline{u_{i,k}}, s_{i,k+1}, \overline{u_{i,k+1}}) \\
            &\vdots\\
            &\modular^j(\overline{u_{i,j-1}}, s_{i,j}, \overline{u_{i,j}}), \qquad k+1 \leq j \leq i+k-1.
        \end{align*}
        All $\overline{u_{i,j}}$ is a $k$ length binary representation.
        \item Addition step: Now for each $\overline{s_i}$, we have a $k$ length binary representation $\overline{u_{i,i+k-1}}$ such that $\sum_{j=0}^{i+k-1}2^j s_{i,j}=\sum_{j=0}^{k-1} 2^j u_{i,i+k-1,j}$ in the finite field $\mathbb{F}_q$. Then, using the $\CNF$ encoding of bit arithmetic for addition in a finite field as we defined above in \Cref{def: cnf addition}, we have the following:
        \begin{align*}
            &\CNF\text{-}\add(\overline{s_0},\overline{u_{1,k}},\overline{v_1})\\
            &\CNF\text{-}\add(\overline{v_1},\overline{u_{2,k+1}},\overline{v_2})\\
            &\vdots\\
            &\CNF\text{-}\add(\overline{v_i},\overline{u_{i+1,k+i}},\overline{v_{i+1}}),\qquad 1 \leq i \leq k-2
        \end{align*}
        where each $\overline{v_i}$ is $k$ length binary representation.
        \item Connection step: for each $0 \leq j \leq k-1$, we have $v_{k-1,j} = z_j$.
    \end{itemize}

    By \Cref{lemma: translate multiple}, 
        \begin{align*}
              \mult(\overline{x},\overline{y},\overline{s_0},\dots, \overline{s_{k-1}}) \sststile{\IPS^\alg}{O(k^2), O(1)} \VAL(\overline{x}) \cdot \VAL(\overline{y}) = \sum_{i=0}^{k-1} \VAL(\overline{s_i}).
        \end{align*}

    By \Cref{lemma: Translate modular}, for each $\overline{s_i}$, there is a $O(2^k\cdot \poly(k))$-size, $O(1)$-depth $\IPS^\alg$ proof of
    \begin{align*}
        &\VAL(s_{i,k-1},\dots,s_{i,0})+2^k s_{i,k} =\VAL(\overline{u_{i,k}}) \\
        &\VAL(\overline{u_{i,k}}) + 2^{k+1}  s_{i,k+1} = \VAL(\overline{u_{i,k+1}})\\
        &\VAL(\overline{u_{i,j-1}}) + 2^j s_{i,j} = \VAL(\overline{u_{i,j}}), \qquad k+1 \leq j \leq i+k-1.
    \end{align*}
    By summing up the above equations, there is a $O(2^k \cdot \poly(k))$-size, $O(1)$-depth $\IPS^\alg$ proof of $\VAL(\overline{s_i})=\VAL(\overline{u_{i,i+k-1}})$.

    Therefore, there is a $O(2^k \cdot \poly(k))$-size, $O(1)$-depth $\IPS^\alg$ proof of $\VAL(\overline{s_i})=\VAL(\overline{u_{i,i+k-1}}), 1 \leq i \leq k-1$.

    By Lemma \ref{lemma: CNF ADD(x,y,z)}, given the formulas in the addition step:
    \begin{align*}
        &\CNF\text{-}\add(\overline{s_0},\overline{u_{1,k}},\overline{v_1})\\
        &\CNF\text{-}\add(\overline{v_1},\overline{u_{2,k+1}},\overline{v_2})\\
        &\vdots\\
        &\CNF\text{-}\add(\overline{v_i},\overline{u_{i+1,k+i}},\overline{v_{i+1}}),\qquad 1 \leq i \leq k-2,
    \end{align*}
    there is a $O(2^k \cdot \poly(k))$-size, $O(1)$-depth $\IPS^\alg$ proof of 
    \begin{align*}
        &\VAL(\overline{s_0}) +\VAL(\overline{u_{1,k}}) = \VAL(\overline{v_1})\\
        &\VAL(\overline{v_i})+\VAL(\overline{u_{i+1,k+i}}) = \VAL(\overline{v_{i+1}}),\qquad 1 \leq i \leq k-2.
    \end{align*}
    By adding them together, 
    \begin{equation*}
        \VAL(\overline{s_0}) + \sum_{i=0}^{k-1}\VAL(\overline{u_{i,i+k-1}})= \VAL(\overline{v_{k-1}}).
    \end{equation*}
    
    Since we have already proved that $\VAL(\overline{s_i})=\VAL(\overline{u_{i,i+k-1}})$ for each $i$ such that $1 \leq i \leq k-1$, by summing up them together with $\VAL(\overline{s_0})$, 

    \begin{align*}
        \sum_{i=0}^{k-1}\VAL(\overline{s_0}) &= \VAL(\overline{s_0}) + \sum_{i=1}^{k-1} \VAL(\overline{s_i}) \\
        & = \VAL(\overline{s_0}) + \sum_{i=0}^{k-1}\VAL(\overline{u_{i,i+k-1}}) \\
        &= \VAL(\overline{v_{k-1}}).
    \end{align*}
  
    By Connection step, for each $o \leq j \leq k-1$, $v_{k-1,j}=z_j$. There is a $O(k)$-size, $O(1)$-size $\IPS^\alg$ proof of $\VAL(\overline{v_{k-1}})=\VAL(\overline{z})$.

    Now, we can conclude that
    \begin{align*}
            \CNF\text{-}\mult(\overline{x},\overline{y},\overline{z}) \sststile{\IPS^\alg}{O(2^k \cdot \poly(k)), O(1)} \VAL(\overline{x}) \cdot \VAL(\overline{y}) = \VAL(\overline{z})
            \end{align*}
    \end{proof}

    \begin{lemma}[Translating from extended $\CNF$s to circuit equations]
        \label{lemma: Translating from extended CNFs to circuit equations}
        Let $\mathbb{F}_q$ be a finite field and let $k$ be $2^{k-1} < q < 2^k$, and let $C(\overline{x})$ be a circuit of depth $\Delta$ over $\overline{x}$ variables. Then, the following hold
        \begin{align*}
            \ecnf(C(\overline{x})=0) \sststile{\IPS^\alg}{O(2^k \cdot \poly(k)|C|),O(\Delta)} C(\overline{x})=0  
        \end{align*}
    \end{lemma}

    To prove \Cref{lemma: Translating from extended CNFs to circuit equations}, we will first show that from the Extended $\CNF$s of each node $g$ with $t$ many children in $C(\overline{x})$, there is a $O(1)$-depth, $O(2^k \cdot \poly(k)t)$-size $\IPS^\alg$ proof of the $\SLP$ formula of $g$. Then, we will show that from $\SLP(C(\overline{x})=0)$, there is a $O(\Delta)$-depth, $O(2^k \cdot \poly(k) |C|)$-size $\IPS^\alg$ proof of the circuit equation $C(\overline{x})=0$.
    
    \begin{proof}[Proof of \Cref{lemma: Translating from extended CNFs to circuit equations}]
        First, we aim to show that from the Extended $\CNF$s of each node $g$ with $t$ many children in $C(\overline{x})$, there is a $O(1)$-depth, $O(2^k \cdot \poly(k)t)$-size $\IPS^\alg$ proof of the $\SLP$ formula of $g$. We start with the $+$ node case.

        Suppose we have the following $\CNF$ formulas:
        \begin{align*}
            &\CNF\text{-}\add(\overline{u_1}, \overline{u_2}, \overline{v_1^g}) \\
            &\CNF\text{-}\add(\overline{u_{i+2}}, \overline{v_i^g}, \overline{v_{i+1}^g}), \quad 1 \leq i \leq t-3\\
            &\CNF\text{-}\add(\overline{u_t}, \overline{v_{t-2}^g}, \overline{g})
        \end{align*}
        which are the $\CNF$ encoding of the node $g$ expressing that $\sum_{i=0}^t u_i = g$.

        By Lemma \ref{lemma: CNF ADD(x,y,z)}, from 
        there is a $O(k t)$-size, $O(1)$-depth $\IPS^\alg$ proof of
        \begin{align*}
            &\VAL(\overline{u_1}) + \VAL(\overline{u_2}) = \VAL(\overline{v_1^g}) \\
            &\VAL(\overline{u_{i+2}} + \VAL(\overline{v_i^g}) = \VAL(\overline{v_{i+1}^g}), \quad 1 \leq i \leq t-3 \\
            &\VAL(\overline{u_t}) + \VAL(\overline{v_{t-2}^g}) = \VAL(\overline{g}).
        \end{align*}

        By \Cref{def: ecnf encoding of algebraic circuit (equations)}, for $\overline{u_i}$, $\overline{v_i^g}$ and $\overline{g}$, $\ecnf(C(\overline{x}) =0)$, $\ecnf(C(\overline{x})=0)$ includes $\VAL(\overline{u_i}) =u_i$, $\VAL(\overline{v_i^g}) = v_i^g$ and $\VAL(\overline{g})= g$, which are the binary value formulas for them. 

        Hence, from the Extended $\CNF$ encoding of the node $g$, there is a $O(k t)$-size, $O(1)$-depth $\IPS^\alg$ proof of the following $\SLP$s:
        \begin{align*}
            &u_1 + u_2 = v_1^g \\
            &u_{i+2} + v_i^g = v_{i+1}^g \quad 1 \leq i \leq t-3 \\
            &u_t + v_{t-2}^g = g
        \end{align*}

        By summing up all the $\SLP$ formulas, 
        \[
            u_1 + u_2 +\cdots + u_t = g.
        \]

        For the $\times$ nodes case, suppose we have the following $\CNF$ formulas:
        \begin{align*}
            &\CNF\text{-}\mult(\overline{u_1}, \overline{u_2}, \overline{v_1^g}) \\
            &\CNF\text{-}\mult(\overline{u_{i+2}}, \overline{v_i^g}, \overline{v_{i+1}^g}), \quad 1 \leq i \leq t-3\\
            &\CNF\text{-}\mult(\overline{u_t}, \overline{v_{t-2}^g}, \overline{g})
        \end{align*}
        which are the $\CNF$ encoding of the node $g$ expressing that $\prod_{i=0}^t u_i = g$.

        By Lemma \ref{lemma: CNF MULT(x,y,z)}, from the $\CNF$ encoding of the node $g$, there is a $O(2^k \cdot \poly(k) t)$-size, $O(1)$-depth $\IPS^\alg$ proof of
        \begin{align*}
            & \VAL(\overline{u_1}) \times \VAL(\overline{u_2}) - \VAL(\overline{v_1^g}) = 0 \\
            & \VAL(\overline{u_{i+2}}) \times \VAL(\overline{v_i^g}) - \VAL(\overline{v_{i+1}^g}) = 0, \quad 1 \leq i \leq t-3 \\
            & \VAL(\overline{u_t}) \times \VAL(\overline{v_{t-2}^g}) -  \VAL(\overline{g}) = 0.
        \end{align*}

        By \Cref{def: ecnf encoding of algebraic circuit (equations)}, $\ecnf(C(\overline{x})=0)$ includes the binary value formulas for those binary representations.

        Hence, from the Extended $\CNF$ encoding of the node $g$, there is a $O(k t)$-size, $O(1)$-depth $\IPS^\alg$ proof of the following $\SLP$s:
        \begin{align*}
            & u_1 \times u_2 - v_1^g = 0 \\
            & u_{i+2} \times v_i^g - v_{i+1}^g = 0, \quad 1 \leq i \leq t-3 \\
            & u_t \times v_{t-2}^g - g = 0
        \end{align*}

        Then, by the following depth-2 $O(t)$-size formula
        \[
            (u_1 \times u_2 - v_1^g ) \times \prod_{j=3}^t u_j + \cdots +(u_{i+2} \times v_i^g - v_{i+1}^g) \times \prod_{j=i+3}^t u_{j} + \cdots + (u_t \times v_{t-2}^g - g)=0,
        \]
        there is a $O(2^k \cdot \poly(k) t)$-size, $O(1)$-depth $\IPS^\alg$ proof of $u_1 \times \cdots \times u_t  - g =0$.
    
        For the output node, by the axioms $g_{out,i}=q_i$ and binary value formulas for $\overline{g_{out}}$ and $\overline{q}$, $g_{out}=\VAL(g_{out})=\VAL(\overline{q})=0$ can be easily derived.

        Therefore, from $\ecnf(C(\overline{x})=0)$, there is a $O(1)$-depth, $O(2^k \cdot \poly(k) |C|)$-size $\IPS^\alg$ proof of $\SLP$ formulas of $C(\overline{x})=0$. 
        
        Now, we aim to show that from these $\SLP$ formulas, there is $O(\Delta)$-depth, $O(2^k \cdot \poly(k) |C|)$-size $\IPS^\alg$ proof of the circuit equation $C(\overline{x} ) =0$.
        We prove this by induction that for each node $g$, there is a $O(\depth(g))$-depth, $O(2^k \cdot \poly(k) |C_g|)$-size $\IPS^\alg$ proof of a circuit equation $C_g(\overline{x} ) -g =0$:
        \begin{itemize}
            \item Base case: Since we have the binary value formulas for all the leaves, this is immediate.
            \item Addition case: suppose we have a $\SLP$ formula $u_1 + u_2 +\cdots + u_t - g=0$ where for each $u_i$, there is a $O(\depth(u_i))$-depth,  $O(2^k \cdot \poly(k) |C_{u_i}|)$-size $\IPS^\alg$ proof of a circuit equation $C_{u_i}(\overline{x} ) - u_i =0$. Then, by summing up all the circuit equations $C_{u_i}-u_i$ together with $u_1 + u_2 +\cdots + u_t = g$, there is a $O(\depth(g))$-depth, $O(2^k \cdot \poly(k) |C_g|)$-size $\IPS^\alg$ proof of a circuit equation $\sum_{i=1}^t C_{u_i} -g =0$ as this proof adds one depth and one node.
            \item Multiplication case: suppose we have a $\SLP$ formula $u_1 \times u_2 \times \cdots \times u_t - g=0$ where for each $u_i$, there is a $O(\depth(u_i))$-depth,  $O(2^k \cdot \poly(k) |C_{u_i}|)$-size $\IPS^\alg$ proof of a circuit equation $C_{u_i}(\overline{x} ) - u_i =0$. Then, computes $(\sum_{i=1}^t ((C_{u_i}-u_i) \times \prod_{j=1}^{i-1} C_{u_j} \prod_{l=i+1}^t u_l))+ (u_1 \times u_2 \times \cdots \times u_t - g) $, which adds two depths and two nodes. This gives a $O(\depth(g))$-depth, $O(2^k \cdot \poly(k) |C_g|)$-size $\IPS^\alg$ proof of a circuit equation $\prod_{i=0}^t C_{u_i} - g =0$.
            \item Output node: given $\SLP$ formula $g_{out} = 0$ and circuit equations $C_{g_{out}} - g_{out} =0$, it is easy to get $C_{g_{out}}=0$.
        \end{itemize}

        Now, we can conclude that 
        \begin{align*}
            \ecnf(C(\overline{x})=0) \sststile{\IPS^\alg}{O(2^k \cdot \poly(k) |C|),O(\Delta)} C(\overline{x})=0  
        \end{align*}
    \end{proof}

    \begin{lemma}
        \label{lemma: SLP+C = full SLP}
        Let $\mathbb{F}_q$ be a finite field, and let $C(\overline{x})$ be a circuit of depth $\Delta$ over $\overline{x}$ variables. Then, the following hold
        \[
            \SLP(C(\overline{x})) , C(\overline{x}) =0 \sststile{\IPS^\alg}{O(|C|), O(\Delta)} \SLP(C(\overline{x}) =0)
        \]
    \end{lemma}

    \begin{proof}[Proof of \Cref{lemma: SLP+C = full SLP}]
        We aim to show that given $\SLP(C(\overline{x}))$ and $C(\overline{x})$, there is a $O(\Delta)$-depth, $O(|C|)$-size $\IPS^\alg$ proof of $g_{out} =0$. By the proof of \Cref{lemma: SLP+C = full SLP}, given $\SLP(C(\overline{x}))$, there is a $O(|C|)$-size, $O(\Delta)$-depth $\IPS^\alg$ proof of $C_{g_{out}} - g_{out} =0$ where $C_{g_{out}}$ is exactly the same as $C(\overline{x})$. $g_{out}=0$ can be obtained by a simple subtraction.
    \end{proof}

    \begin{proposition}
        \label{prop: derive xi =0}
        Let $\overline{x} = x_{k-1},\dots, x_0$ be a $k$-length binary representation, where $x$ is an algebraic variable.
        \begin{gather*}
            \{x_i^2-x_i =0: \quad 0 \leq i \leq k-1\},\\
                x= \VAL(\overline{x}),\\
            x=0 \sststile{\IPS^\alg}{O(2^k \cdot \poly(k) ),O(\Delta)} \{x_i = q_i: \quad 0 \leq i \leq k-1\}
    \end{gather*}
    \end{proposition}

    \Cref{prop: derive xi =0} follows from \Cref{prop: 0-1 implication completeness of IPS}.

    \begin{lemma}[Translating from circuit equations and addition axioms to extended $\CNF$s]
        \label{lemma: translating from circuit equations and addition axioms to ecnf}
        Let $\F_q$ be a finite field and let $k$ be $2^{k-1} < q < 2^k$, and let $C(\overline{x})$ be a circuit of depth $\Delta$ over $\overline{x}$ variables. Then, the following holds
        \begin{gather*}
            \{x_i^2-x_i =0: x_i \text{ is a binary variable in }\ecnf(C(\overline{x})=0)\}, \\
            \text{All formulas in }\ecnf(C(\overline{x})=0) \text{ except the connection formula for the output node},\\
            \SLP(C(\overline{x})),\\
            C(\overline{x})=0 \sststile{\IPS^\alg}{O(2^k\cdot \poly(k)|C|),O(1)} \ecnf(C(\overline{x})=0) .
        \end{gather*}
    \end{lemma}

    \begin{proof}[Proof of \Cref{lemma: translating from circuit equations and addition axioms to ecnf}]

        Note that the only formulas required to derive is the connection formula for the output node:
        \[
            \{g_{out, i} = q_i: \quad  0\leq i \leq k-1\}.
        \]
        By \Cref{lemma: SLP+C = full SLP}, there is $O(\Delta)$-depth, $O(|C|)$-size $\IPS^\alg$ of $g_{out}=0$. Together with additional axioms and Boolean axioms we already have, we have all the axioms needed in \Cref{prop: derive xi =0}. By \Cref{prop: derive xi =0}, there is a $O(1)$-depth, $O(2^k\cdot \poly(k))$ $\IPS^\alg$ proof of $g_{out,i} =q_i$.
    \end{proof}

    \begin{lemma}[Translating between extended $\CNF$s and circuit equations] 
        \label{lemma: Translating between extended CNF formulas and circuit equations}
        Let $\F_q$ be a finite field and let $k$ be $2^{k-1} < q < 2^k$, and let $C(\overline{x})$ be a circuit of depth $\Delta$ over $\overline{x}$ variables. Then, the following holds
        \begin{align*}
            \ecnf(C(\overline{x})=0) \sststile{\IPS^\alg}{O(2^k \cdot \poly(k)|C|),O(\Delta)} C(\overline{x})=0  
        \end{align*} and
        \begin{gather*}
            \{x_i^2-x_i =0: x_i \text{ is a binary variable in }\ecnf(C(\overline{x})=0)\}, \\
            \text{All formulas in }\ecnf(C(\overline{x})=0) \text{ except the connection formula for the output node},\\
            \SLP(C(\overline{x})),\\
            C(\overline{x})=0 \sststile{\IPS^\alg}{O(2^k\cdot \poly(k)|C|),O(1)} \ecnf(C(\overline{x})=0) .
        \end{gather*}
    \end{lemma}

    For an instance of size $s$, we take the characteristics of the finite field to be a prime number between $n^3$ and $(n+1)^3$ that must exist for sufficiently large instances according to \cite{cheng2010explicit}. Let $k$ be the number of bits needed for the binary representation of field elements in $\mathbb{F}_q$. In other words, $2^{k-1} \leq q < 2^k $.
    We let $N = \sum_{j=0}^l \binom{n^2+j-1}{j}=2^{O(n^2+l)}$ be the number of different monomials over $n^2$ variables and degree at most $l$.
    \begin{itemize}
        \item $\VNP=\VAC^0(n,s,l,\Delta)$: circuit equations expressing that there is a constant-depth universal circuit for size $s$ and depth $\Delta$ circuits that compute the Permanent polynomial of dimension $n$, which means there are $n^2$ many variables, over degree $l$.
        \begin{itemize}
                \item Type: circuit equations;
                \item Number of variable: $K_{s, \Delta}$ which is $\poly(s, \Delta)$;
                \item Size: $O(2^{(n+l)l} \cdot \poly(s, \Delta) \cdot N)$.
            \end{itemize}
        \item $\varphi_{n,s,l,\Delta}^\cnf$: the $\CNF$ encoding of $\VNP=\VAC^0(n,s,l,\Delta)$ based on definition \Cref{def: cnf encoding of algebraic circuit equations}
            \begin{itemize}
                \item Type: $\CNF$ formulas;
                \item Number of variable: $O(2^{(n+l)l} \cdot \poly(s, \Delta) \cdot N)$;
                \item Size: $O(k \cdot 2^{(n+l)l} \cdot \poly(s, \Delta) \cdot N)$.
            \end{itemize}
        \item $\varphi_{n,s,l,\Delta}^\ecnf$: the Extended $\CNF$ encoding of $\VNP=\VAC^0(n,s,l,\Delta)$ based on definition \Cref{def: ecnf encoding of algebraic circuit (equations)}
            \begin{itemize}
                \item Type: Extended $\CNF$ formulas;
                \item Number of variable: $O(2^{(n+l)l} \cdot \poly(s, \Delta) \cdot N)$;
                \item Size: $O(k \cdot 2^{(n+l)l} \cdot \poly(s, \Delta) \cdot N)$.
            \end{itemize}
        \item $\IPS_{\refute}(t,\Delta,l,\overline{\mathcal{F}})$: circuit equations expressing that there exists a constant-depth universal circuit for size $t$ and depth $\Delta$ circuits that computes the $\IPS$ refutation of $\mathcal{F}$ over degree $l$.
            \begin{itemize}
                \item Type: circuit equations;
                \item Number of variable: $K_{t, \Delta}$ which is $\poly(t, \Delta)$;
                \item Size: $O(2^{(n+l)l} \cdot \poly(t, \Delta) \cdot |\mathcal{F}| \cdot N)$.
            \end{itemize}
        \item $\varphi_{n,s,l,\Delta}^\ast$: the Extended $\CNF$ encoding $\varphi_{n,s,l,\Delta}^\ecnf$ together with additional extension axioms for $\IPS_{\refute}(s,\Delta,l,\mathcal{F})$ includes:
        \begin{align*}
            &\{x_i^2-x_i =0: x_i \text{ is a binary variable in }\ecnf(\IPS_{\refute}(s,\Delta,l,\mathcal{F}))\}, \\
            &\text{All formulas in }\ecnf(\IPS_{\refute}(s,\Delta,l,\mathcal{F})) \text{ except the connection formula for the output node},\\
            &\SLP(\IPS_{\refute}(s,\Delta,l,\mathcal{F})) \text{ excepts the $\SLP$ formulas for output nodes}
        \end{align*}
            \begin{itemize}
                \item Type: Extended $\CNF$ formulas;
                \item Number of variable: $O(q \cdot 2^{(n+l)l} \cdot \poly(s, \Delta) \cdot N)$;
                \item Size: $O(k \cdot 2^{(n+l)l} \cdot \poly(s, \Delta) \cdot N)$.
            \end{itemize}
        \item $\Phi_{t,l,\Delta^\prime,n,s,\Delta}$: the $\CNF$ encoding of $\IPS_\refute(t,\Delta^\prime,l,\varphi_{n,s,l,\Delta}^\ast)$ expressing that $\IPS$ refutes $\varphi_{n,s,l,\Delta}^\ast$ in size $t$, depth $\Delta^\prime$ and degree $l$.
            \begin{itemize}
                \item Type: $\CNF$ formulas;
                \item Number of variable: $K_{t, \Delta}$ which is $\poly(t, \Delta)$;
                \item Size: $O(k \cdot 2^{(n+l)l} \cdot \poly(t, \Delta) \cdot |\mathcal{F}| \cdot N)$.
            \end{itemize}
    \end{itemize}
By replacing the use of \Cref{lemma: Translate semi-CNFs from circuit equations in Fixed Finite Fields} and \Cref{lemma: Translate circuit equations from semi-CNFs in Fixed Finite Fields}, which is the translation lemma for fixed finite fields, with the above \Cref{lemma: Translating between extended CNF formulas and circuit equations}, which is the translation lemma for polynomial-size finite fields, in the proof of \Cref{theorem: main theorem}, we can get the corollary below.
We fix $l: \mathbb{N} \to \mathbb{N}$ to be a (monotone) size function $l(n) = n^\epsilon$ for some constant $ \epsilon$.

\begin{corollary}[Main theorem in polynomial-size finite fields]\label{theorem: main theorem in Polynomial-size Finite Fields}
        The $\CNF$ family $\{\Phi_{t,l,\Delta^\prime,n,s,\Delta}\}_n$ does not have polynomial-size $\IPS$ refutations infinitely often over $\mathbb{F}_q$, for every prime $q$ such that $q < (|\varphi_{n,s,l,\Delta}^\ast|+1)^3$, in the following sense: 
        for every constant $\Delta$ there exists a constant $c_1$, a constant $\Delta^\prime$ such that for every sufficiently large constant $c_2$ and every constant $\Delta^{\prime\prime}$ and every constant $c_0$, for infinitely many $n, t(n), s(n) \in \mathbb{N}$, $t(n) > |\varphi_{n,s,l,\Delta}^\ast|^{c_1}$ and $n^{c_1} < s(n) < n^{c_2}$, $\Phi_{t,l,\Delta^\prime,n,s,\Delta}$ has no $\IPS$ refutation of size at most $|\Phi_{t,l,\Delta^\prime,n,s,\Delta}|^{c_0}$ and depth at most $\Delta^{\prime\prime}$.
    \end{corollary}

Corollary \ref{theorem: main theorem in Polynomial-size Finite Fields} is a formal statement of \Cref{thm:mainthm4} in the Introduction, which can be obtained from it by setting $l, s, t$ to appropriate polynomial functions of $n$, and by setting $d = \Delta, d' = \Delta', d'' = \Delta''$.

\let\etalchar\relax 

\small
    
\bibliography{reference}        
\bibliographystyle{alpha} 

\let\etalchar\relax 


\end{document}